\documentclass[11pt]{article}
\usepackage{amsmath, amsthm, amssymb, mathtools,comment}
\usepackage{color}
\usepackage[colorlinks=true,citecolor=blue,linkcolor=blue]{hyperref}
\hypersetup{
    colorlinks,
    linkcolor={red!50!black},
    citecolor={blue!50!black},
    urlcolor={blue!80!black}
}
\usepackage{macros}

\title{Discrepancy for Random Linear Codes}
\author{Dean Doron\thanks{Stein Faculty of Computer and Information Science, Ben-Gurion University. \texttt{deand@bgu.ac.il, leonovt@post.bgu.ac.il, mosheiff@bgu.ac.il}.} \and Tal Leonov\footnotemark[1] \and 
Jonathan Mosheiff\footnotemark[1] \and 
Henrique Navas\thanks{Instituto de Telecomunicações and Departamento de Matemática, Instituto Superior Técnico, Universidade de Lisboa. \texttt{h.campos.navas@tecnico.ulisboa.pt, jribeiro@tecnico.ulisboa.pt.}} \and
Nicolas Resch\thanks{Informatics Institute, University of Amsterdam. \texttt{n.a.resch@uva.nl}.} \and 
João Ribeiro\footnotemark[2]}
\date{}

\begin{document}
\maketitle

\begin{abstract}
    We show that random linear codes possess nearly optimal discrepancy-type
    properties in a broad range of settings. Our main results are two general
    discrepancy theorems: one controls all translates of a fixed test, and the
    other controls large families of Fourier-pseudorandom tests. As motivating applications of these two theorems, respectively, we show that:
\begin{itemize}

	\item Random linear codes behave essentially like unstructured random codes for list-decoding from errors \emph{above capacity}. More precisely, a random linear code $C\subseteq \F_q^n$ of rate $1 - \frac{1}{n}\log_q|B_\rho| + \eps$, where $|B_\rho|$ is the volume of a radius-$\rho$ Hamming ball in $\F_q^n$, satisfies
\begin{equation*}
	|C \cap B| = (1\pm o(1)) \frac{|C|\cdot |B|}{q^n}
\end{equation*}
simultaneously for all radius-$\rho$ Hamming balls $B$ in $\F_q^n$ with high probability.

This significantly strengthens the best prior results, which follow from work by Blinovsky (1987, 1990) and Levitin (1988) on the covering properties of random linear codes.

	\item Over prime fields, random linear codes behave essentially like unstructured random codes for zero-error list-recovery, and list-recovery from erasures, above capacity. More precisely, for a prime $q>2$ and input list size $2\leq \ell\leq q-1$, a random linear code $C\subseteq \F_q^n$ of rate $1-\log_q \ell+\eps$ will satisfy
\begin{equation*}
	|C \cap S| = (1\pm o(1)) \frac{|C|\cdot \ell^n}{q^n}
\end{equation*}
simultaneously for all combinatorial rectangles $S=S_1\times S_2\times\cdots\times S_n$, where $|S_i|=\ell$ for all $i$, with high probability. An analogous result also holds when we can bound $|S_i|$ only for some of the $i$'s.

In particular, we use this to show the abundance of locally leakage-resilient $n$-party linear \emph{ramp} secret sharing schemes with any linear reconstruction threshold and sublinear threshold gap $O(n/\log n)$ over fields $\F_q$ of polynomial size $q=\Theta(n^\gamma)$ for a constant $\gamma\in(0,1/5)$.
Prior work on the existence of leakage-resilient linear secret sharing was stuck at reconstruction thresholds above $n/2$ for both threshold and ramp schemes.
\end{itemize}

The translate-family result, and hence the list-decoding application, applies
over arbitrary finite fields, even when the field size grows with $n$.  The
list-recovery and leakage applications over prime fields hold under
moderate-growth conditions on $q$, for example $q\le n^{1/5-o(1)}$.  Our results are obtained through a careful
second-moment analysis of the evolution of intersection sizes as random
generators are added to $C$ one by one.
\end{abstract}

\newpage

\tableofcontents

\newpage

\addtocontents{toc}{\protect\setcounter{tocdepth}{2}}

\section{Introduction} \label{sec:intro}
A central theme in coding theory is to understand when random linear codes behave like unstructured uniformly random codes.
This question is especially natural for list-decoding and list-recovery, two important relaxations of unique decoding.

A code $C\subseteq\F_q^n$ is called \deffont{$(\rho,L)$-list-decodable} (from errors) if, whenever a channel corrupts a $\rho$-fraction of symbols, there are at most $L$ codewords that could have produced the received word.
Combinatorially, this means that for every $z\in\F_q^n$,
\[
    |C\cap B_\rho(z)|\leq L,
\]
where $B_\rho(z)$ denotes the Hamming ball of radius $\rho n$ centered at $z$.\footnote{$B_\rho(z)$ consists of all $y\in\F_q^n$ at Hamming distance at most $\rho n$ from $z$, i.e., all $y$ such that $|\{i:y_i\neq z_i\}|\leq \rho n$.}
Unique decoding is the special case $L=1$.

List-recovery is a further relaxation of list-decoding: instead of receiving a single word in $\F_q^n$, the receiver obtains input lists $(S_1,\dots,S_n)$, where each $S_i\subseteq \F_q$ has size at most $\ell$, and must return all codewords that largely agree with these lists.
We focus on the important special case of zero-error list-recovery.
A code $C$ is \deffont{$(\ell,L)$-list-recoverable} if for every such tuple $(S_1,\dots,S_n)$,
\[
    |C\cap (S_1\times S_2\times\cdots\times S_n)|\leq L.
\]
Thus, list-decoding asks for small intersections with Hamming balls, while zero-error list-recovery asks for small intersections with combinatorial rectangles of bounded side lengths.

With these notions in mind, a basic question is: for a code of a given rate $R$, how small can the list size $L$ be?
Both list-decoding and list-recovery exhibit phase transitions: below a certain rate, constant list size is possible; above it, every code must have exponentially large list size.
For list-decoding, fix a field $\F_q$ and an error radius $\rho$.
For every $\eps>0$, the following are known:
\begin{itemize}
    \item \textbf{Below capacity.} There exist codes of rate $1-\frac1n\log_q{\inabs{B_\rho}}-\eps$ that are $(\rho,L=O(1/\eps))$-list-decodable from errors.

    \item \textbf{Above capacity.} Every code $C\subseteq\F_q^n$ of rate $1-\frac1n\log_q{\inabs{B_\rho}}+\eps$ has some center $z\in\F_q^n$ for which
    \[
        |C\cap B_\rho(z)|\geq q^{\eps n}.
    \]
\end{itemize}
We therefore call $1-\frac1n\log_q{\inabs{B_\rho}}$ the \deffont{list-decoding capacity}.\footnote{Recall that $\frac1n\log_q{\inabs{B_\rho}} = h_q(\rho)-o(1)$, where
    \[
        h_q(x)=x\log_q\frac{q-1}{x}+(1-x)\log_q\frac{1}{1-x}
    \]
    is the $q$-ary entropy function. Hence, $1-h_q(\rho)$ is also often termed the list-decoding capacity.
}
The phase transition for zero-error list-recovery is analogous, with capacity $1-\log_q\ell$.

Unstructured uniformly random codes achieve the random-code benchmark on both sides of capacity.
A standard probabilistic-method argument shows that a uniformly random code of rate $1-\frac1n\log_q{\inabs{B_\rho}}-\eps$ is $(\rho,L=O(1/\eps))$-list-decodable with high probability.
Above capacity, a uniformly random code $C\subseteq\F_q^n$ of rate $1-\frac1n\log_q{\inabs{B_\rho}}+\eps$ satisfies, simultaneously for all $z\in\F_q^n$ with high probability,
\begin{equation}\label{eq:LD-above-cap}
    |C\cap B_\rho(z)|=(1\pm o(1))\frac{|C|\cdot |B_\rho|}{q^n}=(1\pm o(1))q^{\eps n}.
\end{equation}
Analogous statements hold for list-recovery.

A long line of work has asked whether random \emph{linear} codes match the performance of unstructured random codes for list-decoding and list-recovery.
Most prior work has focused on the \emph{below-capacity} setting, where random linear codes are known to match uniformly random codes in some parameter regimes~\cite{GHSZ02,GHK11,CGV13,W13,RW14,LW20,AGL24,DMRR25}, but not in others~\cite{Gur03,GLMRSW21,RY24,LMS25,CZ25,LS25}.
Our understanding of the \emph{above-capacity} setting is much poorer.
This leads to the question that motivates our work:
\begin{quote}
    \em
    How well do random linear codes approximate uniformly random codes in list-decoding and list-recovery at rates above capacity?
\end{quote}

Take list-decoding as a running example.
Above capacity, list-decoding for random linear codes is naturally a \emph{discrepancy} problem: the expected size for $|C\cap B_\rho(z)|$ is $|C|\cdot |B_\rho|/q^n\approx q^{\eps n}$, and the goal is to show that a random linear code has approximately this many codewords in \emph{every} Hamming ball.
This is qualitatively different from the below-capacity setting.
Below capacity, the benchmark $|C|\cdot |B_\rho|/q^n$ is exponentially small; moreover, for a linear code, the zero codeword always lies in $B_\rho(0)$, while for centers $z$ with $0\notin B_\rho(z)$ the fixed-center expectation is exponentially small.
Thus the below-capacity problem is not a uniform multiplicative-discrepancy question in the same sense.
As we discuss in detail later, this above-capacity viewpoint leads us to discrepancy-type questions that generalize list-decoding and list-recovery.

Toward proving~\cref{eq:LD-above-cap} for random linear codes, one could first ask for the much weaker but still nontrivial guarantee that a random linear code $C\subseteq\F_q^n$ of rate $1-\frac1n\log_q{\inabs{B_\rho}}+\eps$ satisfies
\[
    |C\cap B_\rho(z)|\geq 1
\]
for every $z\in\F_q^n$ with high probability.
This is precisely the statement that $C$ has covering radius at most $\rho n$.
Blinovsky~\cite{Bli87} and Levitin~\cite{Lev88} independently proved this over arbitrary fields almost forty years ago, building on earlier work of Goblick~\cite{Gob62}, Cohen~\cite{Coh83}, and Cohen and Frankl~\cite{CF85} which only established the \emph{existence} of a linear code with the property above.
Blinovsky's proof is nicely discussed in the book of Cohen, Honkala, Litsyn, and Lobstein on covering codes~\cite[Theorem 12.3.5]{CHLL97}, and it implies not only that $C$ is $\rho$-covering with high probability, but also the stronger lower bound
    \begin{equation}\label{eq:LB-Bli87}
        |C\cap B_{\rho}(z)|\geq \frac{1-o(1)}{n}\cdot \frac{|C|\cdot |B_{\rho}|}{q^n}
    \end{equation}
    for all $z\in\F_q^n$.
    This is smaller than the optimal lower bound by a factor of roughly $n$.
    To the best of our knowledge, these were essentially the strongest previously known uniform guarantees on $|C\cap B_\rho(z)|$ for random linear codes at rates above capacity.

Beyond being a natural question about the geometry of subspaces of $\F_q^n$, discrepancy questions for random linear codes above capacity arise in several other settings.
For example, zero-error list-recovery above capacity appears naturally in the study of \deffont{leakage-resilient} linear secret-sharing schemes.
Roughly speaking, such schemes should ensure that leaking, say, one bit of information about each share reveals almost nothing about the secret.
There is a natural way, due to Massey, to turn a linear code into a secret-sharing scheme.
For such a scheme, leakage resilience can be phrased as follows: for any functions $f_1,\dots,f_n:\F_q\to\{0,1\}$, the distribution of $(f_1(c_1),\dots,f_n(c_n))$, where $c=(c_1,\dots,c_n)$ is sampled uniformly from $C$, should be close in total variation distance to the distribution of $(f_1(x_1),\dots,f_n(x_n))$, where $x=(x_1,\dots,x_n)$ is sampled uniformly from $\F_q^n$.
For leakage bits $b_1,\dots,b_n\in\{0,1\}$, put $S_i=f_i^{-1}(b_i)$.
Then
\[
    \Pr[(f_1(c_1),\dots,f_n(c_n))=(b_1,\dots,b_n)]
    =\frac{|(S_1\times\cdots\times S_n)\cap C|}{|C|}.
\]
The numerator is the intersection of the code with a combinatorial rectangle, exactly the kind of quantity controlled by zero-error list-recovery.
In this application, assuming the sets $S_i$ are not too small, the relevant intersection size is expected to be exponential; the point is not to make the list small, but to show that its exponential size is the correct one.

\subsection{Our Contributions} \label{sec:our-contributions}

We show that random linear codes match the performance of uniformly random codes in list-decoding (over any field) and in zero-error list-recovery and list-recovery from erasures (over prime fields), above capacity.
In fact, we go further than that. 
We prove two main discrepancy theorems for random linear codes which show that random linear codes have strong above-capacity discrepancy properties for a broad class of combinatorial tests generalizing the above.
Our new theorems also imply new results about leakage-resilient secret sharing.

We state simplified versions of our general theorems here, specialized to the main motivating examples. The body of the paper gives the more general formulations.

We first introduce a useful normalization.
For a code $C\subseteq\F_q^n$ and a nonempty set $B\subseteq\F_q^n$, define the \deffont{relative deviation} of $C$ with respect to $B$ by
\begin{equation}\label{eq:setDeviation}
    \delta_{C,B}
    :=\frac{\Pr_{x\sim C}[x\in B]}{\Pr_{x\sim \F_q^n}[x\in B]}-1
    =\frac{q^n\cdot |C\cap B|}{|C|\cdot |B|}-1.
\end{equation}
Thus $|\delta_{C,B}|\leq \tau$ is equivalent to
\[
    (1-\tau)\frac{|C|\cdot |B|}{q^n}
    \leq |C\cap B| \leq
    (1+\tau)\frac{|C|\cdot |B|}{q^n}.
\]

\paragraph{List-Decoding and Beyond.}
Our first theorem applies to all translates of a single set.
As we discuss below, this is a more general setting than list-decoding.
\begin{theorem}[Informal; see \Cref{thm:singleFunction}] \label{thm:informal-single-set}
For any constant prime power $q$ and every $n \in \N$ the following holds.
    Let $B\subseteq\F_q^n$ be nonempty, and for $z\in\F_q^n$ define $B_z:=z+B$.
    Let
    \[
        \beta:=1-\frac{1}{n}\log_q |B|,
    \]
    and fix $\eta\in(0,1-\beta)$.
    Let $C$ be a random linear code of rate $R=\beta+\eta$.
    Then, with probability at least $1-q^{-\Omega(\eta n)}$,
    \[
        \forall z\in\F_q^n,\qquad |\delta_{C,B_z}|\leq q^{-\Omega(\eta n)}.
    \]
\end{theorem}

For list-decoding, take $B=B_\rho(0)$.
Then $\beta=1-\frac1n\log_q{\inabs{B_\rho}}$, and \Cref{thm:informal-single-set} says that a random linear code of rate $1-\frac1n\log_q{\inabs{B_\rho}}+\eta$ has, with high probability, the correct intersection size
\[
    |C\cap B_\rho(z)|=(1\pm q^{-\Omega(\eta n)})\frac{|C|\cdot |B_\rho(0)|}{q^n}
\]
for every center $z\in\F_q^n$.
In particular, this recovers and significantly strengthens Blinovsky's and Levitin's results~\cite{Bli87,Lev88}, which only guarantee that $|C\cap B_\rho(z)|\geq \frac{1-o(1)}{n}\cdot\frac{|C|\cdot |B_\rho(0)|}{q^n}$ for all $z\in\F_q^n$.

We additionally remark that defining a ``capacity'' with respect to an arbitrary set $B$ via the formula $1-\frac1n\log_q|B|$ is not new to our work; it appeared in an earlier work of Loeliger~\cite{L94}.

\paragraph{List-Recovery and Beyond.}
The translate theorem is well suited to list-decoding because all Hamming balls of a fixed radius are translates of one another.
For list-recovery, the relevant rectangles are not all translates of a single set.
Our second theorem handles large families of sets, provided the sets satisfy a Fourier pseudorandomness condition.
We say that a set $B\subseteq\F_q^n$ is \deffont{$\alpha$-Fourier-concentrated} if
\[
    \forall y\in\F_q^n,\qquad
    \left|\Eover{Z\sim B}{\chi_y(Z)}\right|\leq \alpha^{\mathrm{wt}(y)},
\]
where $Z$ is uniform over $B$ and $\chi_y$ is the Fourier character defined in \cref{sec:Notations}.
Informally, this condition says that $B$ has little correlation with any nontrivial linear test, with the correlation decaying exponentially in the Hamming weight of the test.

\begin{theorem}[Informal; see \Cref{thm:mainGeneral}] \label{thm:informal-general}
    For any constant prime power $q$ and every $n \in \N$ the following holds.
    Let $\cB$ be a family of nonempty $\alpha$-Fourier-concentrated subsets of $\F_q^n$, and let
    \[
        \beta:=1-\frac1n\min_{B\in\cB}\log_q |B|.
    \]
    Fix $\eta\in(0,1-\beta)$, and let $C$ be a random linear code of rate $\beta+\eta$.
    Then, with probability at least
    \[
        1-q^{-\Omega_\alpha(\eta n)}-|\cB|\,q^{-\Omega_\alpha(n^2)},
    \]
    we have
    \[
        \forall B\in\cB,\qquad |\delta_{C,B}|\leq q^{-\Omega(\eta n)}.
    \]
\end{theorem}

As an example, consider zero-error list-recovery with input list size $\ell$.
For the upper-bound guarantee it suffices to control rectangles with side length exactly $\ell$, since smaller input lists can be padded.
Thus consider rectangles $S_1\times\cdots\times S_n$ with $S_i\in\binom{\F_q}{\ell}$.
Over prime fields, standard Fourier estimates imply that such rectangles are $\alpha$-Fourier-concentrated for some $\alpha=\alpha(q,\ell)<1$, assuming $2\leq\ell\leq q$.
Therefore, in any regime where the family size $\binom{q}{\ell}^n$ is $q^{o(n^2)}$, \Cref{thm:informal-general} implies that a random linear code $C$ of rate $1-\log_q\ell+\eta$ satisfies, with high probability,
\[
    \forall S_1,\dots,S_n\in\binom{\F_q}{\ell},\qquad
    \left(1-q^{-\Omega_{q,\ell}(\eta n)}\right)\frac{|C|\cdot \ell^n}{q^n}
    \leq |C\cap(S_1\times\cdots\times S_n)|
    \leq
    \left(1+q^{-\Omega_{q,\ell}(\eta n)}\right)\frac{|C|\cdot \ell^n}{q^n}.
\]
In particular, the upper bound gives the expected above-capacity zero-error list-recovery guarantee for random linear codes in these regimes. See \Cref{thm:list-rec}.

Extending zero-error list recovery, our result holds for list recovery from \emph{erasures} as well. A code $C$ is \deffont{$(\zeta,\ell,L)$-list-recoverable} if for every tuple $(S_1,\dots,S_n)$ in which $(1-\zeta)n$ of the sets $S_i$ satisfy $|S_i| = \ell$, it holds that
\[
    |C\cap (S_1\times S_2\times\cdots\times S_n)|\leq L.
\]
In other words, list recovery from erasures allows for a $\zeta$-fraction of the rectangle's sides to be ``trivial''. The family of such rectangles is still small enough, and $\alpha$-Fourier concentrated, so \Cref{thm:informal-general} will imply that with high probability, a random linear code $C$ of rate $(1-\zeta)(1-\log_{q}\ell)+\eta$ will intersect every such $S$ at $(1 \pm o(1))|C|(\ell/q)^{(1-\zeta)n}$ points. See \Cref{thm:list-rec-erasures}.

\paragraph{Application to Leakage-Resilient Linear Ramp Secret Sharing.}

Our \cref{thm:informal-general} yields new results about the local leakage-resilience of linear secret sharing schemes.
We present here a brief overview of this topic and leave a more careful discussion to \cref{sec:leakage-resilience}.

Informally, a threshold secret sharing scheme encodes a secret into a tuple of $n$ shares such that any appropriately large subset of shares determines the secret, while any appropriately small subset of shares reveals no information about the secret.
For applications we often require that the secret sharing scheme be \emph{linear}.
Roughly speaking, this means that the secret and the shares are all elements of $\F_q$, and any linear combination of shares of two secrets yields a valid sharing of the linear combination of the two secrets.
When $q$ is much smaller than the number of parties $n$ (i.e., when we aim for small shares), there must be a gap between the \emph{reconstruction threshold} $\trec$ (all subsets of shares of size at least $\trec$ can reconstruct the secret) and the \emph{privacy threshold} $\tpriv$ (all subsets of shares of size at most $\tpriv$ give no information about the secret) unless both $\trec$ and $\tpriv$ are unusually large~\cite{CCX13}.
Such schemes with a gap between thresholds are called \emph{ramp secret sharing schemes}, and they are the focus of our discussion.

Motivated by side-channel attacks, recently there has been significant interest in understanding the resilience of linear secret sharing schemes against attacks that leak a small amount of side information from \emph{every} share.
More precisely, in the most popular model of \emph{bounded local leakage}~\cite{BDIR21}, an adversary chooses local leakage functions $g_i:\F_q\to\bits$ for $i\in[n]$, and, letting $S=(S_1,\dots,S_n)$ denote the $n$ shares of some secret, learns the $n$-bit leakage
\begin{equation*}
    g(S)=(g_1(S_1),g_2(S_2),\dots,g_n(S_n)).
\end{equation*}
Then, we want to establish that linear secret sharing schemes are information-theoretically resilient against this type of leakage, meaning that $g(S)$ and $g(S')$ should be very close in statistical distance for any sharings $S$ and $S'$ of two distinct secrets $s$ and $s'$, respectively.

Establishing leakage-resilience of secret sharing schemes is easier for larger reconstruction threshold.
Therefore, a central research direction has been to establish the existence of leakage-resilient linear secret sharing schemes with small reconstruction threshold (and, in the context of ramp secret sharing, with a small non-trivial gap between the privacy and reconstruction thresholds).
Prior results have only managed to show existence of such leakage-resilient schemes for reconstruction thresholds above $n/2$~\cite{MPSW21,TX22,KK23,Kas24,Ngu24}.
Subsequently to our work, Sun and Wootters~\cite{SW26b} and Nguyen~\cite{Ngu26} improved on the results above by showing that Shamir secret sharing is locally leakage-resilient for any reconstruction threshold $t\geq(1/2-\delta)n$ with $\delta>0$ a small constant.

We use \cref{thm:informal-general} to obtain an improved existential result about locally leakage-resilient linear ramp secret sharing schemes.
To avoid complicating the exposition, we state here a simplified version of our result and leave details for \cref{sec:leakage-resilience}.
Informally, it states that for arbitrary $R\in(0,1)$ there exist locally leakage-resilient $n$-party linear ramp secret sharing schemes with reconstruction threshold $\trec\approx Rn$ and threshold gap $\trec-\tpriv\leq 3n/\log q$ over appropriately large fields $\F_q$ (more precisely, for $q$ larger than some constant depending only on $R$ and smaller than about $n^{1/5}$).

\begin{theorem}[Informal version of \cref{thm:rampss-short}. See also \cref{thm:rampss-full,cor:rampss-simpler}]\label{thm:informal-lrss}
    Fix an arbitrary constant $R\in(0,1)$.
    Given $n\in\N$, let $q$ be a prime satisfying
    \begin{equation*}
        2^{800/R} \leq q \leq n^{0.19}.
    \end{equation*}
    Then, for all sufficiently large $n$ (depending only on $R$) there is a locally leakage-resilient $n$-party ramp secret sharing scheme over $\F_q$ with thresholds
    \begin{equation*}
       Rn - \frac{1.01n}{\log q} \leq \tpriv < \trec \leq Rn + \frac{1.01n}{\log q}
    \end{equation*}
    and leakage-resilience error $\eps\leq 4q^{- \frac{R^2 n^{3/5}}{3200}}$.
\end{theorem}

\paragraph{A Glimpse at our Techniques.} 
While we will provide a much more detailed overview of our techniques below in \Cref{sec:proof-overview}, we briefly give a birds-eye view of our approach. To construct a random linear code, we imagine constructing a random linear code by iteratively sampling uniformly random basis vectors. That is, define $C_0 := \{0\}$, and inductively put $C_i := C_{i-1}+\spn\{u_i\}$, where $u_i$ is sampled uniformly at random from $\F_q^n$: the final code is $C = C_k$, where $k=Rn$ is the target dimension.

The main task, then, is to control the growth of the intersections $|C_i \cap B|$ for increasing $i=0,1,\dots,k$, for all sets $B$ under consideration. A first, fairly obvious fact is that the growth rate is ``correct'' in expectation -- after all, we are trying to argue that the $|C_k \cap B|$ is always roughly its expectation! The main challenge is to get a tail bound: we need to argue that it is \emph{very unlikely} that any $|C_i \cap B|$ deviates far from its expectation. 

The main tool at our disposal is the fact that the events ``$x \in C$'' (over distinct $x$) are ``nearly'' \emph{pairwise} independent (there is a slight exception from colinear pairs, but this is negligible): thus, a second-moment method is natural. The challenge is that applying a second moment argument to \emph{each} $B$ under consideration and then applying a union bound is doomed to fail (there are more sets to union over than the reciprocal of the failure probability). Our main accomplishment is to identify a \emph{single common} event that implies \emph{all} the growth rates are adequately controlled; additionally, this event can be proved likely to occur via a basic second-moment argument. 

\subsection{Additional Context and Related Works} \label{sec:related-works}

Having given a taste of our results and our techniques, we now discuss (additional) related works to help situate our contribution.

\paragraph{Covering Radius.} 
In studying codes above capacity the parameter that has seen the most attention is its covering radius $\rho$.
As discussed above around \cref{eq:LB-Bli87}, independent work by Blinovsky~\cite{Bli87} and Levitin~\cite{Lev88} demonstrates that codes of rate $1-\frac 1n\log_q\inabs{B_\rho}+\eta$ are very likely to satisfy $|C\cap B_{\rho}(z)|\geq \frac{1-o(1)}{n}\cdot \frac{|C|\cdot |B_{\rho}|}{q^n}$ for all $z \in \F_q^n$.
Soon after, Blinovsky~\cite{Bli90a} improved the bound on the probability that a random linear code is $\rho$-covering.

The arguments behind these bounds (see~\cite[Theorem 12.3.5]{CHLL97}, which presents the proof from~\cite{Bli87}) indeed consider revealing the random linear codes basis vectors one at a time, analogously to us, and track the step-by-step evolution of the first moment of an appropriate potential.

However, a notable simplification in these works compared to our case is that establishing lower bounds on $|C \cap B_\rho(z)|$ is a \emph{monotone} process, in the sense that adding additional basis vectors to $C$ will not ``undo'' the process and make the lower bound worse.
For example, the argument in~\cite{Bli87} and~\cite[Theorem 12.3.5]{CHLL97} exploits this fact essentially by arguing that after a certain number of steps ``almost all'' points are covered, and then adding a handful of more basis vectors to handle the uncovered vectors (in particular, the union bound here can be much smaller). In contrast to this, since we establish a two-sided bound, we are not able to argue that certain vectors $z$ are ``satisfied,'' as every step could potentially cause an issue. In particular, while covering codes can be viewed as codes ``above capacity,'' the defining property is not a real ``discrepancy'' problem as we focus on addressing here.

\paragraph{Covering and Local Properties.}
Covering also suggests a natural direction for extending the systematic study of random linear codes. A line of work originating in~\cite{MRRSW24} develops a unified theory of threshold rates for \deffont{local properties}. Roughly, a property is local if it is governed by the presence or absence of prescribed configurations involving a bounded number of codewords, independent of the block length. For example, for fixed $L$, $(\rho,L)$-list-decodability is local: its failure is witnessed by $L+1$ distinct codewords contained in a single radius-$\rho n$ Hamming ball. This theory identifies general principles governing the threshold rates of broad classes of local properties simultaneously.

Covering lies outside this local framework, but fits naturally into an extension based on a universal--existential quantifier structure. Indeed, a code $C$ has covering radius at most $\rho n$ precisely when
\[
    \forall z\in\F_q^n\quad
    \exists c\in C\quad
    c\in B_\rho(z).
\]
Informally, we may call conditions of this form \emph{$\Pi_2$ properties}: ``for every test, there exists a codeword (or a bounded tuple of codewords) passing the test.'' Developing a unified theory of such properties is a natural next step beyond the theory of local properties. The present work grew in part out of an investigation of $\Pi_2$ properties of random linear codes.

\paragraph{Linear Hashing.} Discrepancy questions for random linear codes above capacity do arise in the study of linear hashing~\cite{ADMPT99,DD22,JKZ25,B26}. Here one fixes a subset $S\subseteq\F_q^n$ and, identifying $C$ with the kernel of a random linear map, asks that every translate $y+S$ have about the expected intersection size with $C$. The relevant regime is again just above the natural threshold: $C$ is chosen so that $|C|\cdot |S|=q^{n+O(1)}$. The main difference here is that we look at codes with constant \emph{relative} gap above capacity (i.e., $|C| \cdot |S| = q^{(n+\eps)n}$): it is an interesting open problem to determine if our techniques can additionally work with this narrower gap. We remark that most of these works achieve bounds only over the \emph{binary} field~\cite{ADMPT99,JKZ25,B26}; only the work of Dhar and Dvir~\cite{DD22} achieves bounds over arbitrary fields (as we do). 

\paragraph{Potential-Based Arguments.}
Prior works have used potential-based arguments to study the list-decodability of random linear codes below capacity~\cite{GHSZ02,LW20}, linear hashing~\cite{ADMPT99,JKZ25,B26}, and the covering radius of random linear codes~\cite{Gob62,Coh83,CF85,DP86,Bli87,Lev88,Bli90a,Bli90b}.\footnote{The use of the potential method in coding theory seems to go back at least to 1962~\cite[Theorem 4.2]{Gob62}, in the context of characterizing the optimal covering radius of linear codes. We refer to the historical note of Ribeiro~\cite{Rib26} for a more detailed discussion on the early development of this method.} 
Briefly, all these arguments consider a potential of the form $\Phi(C)= \Eover{X \in \F_2^n}{b^{|C \cap (X+S)|}}$, where $S$ is the set under consideration ($B_\rho(0)$ for list-decoding, or an arbitrary set for linear hashing), and $b\geq 0$ is an appropriate base for the exponential (about $2^{n/L}$ for list-decoding, $b\leq 1$ for linear hashing, and $b=0$ for covering problems with the convention that $0^0=1$, as nicely observed in Kumar's PhD thesis~\cite[Section 3.1.2]{Kum26}). As in our work, the arguments in these papers consider constructing the codes by including a basis vector one at a time, and observing how this potential evolves. 
The crucial observation in the papers above on list-decoding and linear hashing is that for any fixed code $C$ we have
\begin{equation}\label{eq:square-potential}
    \Eover{U \sim \F_2^n}{\Phi(C+\langle U\rangle)} \leq \Phi(C)^2.
\end{equation}
That is, on average over the newly included basis vector, the potential is squared. 
By iterating this, ignoring some technicalities, one obtains a bound of the form $\Phi_k \leq \Phi(\{0\})^{2^k}$, which is enough to show that \emph{some} linear code exists with the desired property. 
It is also possible to refine this argument to obtain a statement that holds with high probability over the sampling of the random linear code.
Such refinements were obtained for the covering problem in the late 80s~\cite{Bli87,Lev88,Bli90a,Bli90b}, and much later for list-decoding and linear hashing~\cite{LW20,JKZ25,B26}.

Establishing \cref{eq:square-potential} crucially uses the assumption that the field is binary. Over $\F_2$, every line consists of just two points; additionally, note that adding a random basis vector essentially amounts to adding a line at every codeword in a random direction. For random linear codes (over \emph{any} field), the events ``$x \in C$'' (over distinct non colinear vectors $x$) are always \emph{pairwise} independent; however, they fail to be even 3-wise independent (for example, conditioned on $x \in C$ and $y \in C$, it holds that $x+y \in C$ with probability 1). In particular, one obtains correlations in computing $\Eover{U \sim \F_2^n}{\Phi(C+\langle U\rangle)}$ that appear hard to handle. Thus, this potential-based argument, while extremely elegant, is currently limited to $\F_2$ for list-decoding and linear hashing. In fact, a motivation for our work was to develop techniques allowing for an analysis over arbitrary fields.
To the best of our knowledge, the only works that managed to apply the potential method beyond $\F_2$ focus on the covering radius of random linear codes.
Cohen and Frankl~\cite{CF85} adapted the potential method of Goblick~\cite{Gob62} and Cohen~\cite{Coh83} to pinpoint the optimal covering radius of linear codes over arbitrary fields, and Blinovsky~\cite{Bli87} and Leviti~\cite{Lev88} later established a high probability analogue.
As we discussed above, the covering property satisfies a notion of monotonicity that simplifies the analysis of the step-by-step argument compared to list-decoding and the two-sided bounds we consider here.

Finally, we remark that by combining the potential function-based argument of~\cite{GHSZ02,LW20} and the recent tail bound of Jaber, Kumar, and Zuckerman~\cite{JKZ25} one can prove a much weaker version of \Cref{thm:informal-single-set} \emph{for binary codes only}, where the $2^{-\Omega(\eta n)}$ term is replaced by a constant larger than $2$ (in particular, it only guarantees list size a constant factor larger than the expectation) and the probability bound is also weaker.
Generalizing this type of argument to larger fields has eluded researchers for the last 25 years. 
For future work, it would be interesting to determine to what extent our techniques can be applied to random linear codes \emph{below} capacity. 

\paragraph{Optimal Polynomial Intersection (OPI).} This problem has recently attracted attention as a possible route to demonstrating quantum advantage~\cite{JSW25DQI,CT25,C25,KSGZYMBJ25,R26,SW26}.
In one version of the problem,\footnote{We state here the version for general linear codes. The original formulation focused on Reed--Solomon codes, hence the word ``polynomial,'' but the same question is natural for arbitrary linear codes.} one is given input lists $S_1,\dots,S_n\subseteq\F_q$ and a description of a linear code $C$, and must output a codeword $c\in C$ maximizing $|\{i\in[n]:c_i\in S_i\}|$.
Equivalently, one seeks a codeword in an output list for a list-recovery problem (with errors) at the smallest possible decoding radius.
The parameter regimes of interest are above capacity: the corresponding list will contain exponentially many codewords. However, the question here is more algorithmic -- \emph{find} some codeword from this (exponentially large) set of possibilities -- rather than combinatorial/discrepancy-theoretic, like the questions we study. 

\paragraph{Random Linear Codes Below Capacity.} Finally, while our main goal is to study random linear codes at rates \emph{above} capacity, we remark that there has been a long line of work studying random linear codes \emph{below} capacity (we already mentioned some such works above when discussing potential function-based arguments). 
Firstly, we consider the task of list-decoding: at rates $1-\frac 1n\log_q\inabs{B_\rho}-\eps$, one aims to establish that random linear codes are with high probability $(\rho,O(1/\eps))$-list-decodable, which matches the performance of uniformly random codes. It is worth mentioning a classical argument of Zyablov and Pinsker~\cite{ZP81} demonstrates they are $(\rho,q^{O(1/\eps)})$-list-decodable; so the challenge is to show an exponentially smaller list-size is still sufficient. The aforementioned work~\cite{GHSZ02} derived the $O(1/\eps)$ list-size, but only over the binary alphabet; later works managed to (largely) extend this to arbitrary, constant-sized fields~\cite{GHK11,CGV13,W13}, although often with some parameter restrictions. 

List-recovery has also seen more attention in recent years, both with errors\footnote{Briefly, a code $C$ is $(\rho,\ell,L)$-list-recoverable \emph{from errors} if for any combinatorial rectangle $S_1 \times \cdots \times S_n$ with each $|S_i| \leq \ell$, we have $|\{c \in C:|\{i \in [n]:c_i \notin S_i\}| \leq \rho n\}| \leq L$. That is, we allow a $\rho$ fraction of the coordinates $c_i$ to not lie in the input list $S_i$.} and without. For list-recovery in this setting, there has only been one positive result establishing that random linear codes can be list-recovered with list-size $O_{q,\ell}(1/\eps)$~\cite{DMRR25} (while to match plain random codes, one would hope for list-size $O(\ell/\eps)$). However, there have been a number of \emph{negative} results~\cite{GLMRSW21,RY24}, in some cases showing that random linear codes provably perform \emph{worse} than plain random codes. Notably, the restrictions for zero-error list-recovery arise over fields of small characteristic, but disappear over prime fields: similarly, to establish zero-error list-recovery bounds (and hence our application to leakage-resilient secret-sharing), we work over prime fields. 

When $q$ is large (at least $\exp{\Omega(1/\eps)}$), many recent works~\cite{SI23,BGM23,GZ23,AGL24,LMS25} have established that random linear codes (and, even more interestingly, random \emph{Reed-Solomon} codes;\footnote{Reed-Solomon codes are defined by evaluating low-degree univariate polynomials at $n$ distinct field elements; to generate a random Reed-Solomon code, one samples these evaluation points randomly.} in this case $q\geq \Omega(n) \cdot \exp{\Omega(1/\eps)}$) of rate $R$ are $(1-R-\eps,O(1/\eps))$-list-decodable. More generally, they can be shown to achieve the \emph{generalized Singleton bound}, i.e., are $(\frac{L}{L+1}(1-R-\eps),L)$-list-decodable for any $L \geq 1$. For list-recovery over large alphabets ($q \geq \ell^{\Omega(1/\eps)}$), recent works have shown matching lower~\cite{LS25} and upper~\cite{RW14,RW18,BCDZ26} bounds on the list-size required at decoding radius $1-R-\eps$: $L = \ell^{O(1/\eps)}$. This is provably \emph{worse} than plain random codes (where $L=O(\ell/\eps)$ is sufficient); furthermore, we emphasize that this limitation is inherent over \emph{every} finite field, not just small characteristic fields. 

\subsection{AI Disclosure}
We used ChatGPT 5.6 Sol to assist with editing, proofreading, and verification. All ideas and technical analysis are due to the authors, who take full responsibility for the correctness of all content.

\section{Main Results}\label{sec:main-results}

This section presents the formal versions of the paper's main discrepancy results.
We begin by fixing notation and conventions, and then state two theorems on the discrepancy of random linear codes: the first controls all translations of a single test function, while the second applies to large families of Fourier-concentrated test functions. We then derive the corresponding applications to list-decoding and list-recovery.

The application to leakage-resilient secret-sharing schemes is treated separately in \cref{sec:leakage-resilience}.

\subsection{Notation and Conventions}\label{sec:Notations}

\subsubsection*{Random Linear Codes}
Fix a prime power $q$ and integers $k,n\in\N$ with $k\le n$.
In this work, a \deffont{random linear code (RLC)} of designed dimension $k$ and length $n$ over $\F_q$ is the image of a uniformly random matrix $G\in\F_q^{n\times k}$.
Equivalently, its designed rate is $k/n$.
We refer to this as the \deffont{random generating matrix model}.

Another common model samples $C$ uniformly from the set of all $k$-dimensional subspaces of $\F_q^n$.
The two models agree after conditioning on the event that $G$ has full rank, since the image of a full-rank random matrix is a uniformly random $k$-dimensional subspace.
Moreover, the latter event holds with probability at least $1-\frac{q^{-(n-k+1)}}{1-1/q}$.
Thus, whenever $n-k$ grows linearly with $n$, passing between the two models changes the stated probabilities only by a negligible term. This is the regime of the constant-rate applications below.

\subsubsection*{Fourier Transform}
For functions $f:\F_q^n\to\mathbb{C}$, we use the expectation measure on the primal space and the counting measure on the Fourier side.
The Fourier transform of $f$ at $y\in\F_q^n$ is
\[
    \widehat f(y)=q^{-n}\sum_{x\in\F_q^n} f(x)\chi_y(x),
\]
where $\chi_y$ is the additive character of $\F_q^n$ associated with $y$.
Specifically, if $\F_q$ has characteristic $p$ and $q=p^m$, then
\[
    \chi_y(x)=\exp{\frac{2\pi i}{p}\tr(\langle x,y\rangle)},
\]
where $\langle x,y\rangle=\sum_{j=1}^n x_jy_j$ is the standard inner product over $\F_q$.
Here $\tr:\F_q\to\F_p$ is the absolute field trace,
\[
    \tr(a)=a+a^p+a^{p^2}+\cdots+a^{p^{m-1}}.
\]

Under these normalizations, for $r\ge 1$ we write
\[
    \norm{f}_r
    =\left(q^{-n}\sum_{x\in\F_q^n}|f(x)|^r\right)^{1/r}
    \quad\text{and}\quad
    \norm{\widehat f}_r
    =\left(\sum_{y\in\F_q^n}|\widehat f(y)|^r\right)^{1/r}.
\]
We define convolution by
\[
    (f*g)(y)=q^{-n}\sum_{x\in\F_q^n} f(x)g(y-x).
\]
With these choices, Parseval's identity takes the form $\norm{f}_2=\norm{\widehat f}_2$, and convolution satisfies
\[
    \widehat{f*g}=\widehat f\cdot \widehat g.
\]

A nonnegative function $f:\F_q^n\to\R_{\ge 0}$ is \deffont{normalized} if its expectation is $1$, i.e.,
\[
    \norm{f}_1=\Eover{x\sim\F_q^n}{f(x)}=1.
\]
We will use the following Fourier pseudorandomness condition.

\begin{definition}[$\alpha$-Fourier-concentration]
    A normalized function $f:\F_q^n\to\R_{\ge 0}$ is \deffont{$\alpha$-Fourier-concentrated}, for $\alpha\in[0,1]$, if for every $y\in\F_q^n$,
    \[
        \inabs{\widehat f(y)}\le \alpha^{\wt{y}}.
    \]
\end{definition}

\subsubsection*{Relative Deviation}
In \cref{eq:setDeviation}, we defined the relative deviation of a code $C$ with respect to a set $B$.
We now extend this notion from sets to nonnegative normalized functions.
This formulation is more convenient when the tests are weighted functions rather than indicators.

For a nonempty set $B\subseteq\F_q^n$, let
\[
    f_B:=\frac{q^n}{|B|}1_B.
\]
Then $f_B$ is normalized, and the set-based relative deviation is recovered from the functional definition below.

\begin{definition}[Relative deviation]
    Let $f:\F_q^n\to\R_{\ge 0}$ be normalized, and let $C\subseteq\F_q^n$ be nonempty.
    Write
    \[
        f(C):=\sum_{x\in C}f(x).
    \]
    The \deffont{relative deviation} of $C$ with respect to $f$ is
    \[
        \delta_{C,f}
        :=\Eover{x\sim C}{f(x)}-\Eover{x\sim\F_q^n}{f(x)}
        =\frac{f(C)}{|C|}-1.
    \]
\end{definition}

In particular, for $f_B=(q^n/|B|)1_B$,
\[
    \delta_{C,f_B}
    =\frac{q^n|C\cap B|}{|C|\,|B|}-1
    =\delta_{C,B}.
\]
Thus, $\inabs{\delta_{C,f_B}}\le\tau$ is equivalent to
\[
    (1-\tau)\frac{|C|\,|B|}{q^n}
    \le |C\cap B|\le
    (1+\tau)\frac{|C|\,|B|}{q^n}.
\]

\subsection{Discrepancy Results for Random Linear Codes}

\subsubsection{Small Relative Deviation for Translations of a Single Function}
We first state a detailed version of \cref{thm:informal-single-set}.
It shows that, for any fixed normalized nonnegative function $f$, a random linear code of rate slightly above the natural threshold for $f$ has exponentially small relative deviation with respect to every translation of $f$.

For a normalized function $f$, define its \deffont{capacity parameter} by
\[
    \beta:=\frac1n\log_q\norm{f}_\infty.
\]
This choice is motivated in detail in \cref{sec:proof-overview}.
For now, observe that if $B\subseteq\F_q^n$ is nonempty and $f=f_B=(q^n/|B|)1_B$, then
\[
    \beta=1-\frac1n\log_q|B|.
\]
Moreover, for a code of rate $R$, the expected intersection of $C$ with a uniformly random translate of $B$ is
\[
    \frac{|C|\,|B|}{q^n}=q^{(R-\beta)n}.
\]
Thus $R=\beta$ is the point at which this expected intersection size is $1$, while rates $R=\beta+\eta$ correspond to the above-capacity regime where the expected intersection size is $q^{\eta n}$.

\begin{definition}[Translation]\label{def:translation}
    For a function $f:\F_q^n\to\mathbb{C}$ and a vector $z\in\F_q^n$, define the \deffont{translation} $f_z:\F_q^n\to\mathbb{C}$ by
    \[
        f_z(x)=f(x+z).
    \]
\end{definition}

\begin{theorem}[Small deviation of a translation family]\label{thm:singleFunction}
    Fix a normalized function $f:\F_q^n\to \R_{\ge 0}$, and let $\beta = \frac 1n\cdot \log_q \norm f_\infty$. Let $\mathcal{F} = \{f_z : z \in \F_q^n\}$ be the family of all translations of $f$. 
    
  Let $C\subseteq \F_q^n$ be a random linear code of rate $R = \beta + \eta$ (where $\eta > 0$ and $k = Rn$), and assume that $n$ is large enough that $\eta n \ge 240\log_q n$. Let $\eps = \frac{\eta}{12}$. Then, with probability at least
    $$ 1 - 4q^{1 - \frac{\eta n}{3}} - \frac{q^{k-n+1}}{(q-1)^2} - q^{2n - \frac{\eta^2 n^2}{1440}} $$
    over the choice of $C$, we have 
    $$\max_{z \in \F_q^n}|\delta_{C,f_z}|\le q^{-\eps n}\eperiod$$
\end{theorem}

\subsubsection{Small Relative Deviation for a Family of Fourier-Concentrated Functions}
We next state a detailed version of \cref{thm:informal-general}, dealing with general families of Fourier-concentrated functions.

\begin{theorem}[Small deviation of Fourier-concentrated Functions]\label{thm:mainGeneral}
Fix a family $\cF$ of normalized and $\alpha$-Fourier-concentrated functions $f:\F_q^n\to \R_{\ge 0}$ (where $0<\alpha<1$), and set $\beta = \max_{f\in \cF} \frac 1n\cdot \log_q \norm f_\infty$. Let $C\subseteq \F_q^n$ be a random linear code of rate $R = \beta + \eta\le 1$ for some $\eta > 0$. Denote
$$d = \max\inset{0, \left\lceil \log_2 \left( \frac{\log_2\inparen{(q-1)/(q^{\eta/100}-1)}}{\log_2(1/\alpha)} \right) \right\rceil} \eperiod$$
Fix a parameter $t > 20\log_q(n)$ such that  $2t/n < 2^{-d} \min\inparen{R - \frac{\eta }{50}, 2\eta}$ and let $$\eps = 2^{-d} \min\inparen{R - \frac{\eta }{50}, 2\eta} - \frac{2t}{n}\eperiod$$ Then, for sufficiently large $n$, with probability at least 
$$ 1 - n(Rn+1)q^{-\frac{\eta n}{100}} - |\cF| \cdot d \cdot q^{n - t^2/10} $$
over the choice of $C$, we have $$|\delta_{C,f}|\le q^{-\eps n}$$ for every $f\in \cF$.
\end{theorem}

The following corollary is a simplified and slightly weaker version of \cref{thm:mainGeneral}.
It is intended to make the trade-off between the family size, the success probability, and the resulting deviation bound easier to parse.

\begin{corollary}[Simplified version of \cref{thm:mainGeneral}]\label{cor:SimplerMain}
For every $n\in \N$, prime power $q$, and $0<\alpha<1$, the following holds. Let
$\cF$ be a family of normalized and $\alpha$-Fourier-concentrated functions
$f:\F_q^n\to\R_{\ge 0}$. 
Let
\[
\beta:=\max_{f\in\cF}\frac1n\log_q\norm f_\infty .
\]
Let $0 < \eta\le 50\beta/51$, and let $C\subseteq\F_q^n$ be a random linear code
of rate $R=\beta+\eta \le 1$. Set
\[
\lambda
:=
\min\left\{
1,\frac{\log(1/\alpha)}{2\log(100q/\eta)}
\right\}.
\]
Let $0 < \theta < \eta \lambda$. Then, assuming $n\ge n_0(q,\eta,\alpha,\theta)$, it holds with probability at least
\[
1-q^{-\frac{\eta n}{200}} - |\cF|q^{n-\frac{\theta^2n^2}{20}},
\]
that
\[
|\delta_{C,f}|\le q^{-2(\eta\lambda-\theta)n}
\]
for every $f\in\cF$. In particular, it suffices to take
$$
n_0 \ge \max\inset{\frac{c\log_q(1/\eta)}\eta, 
    \frac{c}{\theta}
    \sqrt{
        \log_q\left(
            2+\log\left(
                \frac{q}{\eta(1-\alpha)}
            \right)
        \right)}}\ecomma
$$
where $c > 0$ is a sufficiently large universal constant.
\end{corollary}

\begin{remark}[On the probability--deviation trade-off in \cref{cor:SimplerMain}]
The parameter $\theta$ controls the trade-off between the success probability and the deviation bound.
For a simple way to read the corollary, suppose that $\log_q|\cF|$ grows at least linearly in $n$.
If one only wants success probability $1-o(1)$, one may take
\[
    \theta:=\frac{\sqrt{\log_q|\cF|}}{n}\,g(n),
\]
where $g(n)\to\infty$ sufficiently slowly that
\[
    g(n)\sqrt{\log_q|\cF|}=o(\lambda n).
\]
The resulting deviation bound is
\[
    \inabs{\delta_{C,f}}
    \le q^{-2(\eta\lambda-\theta)n}
    =q^{-2\eta\lambda n+2g(n)\sqrt{\log_q|\cF|}}.
\]

In particular, if $|\cF|\le q^{n^\kappa}$ for some fixed $1\le\kappa<2$, then the same calculation yields constants $\gamma=\gamma_{\kappa,\eta}>0$ for which, in the corresponding parameter regimes, the bound takes the form
\[
    \inabs{\delta_{C,f}}\le q^{-n^\gamma},
\]
for all sufficiently large $n$.
For example, this applies in regimes such as $q\le 2^{n^\gamma}$ and $\alpha\le 1-n^{-\gamma}$, after choosing $\gamma$ sufficiently small as a function of $\kappa$ and $\eta$.
\end{remark}

We defer the proof of \cref{cor:SimplerMain} to \cref{sec:corProof}.

\subsection{Applications to List-Decoding and List-Recovery}
We now record two direct coding-theoretic consequences of \cref{thm:singleFunction,thm:mainGeneral}.
The definitions of list-decoding and zero-error list-recovery were given in the introduction; we repeat them here only to fix notation for the formal statements.

\begin{definition}[List-decodable code]
    A code $C\subseteq\F_q^n$ is \emph{$(\rho,L)$-list-decodable} if, for every $z\in\F_q^n$,
    \begin{equation*}
        |B_\rho(z)\cap C|\leq L.
    \end{equation*}
\end{definition}

Recall that $1-\frac 1n\log_q\inabs{B_\rho}$ is the list-decoding capacity.
Below this rate, uniformly random codes achieve constant list size with high probability; above this rate, the relevant benchmark is the expected intersection size with a Hamming ball.
More precisely, a uniformly random code of rate $1-\frac 1n\log_q\inabs{B_\rho}+\gamma$ has about
\[
    \frac{|C|\cdot |B_\rho|}{q^n}\approx q^{\gamma n}
\]
codewords in each Hamming ball with high probability.
The next theorem shows that random linear codes satisfy the same above-capacity discrepancy guarantee.

\begin{theorem}[Random linear codes achieve nearly optimal list size in list-decoding above capacity]\label{thm:list-dec}
   Let $n\in \N$, let $q$ be a prime power, let $0<\rho<1-1/q$, and let
    $0<\gamma<\frac1n\log_q |B_\rho(0)|$, where $\gamma n \ge 240\log_q n$.
Set
    \[
        R:=1-\frac1n\log_q |B_\rho(0)|+\gamma.
    \]
    Let $C\subseteq\F_q^n$ be a random linear code of designed rate $R$.
    Then, with probability at least $$
    1
    -4q^{1-\gamma n/3}
    -\frac{q^{-(1-R)n+1}}{(q-1)^2}
    -q^{2n-\gamma^2n^2/1440}\ecomma$$
     the following holds
    simultaneously for every $z\in\F_q^n$:
    \begin{equation}\label{eq:list-dec-discrepancy}
        |B_\rho(z)\cap C|
        =\left(1\pm q^{-\gamma n/12}\right)\frac{|C|\cdot |B_\rho(0)|}{q^n}.
    \end{equation}
   In particular, $C$ is $(\rho,L)$-list-decodable with
    \[
        L=\left\lceil\left(1+q^{-\gamma n/12}\right)q^{\gamma n}\right\rceil.
    \]
\end{theorem}
\begin{remark}
    Observe that the probability bound is $1-o(1)$ whenever $\gamma^2 n\to \infty$ and $(1-R)n\to \infty$
\end{remark}

\begin{proof}
    We apply \cref{thm:singleFunction} to the normalized indicator of the Hamming ball centered at the origin,
    \[
        f=\frac{q^n}{|B_\rho(0)|}\mathbf{1}_{B_\rho(0)}.
    \]
    Its capacity parameter is
    \[
        \beta_n:=\frac1n\log_q\norm f_\infty
        =1-\frac1n\log_q|B_\rho(0)|.
    \]
    Thus \cref{thm:singleFunction} gives
    \[
        \max_{z\in\F_q^n}|\delta_{C,f_z}|
        \le q^{-\gamma n/12}
    \]
    with probability at least $$1 - 4q^{1 - \frac{\gamma n}{3}} - \frac{q^{Rn-n+1}}{(q-1)^2} - q^{2n - \frac{\gamma^2 n^2}{1440}}\eperiod$$ The functions \(f_z\) are the normalized indicators of the balls \(B_\rho(-z)\);
as \(z\) ranges over \(\F_q^n\), these are exactly all Hamming balls of radius \(\rho n\).
    Translating this relative-deviation bound gives \cref{eq:list-dec-discrepancy}.

    Finally, $|C|\le q^{Rn}$ by construction, so
    \[
        \frac{|C|\cdot |B_\rho(0)|}{q^n}
        \le q^{\gamma n}.
    \]
    This yields the stated list-size bound.
\end{proof}

We now turn to zero-error list-recovery.

\begin{definition}[Zero-error list-recoverable code]
    A code $C\subseteq\F_q^n$ is \emph{$(\ell,L)$-list-recoverable} if, for every collection of sets $S_1,S_2,\dots,S_n\subseteq\F_q$ with $|S_i|\le \ell$ for all $i\in[n]$,
    \begin{equation*}
        |(S_1\times S_2\times\cdots\times S_n)\cap C|\leq L.
    \end{equation*}
    We call $\ell$ the \emph{input list size} and $L$ the \emph{output list size}.
\end{definition}

For the upper-bound guarantee, it suffices to consider rectangles $S_1\times\cdots\times S_n$ with $|S_i|=\ell$ for every $i$, since smaller input lists can be padded to size $\ell$.
The zero-error list-recovery capacity is $1-\log_q\ell$: below this rate uniformly random codes achieve constant output list size, while above this rate the expected intersection with an $\ell$-by-$\cdots$-by-$\ell$ rectangle is exponential.
We focus on the nontrivial range $2\le \ell\le q-1$, since $\ell=1$ and $\ell=q$ give the trivial bounds $L=1$ and $L=|C|$, respectively.

\begin{theorem}[Random linear codes achieve nearly optimal output list size in zero-error list-recovery above capacity]\label{thm:list-rec}
    There exist universal constants $c,c'>0$ such that the following holds.
    Let $n\in\N$, let $q$ be prime, and let $2\le \ell\le q-1$.
    Let
    \[
        0<\gamma\le
        \frac12\min\{\log_q\ell,\,1-\log_q\ell\},
    \]
    and set $R=1-\log_q\ell+\gamma$.  Assume that
    \begin{equation}\label{ineq:recovery1}
        n\ge
        c' \cdot \frac{q^5\log^2(q/\gamma)}{\gamma^2}\eperiod
    \end{equation}

    Let $C\subseteq\F_q^n$ be a random linear code of designed
    rate $R$.  Then, with probability at least
    \begin{equation}\label{ineq:recovery2}
        1-2q^{-c\gamma n/(q^2\log(q/\gamma))},
    \end{equation}
        the following events both hold:
        \begin{enumerate}
            \item the code $C$ is $(\ell,L)$-list-recoverable with
                \[
                    L=
                    \left\lceil
                    \left(1+q^{-c\gamma n/(q^2\log(q/\gamma))}\right)q^{\gamma n}
                    \right\rceil\eperiod
                \]
            \item For every combinatorial rectangle $S = S_1\times \dots \times S_n$ with $|S_i| = \ell$ for each $i$, we have
            $$\inabs{C\cap S} = \inparen{1\pm q^{-c\gamma n/(q^2\log(q/\gamma))}}\frac{|C|\ell^n}{q^n}\eperiod$$
        \end{enumerate}
\end{theorem}

We give the proof of that theorem, which is a careful application of \cref{cor:SimplerMain}, in \cref{app:sec1}.

Essentially the same theorem can also be proven for list recovery from \emph{erasures}, wherein we allow some fraction of the coordinates to be fully corrupt, i.e., have $S_i = \F_q$.
\begin{definition}[Erasure list-recoverable code]
    A code $C\subseteq\F_q^n$ is \emph{$(\zeta,\ell,L)$-list-recoverable} if for every collection of sets $S_1,S_2,\dots,S_n\subseteq\F_q$ with $|S_i|\le \ell$ for at least $(1-\zeta)n$ of the $i$'s it holds that
    \begin{equation*}
        |(S_1\times S_2\times\cdots\times S_n)\cap C|\leq L.
    \end{equation*}
\end{definition}
The erasure list-recovery capacity is $R^{\star}=(1-\zeta)(1-\log_{q}\ell)$, which coincides with the zero-error setting when $\zeta=0$.

\begin{theorem}[Random linear codes achieve nearly optimal output list size in erasure list-recovery above capacity]\label{thm:list-rec-erasures}
    There exist universal constants $c,c'>0$ such that the following holds.
    Let $n\in\N$, let $q$ be prime, and let $2\le \ell\le q-1$ and $\zeta \in (0,1)$.
    Let $\beta = (1-\zeta)(1-\log_{q}\ell)$,
    \[
        0<\gamma\le
        \frac12\min\{1-\beta,\beta\},
    \]
    and set $R=\beta+\gamma$.  

    Let $C\subseteq\F_q^n$ be a random linear code of designed
    rate $R$. Then, assuming the lower bound on $n$ given in \Cref{ineq:recovery1} above holds,
    then the following events both hold, with the probability bound given in \Cref{ineq:recovery2}.
        \begin{enumerate}
            \item The code $C$ is $(\zeta,\ell,L)$-list-recoverable with
                \[
                    L=
                    \left\lceil
                    \left(1+q^{-c\gamma n/(q^2\log(q/\gamma))}\right)q^{\gamma n}
                    \right\rceil\eperiod
                \]
            \item For every combinatorial rectangle $S = S_1\times \dots \times S_n$ with $|S_i| = \ell$ for $(1-\zeta)n$ of the $i$'s, and $|S_i|=q$ for the rest of them, we have
            $$\inabs{C\cap S} = \inparen{1\pm q^{-c\gamma n/(q^2\log(q/\gamma))}}|C|\left( \frac{\ell}{q}\right)^{(1-\zeta)n}\eperiod$$
        \end{enumerate}
\end{theorem}
We defer the (nearly identical) proof to \Cref{app:sec1}.

\section{Proof Overview}
\label{sec:proof-overview}

We organize the proof around two motivating examples.  We begin with binary
list-decoding, where the relevant sets are Hamming balls of a fixed radius.
This illustrates \cref{thm:singleFunction}: all such balls are
translates of the ball centered at the origin, so the family consists of the
translations of a single normalized indicator function.  We then turn to
list-recovery over a constant-size alphabet, where the relevant sets are
combinatorial rectangles
\[
    A_1\times\cdots\times A_n,
    \qquad |A_i|=\ell,\qquad 2\le \ell\le q-1.
\]
This illustrates \cref{thm:mainGeneral}: these rectangle indicators are
$\alpha$-Fourier-concentrated for some $\alpha<1$ bounded away from $1$, and the number
of such rectangles is only $q^{o(n^2)}$.

\subsubsection*{Hamming Balls and the Translation-Family Case}

Let
\[
    B=B(0,\rho n)\subseteq \F_2^n
\]
be the Hamming ball of radius $\rho n$ around the origin, and let
\[
    f_B=\frac{2^n}{|B|} \cdot 1_B
\]
be its normalized indicator.  The Hamming ball around a center $z\in\F_2^n$ is
$B+z$, whose normalized indicator is the translate
\[
    (f_B)_z(x)=f_B(x+z).
\]
Thus, in the binary list-decoding example, the family we need to control is
precisely
\[
    \{(f_B)_z:z\in\F_2^n\},
\]
the set of all translates of one function.  This is the motivating special
case of \cref{thm:singleFunction}.

We expose the random code one generator at a time.  That is, starting from
$C_0=\{0\}$, we consider the random sequence
\[
    C_0\subseteq C_1\subseteq\cdots\subseteq C_k,
    \qquad
    C_i=C_{i-1}+\spn\{u_i\},
\]
where the vectors $u_1,\dots,u_k$ are chosen independently and uniformly from
$\F_2^n$.  The final code is $C=C_k$.

The reason for using this incremental construction is that the deviation $\delta$ has a
simple expected evolution in one step.  Suppose
\[
    C'=C+\spn\{u\}.
\]
Over $\F_2$, this means
\[
    C'=C\cup(C+u).
\]
The old coset $C$ contributes the old deviation, while the new coset $C+u$
behaves like a random translate of $C$.  A direct calculation gives
\[
    \Eover{u}{\delta_{C',f}}
    =
    \frac12\delta_{C,f}.
\]
Thus one new generator should ideally divide the deviation by $2$.

This explains the rate threshold:  If
\[
    \beta=\frac1n\log_2\norm {f_B}_\infty
\]
then the initial deviation of a translate is at most
on the order of $2^{\beta n}$.  If $k=Rn=(\beta+\eta)n$ and each of the $k$ steps divided the deviation
by $2$, then the final deviation would be roughly
\[
    2^{\beta n-k}
    =
    2^{-\eta n},
\]
which is exponentially small.

We formalize this ideal behavior using smoothness.  Given a function $g$ and a
parameter $\Gamma$, we say that the construction sequence is
\deffont{$(g,\Gamma)$-smooth} if
\[
    |\delta_{C_i,g}|
    \le
    2^{k-i}\Gamma
    \qquad
    \text{for every }0\le i\le k.
\]
This is a two-sided deviation bound.  It says that the deviation at time $i$
is no larger than what one would expect if the final deviation were $\Gamma$
and each preceding step divided the deviation by $2$.  In the $q$-ary setting,
the corresponding definition uses $q^{k-i}\Gamma$ in place of
$2^{k-i}\Gamma$.

An iterative dimension-by-dimension viewpoint also appears in Blinovsky's
high-probability covering-radius results~\cite{Bli87,Bli90a} (see
\cite[Theorems 12.3.5 and 12.3.10]{CHLL97}), but, as discussed earlier, the discrepancy problem requires a more
precise version of this process.  
For comparison, in a covering argument with
$g$ the normalized indicator of a ball, it is enough to prove the one-sided
condition
\[
    \delta_{C_i,g}>-1
\]
for some $i$, since this is exactly the statement that the ball contains at least one
codeword of $C_i$.  Once this happens, it remains true for every later code
$C_j\supseteq C_i$.  In the present problem we instead want
$|\delta_{C_k,g}|$ to be small, meaning that the final code samples the ball
with approximately the correct multiplicity.  This property is not monotone:
even if $|\delta_{C_i,g}|$ is small at some intermediate stage, later added
codewords may oversample the same ball and make the deviation large again.
Although the theorem only concerns the final code $C_k$, our proof method requires
the stronger inductive invariant that the deviations are controlled
throughout the construction.

A first attempt would be to prove smoothness directly for every ball
$B+z$.  Indeed, the second-moment smoothness lemma (see
\cref{lem:fSmoothBySecondMoment}) shows that, for a fixed normalized function
$g$ and any $\Gamma>2^{-k}g(0)$, the failure probability is at most
\[
    \frac{(k+1)2^{-k}}{(\Gamma-2^{-k}g(0))^2}\norm g_2^2
    +2^{k-n}.
\]
In the overview, it is useful to read this as roughly
$(k+1)2^{-k}\Gamma^{-2}\norm g_2^2$, as long as $\Gamma$ is not too close to
$2^{-k}g(0)$.  Since
\[
    \norm{(f_B)_z}_2^2
    \le
    \norm{(f_B)_z}_\infty
    =
    \norm {f_B}_\infty
    =
    2^{\beta n},
\]
this gives exponentially small failure probability for each fixed center
$z$, provided $\Gamma$ is chosen exponentially small but still larger than
$2^{-k}(f_B)_z(0)$ with sufficient slack.

This direct argument is not enough.  There are $2^n$ possible centers, while
the second-moment estimate gives only a bound of order $2^{-c n}$ for some
small constant $c>0$.  The latter is not large enough to support a
union bound over all centers.  Thus we need a way to upgrade the probability
after paying for only one direct second-moment estimate. We apply
the direct second-moment estimate only once, to a common auxiliary function,
and then use that common event to control all translates with much better
conditional probability.

In the binary case, the auxiliary function is the ordinary self-convolution
\[
    F_f=f*f.
\]
For the ball function $f_B$, the value $F_{f_B}(w)$ is a normalized count of
pairs
\[
    x,y\in B
    \qquad\text{with}\qquad
    x+y=w.
\]
Thus controlling $F_{f_B}$ on a code is the smoothed version of controlling
\[
    (B+B)\cap C.
\]

The reason for introducing $F_f$ is the one-step variance identity
\[
    \Varover{u}{\delta_{C',f}}
    =
    \frac14\,\delta_{C,F_f}\ecomma
\]
readily proven as part of \cref{lem:prob_step}.
Thus $F_f$ measures the noise in the update of $\delta_{C,f}$.  If the
construction sequence is smooth for $F_f$, then the update for $f$ is
concentrated around its mean $\delta_{C_i,f}/2$ at every step.  Consequently,
unless the deviation of $f$ is already small, the next step is very likely to
reduce it by almost the expected factor $2$.

We now bootstrap as follows.  First, apply the second-moment bound only to the
single function $F_{f_B}$.  This is possible because $F_{f_B}$ has essentially
the same second-moment size as $f_B$:
\[
    \norm{F_{f_B}}_2
    =
    \norm{f_B*f_B}_2
    \le
    \norm{f_B}_1\norm{f_B}_2
    =
    \norm{f_B}_2,
\]
and hence
\[
    \norm{F_{f_B}}_2^2
    \le
    \norm{f_B}_2^2
    \le
    \norm{f_B}_\infty
    =
    2^{\beta n}.
\]
Therefore \cref{lem:fSmoothBySecondMoment} implies that, with probability
at least $1-2^{-\Omega(n)}$, the construction sequence is
\[
    (F_{f_B},\Gamma)\text{-smooth}
\]
for a suitable exponentially small $\Gamma=2^{-c n}$.

The crucial point is that this smoothness event is common to all balls.  For
every translate $(f_B)_z$, the binary self-convolution appearing in the
variance identity is the same:
\[
    (f_B)_z*(f_B)_z=f_B*f_B.
\]
Thus we do not need to prove smoothness separately for the self-convolution of
each ball.  We condition on the single event that $F_{f_B}$ is smooth.

After this conditioning, the probability scale improves dramatically.
The smoothness-transfer lemma, \cref{lem:fSmooth}, says that if the sequence is
$(F_{f_B},\Gamma)$-smooth, then for any fixed translate $(f_B)_z$, the
conditional probability that $(f_B)_z$ fails to be smooth is at most
$2^{-\Omega(t^2)}$, where $t$ is a slackness parameter.  A union bound over all
$2^n$ centers gives failure probability at most
\[
    2^{n-\Omega(t^2)}.
\]
Taking $t=\Theta(n)$ makes this probability $2^{-\Omega(n^2)}$.

This proves the translation-family theorem in the motivating binary
list-decoding case.  The statement of \cref{thm:singleFunction} is more
general: it applies to arbitrary normalized functions and to $q$-ary alphabets.
For $q>2$, the only change is that the ordinary self-convolution is replaced
by the \deffont{mirrored self-convolution}
\[
    F_f=f*\check f,
    \qquad
    \check f(x)=
    f(-x)\eperiod
\]

\subsubsection*{Combinatorial Rectangles and the General Family Case}

We now turn to the setting that motivates \cref{thm:mainGeneral}.  Think of
zero-error list-recovery over a constant-size alphabet $\F_q$.  A typical bad set is a
combinatorial rectangle 
\[
    S_1\times S_2\times\cdots\times S_n,
    \qquad
    S_i\subseteq \F_q,\quad |S_i|=\ell,
\]
where
\[
    2\le \ell\le q-1.
\]
Let $f_S$ denote the normalized indicator of such a rectangle.  The important
facts about these functions are the following.  First, they are
$\alpha$-Fourier-concentrated for some constant $\alpha<1$ depending only on $q$ and
$\ell$.  Indeed, the Fourier transform factors coordinate by coordinate, and
the one-dimensional nontrivial Fourier coefficients are bounded away from
$1$.  Second, there are
\[
    \binom{q}{\ell}^n \le q^{o(n^2)}
\]
such rectangles in the constant-alphabet regime under discussion.  More
generally, \cref{thm:mainGeneral} is designed to handle any family of
$q^{o(n^2)}$ normalized $\alpha$-Fourier-concentrated functions.

The argument for Hamming balls breaks down because the functions in this
family are not all translates of a single function.  For different rectangles
$f\in\cF$, the mirrored self-convolutions $F_f$ may be different.  Thus we cannot apply the direct second-moment estimate to
each $F_f$ separately: that estimate gives only $q^{-\Omega(n)}$ failure
probability for each fixed function, whereas the family size allowed in
\cref{thm:mainGeneral} is much larger than what such a direct union bound can
handle.

The extra assumption we use is Fourier-concentration.  The mirrored self-convolution
improves Fourier-concentration: if $f$ is $\alpha$-Fourier-concentrated, then $F_f$ is
$\alpha^2$-Fourier-concentrated.  Iterating this operation, define
\[
    f^{(0)}=f,
    \qquad
    f^{(r+1)}=F_{f^{(r)}}.
\]
Then $f^{(r)}$ is $\alpha^{2^r}$-Fourier-concentrated.  We choose $d$ so that
\[
    1+(q-1)\alpha^{2^d}\le q^{\eta/100}.
\]
At this depth, the Fourier coefficients of $f^{(d)}$ decay quickly enough that
smoothness follows from a mild code-dependent event.

The event we need is that the dual codes do not have too many vectors at any
Hamming weight.  More precisely,
for
\[
    W_j=\{y\in\F_q^n:\wt{y}=j\},
\]
we require
\[
    |C_i^\perp\cap W_j|
    \le
    \gamma q^{-i}|W_j|
    \qquad
    \text{for every }0\le i\le k
    \text{ and every }1\le j\le n,
\]
where $\gamma=q^{\eta n/100}$.  This event depends only on the construction
sequence, not on the function $f$.  It holds with high probability by a
standard first-moment argument applied to the sets $W_j$.

Assume this dual weight event holds.  If $g$ is normalized and
$\alpha'$-Fourier-concentrated, then
\[
    |\widehat g(y)|
    \le
    (\alpha')^{\wt{y}}.
\]
Using the Fourier identity
\[
    \delta_{C_i,g}
    =
    \sum_{y\in C_i^\perp\setminus\{0\}}\widehat g(y),
\]
we obtain
\[
    |\delta_{C_i,g}|
    \le
    \sum_{y\in C_i^\perp\setminus\{0\}}|\widehat g(y)|
    \le
    \gamma q^{-i}
    \sum_{j=1}^n |W_j|(\alpha')^j
    \le
    \gamma q^{-i}\bigl(1+(q-1)\alpha'\bigr)^n.
\]
Thus, under the dual weight event, every sufficiently Fourier-concentrated function is
smooth.  Applying this to $g=f^{(d)}$ gives
\[
    C_0,\dots,C_k
    \text{ is }(f^{(d)},\Gamma_d)\text{-smooth},
\]
with
\[
    \Gamma_d
    =
    q^{-k}\gamma\bigl(1+(q-1)\alpha^{2^d}\bigr)^n
    \le
    q^{-k+\eta n/50}.
\]

We then build back from $f^{(d)}$ to the original function $f=f^{(0)}$.
At each step we use the same smoothness-transfer lemma as in the
translation-family case.  If
\[
    f^{(r)}=F_{f^{(r-1)}}
\]
is smooth with parameter $\Gamma_r$, then $f^{(r-1)}$ is smooth with parameter
\[
    \Gamma_{r-1}
    =
    q^t
    \max\left\{
        \|f^{(r-1)}\|_\infty q^{-k},
        \sqrt{\Gamma_r}
    \right\}.
\]
except with probability $q^{n-\Omega(t^2)}$.

The square root is the main quantitative loss.  If
$\Gamma_r\approx q^{-a n}$, then this step gives
\[
    \Gamma_{r-1}\approx q^{-a n/2}
\]
up to the slack factor $q^t$.  Thus each reverse step essentially halves the
smoothness exponent.  After $d$ steps, this accounts for the factor $2^{-d}$
in the final exponent.  Since the starting functions are
$\alpha$-Fourier-concentrated with $\alpha<1$ bounded away from $1$, the required
number of iterations $d$ does not grow with $n$.

The slack factor $q^t$ is chosen to make the conditional failure probability
small enough.  Taking
\[
    t=\Theta(2^{-d}\eta n)
\]
makes
\[
    q^{n-\Omega(t^2)}=q^{-\Omega(n^2)}.
\]
This is strong enough to union bound over all $f\in\cF$ and over all $d$
levels of the convolution iteration, even when $|\cF|=q^{o(n^2)}$.

Finally, the initial terms in the recurrence are small because
$R=\beta+\eta$.  For every relevant iterate,
\[
    \|f^{(r)}\|_\infty q^{-k}\le q^{-\eta n},
\]
using the bound $\|f^{(r)}\|_\infty\le \|f\|_\infty\le q^{\beta n}$.
Unfolding the recurrence for the $\Gamma_r$'s gives a final smoothness
parameter
\[
    \Gamma_0\le q^{-\eps n},
\]
where
\[
    \eps=\Theta(\eta 2^{-d})
\]
with the constants made explicit in \cref{thm:mainGeneral}.  Hence
\[
    |\delta_{C_k,f}|\le q^{-\eps n}
\]
for every function in the family.  The dual weight event is common to the
entire family, and the remaining bootstrapping failures have
$q^{-\Omega(n^2)}$ probability, so the final union bound gives
\cref{thm:mainGeneral}.

\section{Random Linear Codes are Smooth---Proof of \cref{thm:mainGeneral,thm:singleFunction}}

In this section we prove \cref{thm:mainGeneral,thm:singleFunction}. We begin with several definitions. As explained in \cref{sec:proof-overview}, we view a given RLC as being constructed one dimension at a time, leading us to define a \deffont{construction sequence}.

\begin{definition}
    A \deffont{construction sequence} is a sequence of codes $C_0,\dots,C_k \subseteq \F_q^n$ where $C_i = C_{i-1} + \spn\{u_i\}$ for some vector $u_i\in \F_q^n$. When $u_1,\dots, u_k$ are chosen uniformly at random, we say that this is a \deffont{random construction sequence}, and each $C_i$ is an \deffont{RLC} of designed dimension $i$.
\end{definition}

We would like to talk about a code having small deviation with regard to a specific function not just for the full code but also throughout its construction sequence. This requirement is captured by the notion of \deffont{smoothness}, motivated in \cref{sec:proof-overview}.

\begin{definition}
    Fix a non-zero function $f:\F_q^n\to \R_{\ge 0}$. A construction sequence $C_0,\dots, C_k$ is \deffont{$(f,\Gamma,j)$-smooth} (for some real $\Gamma \ge 0$ and integer $0\le j\le k$) if 
    $$\inabs{\delta_{C_i,f}} \le q^{k-i}\cdot \Gamma$$
    for all $0\le i\le j$.
    As a shorthand, we write $(f,\Gamma)$-smooth to mean $(f,\Gamma,k)$-smooth.
\end{definition}

\begin{definition}
    For $0\le j\le n$, let $W_j = \inset{x\in \F_q^n : \text{wt}(x) = j}$, where $\text{wt}(x)$ denotes the Hamming weight.

    Let $C_0,\dots, C_k$ be a construction sequence. We say that this sequence is \deffont{$\gamma$-dual-typical} ($\gamma \ge 1$) if $$\inabs{ C_i^\perp \cap W_j} \le \gamma \cdot q^{-i} \cdot |W_j|$$ for all $0\le i\le k$ and $1\le j\le n$. 
\end{definition}

In \cref{sec:bootstrap} we begin with a weaker probabilistic argument that shows that RLCs are smooth for any fixed given function, and also dual-typical, with probability $q^{-\Omega(n)}$. In \cref{sec:convolution} we introduce the averaged convolution function $F_f$ and relate $\delta_{C_i,F_f}$ to the $i$-th step change in $\delta_{C_{i+1},f}$. In \cref{sec:nestedSmoothness} we use this relation to control the smoothness of $f$ via that of $F_f$, and use this implication as the basis to prove \cref{thm:mainGeneral,thm:singleFunction}.

\subsection{Smoothness via a Direct Second Moment Argument}\label{sec:bootstrap}

In this subsection we use a simple second-moment argument to show that RLCs are smooth and dual-typical, with probability ${1-q^{-\Omega(n)}}$. The core of the argument is the following lemma.

\begin{lemma}\label{lem:randomCodeConcentration}
    Fix $n\in\N$ and $m\le n$. Let
    $f:\F_q^n\to\R_{\ge 0}$ be normalized, and let $C$ be the linear code generated by a uniformly random matrix $G \in \F_q^{n \times m}$. Then, for every $\Delta>q^{-m}|f(0)-1|$,
    \[
        \PROver{G}{\inabs{\delta_{C,f}}>\Delta}
        \le
        \frac{q-1}
             {q^m\inparen{\Delta-q^{-m}|f(0)-1|}^2}
        \norm f_2^2
        + \frac{q^{m-n}}{q-1}
        \eperiod
    \]
\end{lemma}

\begin{proof}
    We evaluate the relative deviation by first considering the multiset of all vectors generated by $G$. Define
    \[
        X = \sum_{v \in \F_q^m} f(Gv) \ecomma
    \]
    and let $\delta_{G,f} = q^{-m}X - 1$ be the relative deviation over this multiset. 

    If $G$ has full column rank, the multiset $\{Gv : v \in \F_q^m\}$ contains exactly $q^m$ distinct vectors, meaning it coincides perfectly with the code subspace $C$. In this case, $\delta_{C,f} = \delta_{G,f}$. Therefore, the event that the set deviation exceeds $\Delta$ can only occur if either the multiset deviation exceeds $\Delta$, or $G$ fails to have full rank. By the union bound:
    \[
        \PROver{G}{\inabs{\delta_{C,f}} > \Delta} \le \PROver{G}{\inabs{\delta_{G,f}} > \Delta} + \PROver{G}{\text{rank}(G) < m} \eperiod
    \]

    We first bound the multiset deviation unconditionally over the choice of $G$. Since $f$ is normalized, $\sum_{x\in\F_q^n}f(x)=q^n$. For any fixed $v \ne 0$, the vector $Gv$ is uniformly distributed over $\F_q^n$, so $\Eover{G}{f(Gv)} = 1$. The expectation of $X$ is
    \[
        \Eover{G}{X} = f(0) + \sum_{v \in \F_q^m \setminus \{0\}} \Eover{G}{f(Gv)} = f(0) + q^m - 1 \eperiod
    \]
    Hence, the exact expected multiset deviation is
    \[
        \mu = \Eover{G}{\delta_{G,f}} = q^{-m}\Eover{G}{X} - 1 = q^{-m}\inparen{f(0) - 1} \eperiod
    \]

    We next bound the variance of $X$. Since $f(0)$ is a constant, we can expand the variance using the sum definition:
    \[
        \Varover{G}{X} = \sum_{v,w \in \F_q^m \setminus \{0\}} \inparen{\Eover{G}{f(Gv)f(Gw)} - \Eover{G}{f(Gv)}\Eover{G}{f(Gw)}} \eperiod
    \]
    If $v$ and $w$ are linearly independent, the vectors $Gv$ and $Gw$ are independent and uniformly distributed in $\F_q^n$. Thus, $\Eover{G}{f(Gv)f(Gw)} = \Eover{G}{f(Gv)}\Eover{G}{f(Gw)} = 1$, making the term exactly zero. 

    The only non-zero terms in the double sum occur when $v$ and $w$ are linearly dependent, meaning $w = \alpha v$ for some $\alpha \in \F_q^*$. In this case, $Gw = \alpha Gv$. Since $f$ is non-negative, we can form an upper bound by dropping the subtracted product $\Eover{G}{f(Gv)}\Eover{G}{f(Gw)} = 1$:
    \[
        \Eover{G}{f(Gv)f(\alpha Gv)} - 1 \le \Eover{G}{f(Gv)f(\alpha Gv)} = q^{-n} \sum_{x \in \F_q^n} f(x)f(\alpha x) \eperiod
    \]
    By Cauchy--Schwarz, for every $\alpha \in \F_q^*$:
    \[
        \sum_{x \in \F_q^n} f(x)f(\alpha x) \le \sqrt{\sum_{x \in \F_q^n} f(x)^2} \sqrt{\sum_{x \in \F_q^n} f(\alpha x)^2} = \sum_{x \in \F_q^n} f(x)^2 = q^n \norm f_2^2 \eperiod
    \]
    Substituting this back into the variance expansion, we sum only over the pairs $(v, \alpha v)$. There are $q^m - 1$ choices for $v \ne 0$ and $q - 1$ choices for $\alpha \in \F_q^*$, giving:
    \[
        \Varover{G}{X} \le \sum_{v \in \F_q^m \setminus \{0\}} \sum_{\alpha \in \F_q^*} q^{-n} \inparen{q^n \norm f_2^2} = (q^m - 1)(q - 1) \norm f_2^2 \le q^m(q - 1) \norm f_2^2 \eperiod
    \]
    Scaling down to the relative deviation $\delta_{G,f}$, we get
    \[
        \Varover{G}{\delta_{G,f}} = q^{-2m} \Varover{G}{X} \le \frac{q-1}{q^m} \norm f_2^2 \eperiod
    \]
    By Chebyshev's inequality, to bound the absolute deviation we measure the distance from the mean $\mu$ to the boundary of the interval $[-\Delta, \Delta]$. Provided $\Delta > |\mu|$, we have:
    \[
        \PROver{G}{\inabs{\delta_{G,f}} > \Delta} \le \PROver{G}{\inabs{\delta_{G,f} - \mu} > \Delta - |\mu|} \le \frac{\Varover{G}{\delta_{G,f}}}{\inparen{\Delta - q^{-m}|f(0)-1|}^2} \eperiod
    \]
    Substituting our variance bound yields:
    \[
        \PROver{G}{\inabs{\delta_{G,f}} > \Delta} \le \frac{q-1}{q^m\inparen{\Delta - q^{-m}|f(0)-1|}^2} \norm f_2^2 \eperiod
    \]

    Finally, we bound the probability that $G$ fails to have full column rank. The matrix $G$ is rank-deficient if and only if at least one of its $m$ columns falls into the span of the previous columns. The $i$-th column falls into the span of the first $i-1$ columns with probability at most $q^{i-1}/q^n$. Applying the union bound over all $m$ columns gives
    \[
        \PROver{G}{\text{rank}(G) < m} \le \sum_{i=1}^m \frac{q^{i-1}}{q^n} = q^{-n} \frac{q^m-1}{q-1} \le \frac{q^{m-n}}{q-1} \eperiod
    \]
    Adding this rank-deficiency bound to the multiset deviation bound completes the proof.
\end{proof}
Using \cref{lem:randomCodeConcentration}  across every dimension in the construction sequence and union-bounding the failures, we can show that a random linear code is highly likely to be smooth with respect to $f$.
\begin{lemma}\label{lem:fSmoothBySecondMoment}
    Let $f:\F_q^n\to\R_{\ge 0}$ be normalized, and let
    $C_0,\dots,C_k$ be a random construction sequence.  If
    \[
        \Gamma>q^{-k}|f(0)-1|,
    \]
    then $C_0,\dots,C_k$ is $(f,\Gamma)$-smooth with probability at least
    \[
        1-
        \frac{(k+1)(q-1)q^{-k}}
             {\inparen{\Gamma-q^{-k}|f(0)-1|}^2}
        \norm f_2^2 - \frac{q^{k-n+1}}{(q-1)^2}
        \eperiod
    \]
\end{lemma}

\begin{proof}
    Fix $0\le i\le k$.  Since $C_i$ is an RLC of dimension $i$, applying
    \cref{lem:randomCodeConcentration} with
    \[
        \Delta_i=q^{k-i}\Gamma
    \]
    gives
    \begin{align*}
        \PROver{C_i}{
            \inabs{\delta_{C_i,f}}>q^{k-i}\Gamma
        }
        &\le
        \frac{q-1}
             {q^i\inparen{q^{k-i}\Gamma-q^{-i}|f(0)-1|}^2}
        \norm f_2^2 + \frac{q^{i-n}}{q-1}
        \\&= \frac{(q-1)q^i}
             {\inparen{q^k\Gamma-|f(0)-1|}^2} \norm f_2^2 + \frac{q^{i-n}}{q-1}
        \\&\le
        \frac{(q-1)q^k}
             {\inparen{q^k\Gamma-|f(0)-1|}^2} \norm f_2^2 + \frac{q^{i-n}}{q-1}
        \\&=
        \frac{(q-1)q^{-k}}
             {\inparen{\Gamma-q^{-k}|f(0)-1|}^2} \norm f_2^2 + \frac{q^{i-n}}{q-1}
        \eperiod
    \end{align*}
    A union bound over $i=0,\dots,k$ gives the bound on the sum of the first terms, and sums the rank-deficiency error as $\sum_{i=0}^k \frac{q^{i-n}}{q-1} \le \frac{q^{k-n+1}}{(q-1)^2}$.
\end{proof}
We next show that a random construction sequence is dual-typical with high probability.

\begin{lemma}\label{lem:typical}
    A random construction sequence $C_0,\dots,C_k$ is $\gamma$-dual-typical
    with probability at least
    \[
        1-\frac{n(k+1)}{\gamma}.
    \]
\end{lemma}

\begin{proof}
    Fix $0\le i\le k$ and $1\le j\le n$, and write
    $X_{i,j}:=|C_i^\perp\cap W_j|$.
    For every fixed nonzero $y\in W_j$, the event $y\in C_i^\perp$ is the
    event that $\langle y,u_s\rangle=0$ for every $s\in[i]$.
    Since $y\ne 0$ and the generators $u_s$ are independent and uniform in
    $\F_q^n$, this event has probability $q^{-i}$. Hence
    $\E X_{i,j}=q^{-i}|W_j|$.

    By Markov's inequality,
    \[
        \Pr\left[X_{i,j}>\gamma q^{-i}|W_j|\right]\le \frac1\gamma.
    \]
    A union bound over the $n(k+1)$ choices of $(i,j)$ proves the claim.
\end{proof}

We finish this section by showing that a dual-typical code is smooth on all sufficiently Fourier-concentrated functions.
\begin{lemma}\label{lem:ConcentratedSmooth}
      Suppose that the construction sequence $C_0,\dots,C_k$ is $\gamma$-dual-typical and that $f:\F_q^n\to \R_{\ge 0}$ is normalized and $\alpha$-Fourier-concentrated. Then $C_0,\dots, C_k$ is $(f,q^{-k}\cdot \gamma\inparen{1+(q-1)\alpha}^n)$-smooth. 
\end{lemma}
\begin{proof}
    Observe that
    \begin{align*}
        \inabs{\delta_{C_i,f}} &\le \sum_{y\in C_i^\perp\setminus \{0\}} \left|\widehat f(y)\right| \\
        &\le \sum_{j=1}^n |C_i^\perp\cap W_j|\cdot \alpha^j \\
        &\le \sum_{j=1}^n \gamma\cdot q^{-i}\cdot |W_j|\cdot \alpha^j \\
        &= \gamma\cdot q^{-i}\cdot \sum_{j=1}^n \binom n j \cdot (q-1)^{j}\cdot \alpha^j \\
        &= \gamma \cdot q^{-i} \cdot \inparen{(1+(q-1)\alpha)^n-1}\eperiod \qedhere
    \end{align*}    
\end{proof}

\subsection{Self Convolution and Variance of the Relative Deviation}\label{sec:convolution}
The direct second-moment bound from \cref{sec:bootstrap} yields a failure probability of $q^{-\Omega(n)}$, which is too weak to union-bound over exponentially large function families. To achieve the required $q^{-\Omega(n^2)}$ probability, we must sharply bound the variance of the relative deviation at each step of the code's construction. Remarkably, this one-step variance is entirely controlled by the code's deviation with respect to an auxiliary smoothed function, which we define here as the averaged convolution $F_f$.
\begin{definition}[Self-Convolution Function $F_f$]
    Let $f:\F_q^n\to\R_{\ge 0}$ be a function. We define the \deffont{reflected function} $\check{f}:\F_q^n\to\R$ by $\check{f}(x) = f(-x)$. We then define the \deffont{self-convolution function} $F_f$ as $F_f = f * \check{f}$. Equivalently, $F_f$ can be defined by
    \begin{align*}
        F_f(w) &= q^{-n} \sum_{x \in \F_q^n} f(x) \cdot \check{f}(w-x) \\
               &= q^{-n} \sum_{x \in \F_q^n} f(x) \cdot f(x-w) \eperiod
    \end{align*}
\end{definition}

We next prove several simple properties of $F_f$.

\begin{lemma}[Fourier transform of $F_f$]\label{lem:F_fFourier}
Let $f:\F_q^n \to \mathbb R_{\ge 0}$ be a normalized function. Then, 
$$\widehat{F_f}(y) = \inabs{\widehat f(y)}^2$$
for all $y\in \F_q^n$.
\end{lemma}
\begin{proof}   
    First, observe that
    $$ \widehat{\check{f}}(y) = q^{-n} \sum_{x \in \F_q^n} f(-x) \chi_y(x) =  q^{-n} \sum_{x \in \F_q^n} f(x) \chi_y(-x) = q^{-n} \sum_{x \in \F_q^n} f(x) \overline{\chi_y(x)} = \overline{\widehat f(y)}\eperiod $$
    
    Thus, by the Convolution Identity, $$\widehat{F_f}(y) = \widehat{f}(y) \cdot \widehat{\check{f}}(y) = \widehat{f}(y) \cdot \overline{\widehat f(y)} = \inabs{\widehat f(y)}^2\eperiod$$ 
\end{proof}
\cref{lem:F_fFourier} immediately yields the following corollary.
\begin{lemma}[$\alpha$-Fourier-concentration of $F_f$] \label{lem:concentration_F_f}
    Let $f:\F_q^n\to \R_{\ge 0}$ be a normalized function. If $f$ is $\alpha$-Fourier-concentrated, then its self-convolution function $F_f$ is $\alpha^2$-Fourier-concentrated.
\end{lemma}

\begin{lemma}[Translation Invariance of $F_f$] \label{lem:translation_invariance}
    For any function $f:\F_q^n\to\R$ and any translation vector $z \in \F_q^n$, let $f_z$ be its translation. Then $F_{f_z} = F_f$.
\end{lemma}

\begin{proof}

    We prove this directly using the expanded definition of the self-convolution. Evaluating $F_{f_z}$ at any point $w \in \F_q^n$ yields:
    $$ F_{f_z}(w) = q^{-n} \sum_{x \in \F_q^n} f_z(x) \cdot f_z(x-w) = q^{-n} \sum_{x \in \F_q^n} f(x+z) \cdot f(x-w+z) \eperiod $$
    
    Let us apply the change of variables $y = x+z$. As $x$ iterates over all of $\F_q^n$, $y$ also iterates over all of $\F_q^n$. Substituting $y$ into the summation gives:
    $$ F_{f_z}(w) = q^{-n} \sum_{y \in \F_q^n} f(y) \cdot f(y-w) = F_f(w) \eperiod $$
    Since $F_{f_z}(w) = F_f(w)$ for all $w$, we have $F_{f_z} = F_f$.
\end{proof}

The following lemma describes how $\delta_{C,F_f}$ controls the single-step change in $\delta_{C,f}$.

\begin{lemma}[Deviation of $F_f$ controls the change in deviation of $f$] \label{lem:prob_step}
    Let $f:\F_q^n\to \R_{\ge 0}$ be normalized. Fix a linear code $C\subseteq \F_q^n$ and assume $\delta_{C,f}\ne0$. Let $u \sim \F_q^n$ be a uniformly random vector, and let $C' = C + \operatorname{span}(u)$. Then, for any $\lambda > 0$:
    \[ \PROver{u}{\inabs{\delta_{C',f}} > \inabs{\delta_{C,f}} \inparen{\frac{1}{q} + \lambda}} \le \frac{\inparen{\frac{q-1}{q}}^2 \delta_{C, F_f}}{\lambda^2\cdot \delta_{C,f}^2} \eperiod \]   
\end{lemma}
\begin{proof}
    For any normalized function $f:\F_q^n\to\R$ (meaning $\widehat{f}(0) = \norm{f}_1 = 1$), we can express its relative deviation over $C$ using the Fourier transform. Utilizing the property that $\sum_{x\in C}\chi_y(x)=|C|\cdot \ind{[y\in C^\perp]}$, we have:
    \begin{equation} \label{eq:fourier_deviation}
        \delta_{C,f} = \frac{1}{|C|}\sum_{x\in C}f(x)-1 = \sum_{y\in C^\perp}\widehat{f}(y)-1 = \sum_{y\in C^\perp\setminus\{0\}}\widehat{f}(y) \ecomma
    \end{equation}
    where the second equality is due to the Convolution Identity.

    For the augmented code $C' = C + \text{span}(u)$, its dual space is $(C')^\perp=\{y\in C^\perp \mid \langle y,u\rangle=0\}$. Let $\xi_y(u)$ be the indicator random variable $\ind{ [\langle y,u\rangle=0]}$. Thus, for any vector $u \in \F_q^n$, the deviation over the augmented code becomes:
    $$ \delta_{C',f} = \sum_{y\in C^\perp\setminus\{0\}}\widehat{f}(y)\cdot \xi_y(u) \eperiod $$

    Since $u$ is chosen uniformly at random, for any non-zero $y \in \F_q^n$, the inner product $\langle y,u\rangle$ is uniformly distributed over $\F_q$. Therefore, $\Eover{u}{\xi_y(u)}=1/q$. By linearity of expectation,
    $$ \Eover{u}{\delta_{C',f}} = \frac{1}{q}\sum_{y\in C^\perp\setminus\{0\}}\widehat{f}(y) = \frac{\delta_{C,f}}{q} \eperiod $$
    
    Next, we bound the variance of $\delta_{C',f}$. We evaluate
    $$\Varover{u}{\delta_{C',f}} = \sum_{y,y'\in C^\perp\setminus\{0\}} \widehat{f}(y)\cdot \overline{\widehat{f}(y')} \inparen{ \Eover{u}{\xi_y(u)\xi_{y'}(u)} - \Eover{u}{\xi_y(u)}\Eover{u}{\xi_{y'}(u)} }\eperiod$$
    Observe that for linearly independent $y$ and $y'$, 
    $$\Eover{u}{\xi_y(u)\xi_{y'}(u)} - \Eover{u}{\xi_y(u)}\Eover{u}{\xi_{y'}(u)} = q^{-2} - q^{-2} = 0\ecomma$$
    so these terms vanish from the sum. For $y$ and $y'$ colinear, we have
    $$\Eover{u}{\xi_y(u)\xi_{y'}(u)} - \Eover{u}{\xi_y(u)}\Eover{u}{\xi_{y'}(u)} = q^{-1} - q^{-2} = \frac{q-1}{q^2}\eperiod$$
    Thus,
    \begin{align*}
        \Varover{u}{\delta_{C',f}} 
        &= \frac{q-1}{q^2}\sum_{y\in C^\perp\setminus\{0\}} \sum_{c\in\F_q^*}\widehat{f}(y)\cdot \overline{\widehat{f}(cy)}  \\
        &\le \frac{q-1}{q^2}\sum_{y\in C^\perp\setminus\{0\}}\sum_{c\in\F_q^*}\inabs{\widehat{f}(y)}\cdot \inabs{\widehat{f}(cy)} && \text{(triangle inequality)} \\
        &\le \frac{q-1}{2q^2}\sum_{y\in C^\perp\setminus\{0\}}\sum_{c\in\F_q^*} \inparen{\inabs{\widehat{f}(y)}^2 + \inabs{\widehat{f}(cy)}^2} && \text{(AM-GM inequality)} \\
        &= \frac{q-1}{2q^2} \inparen{ (q-1)\sum_{y\in C^\perp\setminus\{0\}}\inabs{\widehat{f}(y)}^2 + (q-1)\sum_{y'\in C^\perp\setminus\{0\}} \inabs{\widehat{f}(y')}^2 } && \text{(taking } y'=cy\text{)}\\
        &= \frac{q-1}{2q^2} \cdot 2(q-1) \sum_{y\in C^\perp\setminus\{0\}}\inabs{\widehat{f}(y)}^2 &&  \\
        &= \inparen{\frac{q-1}{q}}^2\sum_{y\in C^\perp\setminus\{0\}}\inabs{\widehat{f}(y)}^2 \\
        &=\inparen{\frac{q-1}{q}}^2\sum_{y\in C^\perp\setminus\{0\}}\widehat{F_f}(y) &&\text{(by \cref{lem:F_fFourier})} \\
        &= \inparen{\frac{q-1}{q}}^2\delta_{C,F_f}\eperiod &&\text{(by \cref{eq:fourier_deviation}})
    \end{align*}
   Therefore, by Chebyshev's inequality,
\begin{align*}
    \PROver{u}{
        |\delta_{C',f}|
        >
        |\delta_{C,f}|\left(\frac1q+\lambda\right)
    }
    &\le
    \PROver{u}{
        \left|
        \delta_{C',f}-\Eover{u}{\delta_{C',f}}
        \right|
        >
        \lambda|\delta_{C,f}|
    } \\
    &\le
    \frac{\Varover{u}{\delta_{C',f}}}
    {\lambda^2\delta_{C,f}^2} \le
    \frac{\left(\frac{q-1}{q}\right)^2\delta_{C,F_f}}
    {\lambda^2\delta_{C,f}^2}.
\end{align*}
\end{proof}

  \subsection{Controlling the Deviation throughout the Construction Sequence}\label{sec:nestedSmoothness}
  To extend our one-step variance bound across the entire $k$-step construction, we must carefully track the deviation's trajectory. The following lemma achieves this by categorizing each step as: 
  \begin{itemize}
      \item \deffont{Horrible}---steps in which the deviation grows.
      \item \deffont{Good}---steps in which the deviation goes down significantly, or is already as small as expected to begin with.
      \item \deffont{Bad}---steps in which the deviation does not increase, but also does not drop significantly even though it is expected to.  
  \end{itemize}
  The lemma demonstrates that a breakdown in $f$'s smoothness requires an accumulation of too many ``bad'' steps or at least one ``horrible'' step. We bound the probability of the first sort of failure, namely, we essentially show that a failure due to $t$ bad steps occurs with probability at most $q^{-\Omega(t^2)}$. Later, in \cref{lem:fSmooth}, we take care of the other failure condition.
\begin{lemma}\label{lem:fSmoothOrIncreasing}
Let $f:\F_q^n\to\R_{\ge 0}$ be normalized, let $C_0,\dots,C_k$ be a random
construction sequence, let $\Gamma>0$, and let $t\in\N$ satisfy $20\log_q(n)\le t\le k$.
Set
\[
    \Lambda:=\max\{\|f\|_\infty \cdot q^{-k},\sqrt{\Gamma}\}.
\]
    \begin{itemize}
        \item Let $S_\textrm{$F$-smooth}(f)$ denote the event that $C_0,\dots, C_k$ is $(F_f, \Gamma)$-smooth.
        \item Let $S_{\textrm{monotone}}(f)$ denote the event that $|\delta_{C_j,f}| \le |\delta_{C_{j-1},f}|$ for all $j$ such that $(C_0,\dots, C_{j-1}) $ is $\inparen{f, q^t\cdot \Lambda,j-1}$-smooth.
        \item Let $S_{\textrm{smooth}}(f)$ denote the event that $C_0,\dots, C_{k}$ is $\inparen{f, q^t\cdot \Lambda}$-smooth.
    \end{itemize}  

    Then \begin{equation}\label{eq:ConcentrationWantToProve}
        \PROver{C_0,\dots, C_k}{S_\textrm{$F$-smooth}(f)\text{ and } S_{\textrm{monotone}}(f) \textrm{ and \underline{not} } S_{\textrm{smooth}}(f)} \le q^{- \frac{1}{10} t^2 }.
    \end{equation}
\end{lemma}
\begin{proof}[Proof of \cref{lem:fSmoothOrIncreasing}]
   The event $S_{\textrm{smooth}}(f)$ holds if and only if $|\delta_{C_i,f}| \le q^{k-i+t} \cdot \Lambda$ for all $0\le i\le k$.

Suppose that $S_{\textrm{monotone}}(f)$ holds but $S_{\textrm{smooth}}(f)$ fails. Let $0\le \ell \le k$ be the minimum index for which $C_0,\dots, C_k$ is not $(f,q^t\Lambda,\ell)$-smooth In particular, 
\begin{equation}\label{eq:NotSmooth}
    q^{k-\ell+t} \Lambda < |\delta_{C_\ell,f}|\eperiod
\end{equation}
For a given step $i \in [\ell]$ in the construction sequence, we make the following classification where $\lambda=\frac{t}{4qk}$:
    \begin{enumerate}
        \item The $i$-th step is \deffont{horrible} if $|\delta_{C_i,f}|>|\delta_{C_{i-1},f}|$;

        \item The $i$-th step is \deffont{good} if either $|\delta_{C_i,f}|\leq |\delta_{C_{i-1},f}|(\frac{1}{q}+\lambda)$ or $|\delta_{C_i,f}|\leq |\delta_{C_{i-1},f}|\leq q^k \Lambda \inparen{\frac{1}{q}+\lambda}^{i-t/2}$;

        \item The $i$-th step is \deffont{bad} if it is not horrible, $|\delta_{C_{i-1},f}|>q^k \Lambda (\frac{1}{q}+\lambda)^{i-t/2}$, and $|\delta_{C_i,f}|>|\delta_{C_{i-1},f}|(\frac{1}{q}+\lambda)$.
        In other words, the $i$-th step is bad if it is neither horrible nor good.

    \end{enumerate}
    First, note that if the event $S_{\textrm{monotone}}(f)$ holds then no horrible steps may occur up to and including the $\ell$-th step. 
    Now, for $0\le i\le \ell$ let $g_i$ denote the number of good steps from among the steps $\{1,\dots, i\}$. We claim that if no step is horrible then 
    \begin{equation}\label{eq:GoodStepsGeneral}
        |\delta_{C_i,f}| \leq q^k \Lambda \left(\frac{1}{q}+\lambda\right)^{g_i - t/2} \eperiod
    \end{equation} 
    We prove this by induction on $i$.
    This is clearly true for $i=0$, since $\delta_{C_0,f}\leq q^k \Lambda$.
    Now, suppose \cref{eq:GoodStepsGeneral} holds for $i$.
    We will show that it also holds for $i+1$.
    If the $(i+1)$-st step is bad then $g_{i+1}=g_i$, and so in this case 
    \begin{equation*}
        |\delta_{C_{i+1},f}|\leq |\delta_{C_i,f}|\leq q^k \Lambda \left(\frac{1}{q}+\lambda\right)^{g_i - t/2}=q^k \Lambda \left(\frac{1}{q}+\lambda\right)^{g_{i+1} - t/2}\eperiod
    \end{equation*}
    On the other hand, if the $(i+1)$-st step is good then $g_{i+1}=g_i+1\leq i+1$.
    We now consider two cases.
    \begin{itemize}
        \item If this step is good because 
        \begin{equation*}
            |\delta_{C_{i+1},f}|\leq |\delta_{C_{i},f}|\leq q^k \Lambda 
            \inparen{\frac{1}{q}+\lambda}^{i+1-t/2}\ecomma
        \end{equation*}
        then we also get $|\delta_{C_{i+1},f}|\leq q^k \Lambda (\frac{1}{q}+\lambda)^{g_{i+1}-t/2}$ since $g_{i+1}\leq i+1$.

        \item If this step is good because $|\delta_{C_{i+1},f}|\leq |\delta_{C_{i},f}|(\frac{1}{q}+\lambda)$, then
        \begin{equation*}
            |\delta_{C_{i+1},f}|\leq q^k \Lambda \inparen{\frac{1}{q}+\lambda}^{g_{i}-t/2+1} = q^k \Lambda \inparen{\frac{1}{q}+\lambda}^{g_{i+1}-t/2}, 
        \end{equation*}
        since $g_{i+1}=g_i+1$.
    \end{itemize}
    In both cases we get the desired upper bound on $|\delta_{C_{i+1},f}|$, proving \cref{eq:GoodStepsGeneral}.
    
We turn to show that the event combination $S_{\textrm{monotone}}(f)$ and $\neg S_{\textrm{smooth}}(f)$ implies an upper bound on the number of good steps $g_\ell$. 
Let $c = 1 - \log_q\inparen{1 + \frac{t}{4k}}$. Since $q \ge 2$, we strictly have $c > 0$. We can rewrite our step multiplier exactly as $\frac{1}{q} + \lambda = q^{-c}$. By \cref{eq:GoodStepsGeneral}, $|\delta_{C_\ell,f}|\leq q^k \Lambda q^{-c(g_{\ell}-t/2)}$. Therefore:
$$ q^{k-\ell+t} \Lambda < |\delta_{C_\ell,f}| \le q^{k - c(g_\ell-t/2)} \Lambda \ecomma $$where the left-hand inequality is due to \cref{eq:NotSmooth}. Comparing exponents yields:
$$ c(g_\ell-t/2) < \ell - t \implies g_\ell < \frac{\ell - t}{c} + \frac{t}{2} \eperiod $$

Let $b_\ell$ be the number of bad steps up to step $\ell$, so $g_\ell = \ell - b_\ell$. Combining these inequalities and isolating $b_\ell$:
$$ \ell - b_\ell < \frac{\ell - t}{c} + \frac{t}{2} \implies b_\ell > \ell\left(1 - \frac{1}{c}\right) + t\left(\frac{1}{c} - \frac{1}{2}\right) \eperiod $$

Since $c < 1$, the term $(1 - \frac{1}{c})$ is negative. To find the minimum required bad steps $\tau$, we minimize the right-hand side by setting $\ell = k$. We explicitly set:
 $$ \tau = \max\inset{0, \left\lfloor k\left(1 - \frac{1}{c}\right) + t\left(\frac{1}{c} - \frac{1}{2}\right) \right\rfloor} \eperiod $$
    
    We have shown that the conjunction of events $S_{\textrm{monotone}}(f)\land \neg S_{\textrm{smooth}}(f)$ implies that there exists some failure step $\ell \in [\tau, k]$ such that $b_\ell > \tau$, and crucially, the sequence is monotone up to step $\ell-1$. 

For $0\le i\le k$, let $T_i$ denote the event that $|\delta_{C_j, F_f}| \le q^{k-j}\cdot \Gamma$ for all $1\le j\le i$. Note that the event $S_{\textrm{$F$-smooth}}(f)$ directly implies $T_k$.
For an increasing sequence $1\le i_1 < i_2 < \dots < i_\tau \le k$, let $W_{i_1,\dots,i_\tau}$ denote the event that the steps $i_1,\dots, i_\tau$ are bad \textbf{and the sequence is monotone up to step $i_\tau-1$}.

Therefore, the left-hand side of \cref{eq:ConcentrationWantToProve} is bounded from above by the union over all such failing sequences:
$$\PROver{C_0,\dots, C_k}{S_\textrm{$F$-smooth}(f)\text{ and } \bigcup_{i_1<\dots<i_\tau} W_{i_1,\dots,i_\tau}}\eperiod$$
We bound this by summing over all sequences of $\tau$ bad steps:
\begin{align*}
\PROver{C_0,\dots, C_k}{S_\textrm{$F$-smooth}(f)\text{ and } \bigcup_{i_1<\dots<i_\tau} W_{i_1,\dots,i_\tau}}&\le \sum_{i_1,\dots,i_\tau} \Pr [W_{i_1,\dots,i_\tau} \cap T_k]  \\= \sum_{i_1,\dots,i_\tau} \prod_{r=1}^\tau \Pr_{u_{i_r}} [ W_{i_1,\dots,i_r} \cap T_{i_r} \mid W_{i_1,\dots,i_{r-1}} \cap T_{i_{r-1}}]& \text{(conditioning fixes $C_{i_r-1}$, leaving only $u_{i_r}$ as random)}  \\
    \le \sum_{i_1,\dots,i_\tau} \prod_{r=1}^\tau \Pr_{u_{i_r}} [ W_{i_1,\dots,i_r} \mid W_{i_1,\dots,i_{r-1}} \cap T_{i_{r-1}} ]  \eperiod \numberthis \label{eq:ManyBadStepsBoundGeneral}
\end{align*}

    To bound the probability $\Pr_{u_{i_r}} [ W_{i_1,\dots,i_r} \mid W_{i_1,\dots,i_{r-1}} \cap T_{i_{r-1}} ]$, fix a code $C_{i_r-1}$ for which the conditioning holds. If $\delta_{C_{i_r-1},f}=0$ then the $i_r$-th step cannot be bad, so the probability of $W_{i_1,\dots,i_r}$ under the conditioning is $0$. Suppose now that $\delta_{C_{i_r-1},f}\ne 0$. Then, \cref{lem:prob_step} yields
    $$\PROver{u_{i_r}}{\text{step } i_r \text{ is bad} \mid C_{i_r-1}} \le \frac{\inparen{\frac{q-1}{q}}^2 |\delta_{C_{i_r-1}, F_f}|}{\lambda^2 \cdot \delta_{C_{i_r-1},f}^2} \eperiod$$

    The event $T_{i_r-1}$ yields an upper bound of $q^{k-(i_r-1)} \Gamma = q^{k-i_r+1} \Gamma$ on the numerator. Since $\Lambda = \max\inset{\norm{f}_\infty \cdot q^{-k}, \sqrt \Gamma}$ and $\norm{f}_\infty  \ge |\delta_{C_0,f}|$, we have $\Gamma \le \Lambda^2$, so the numerator is at most $q^{k-i_r+1} \Lambda^2$.
   For the denominator, we need a lower bound on $|\delta_{C_{i_r-1},f}|$. By definition, the $i_r$-th step being bad implies that
$$ |\delta_{C_{i_r-1},f}|>q^k \Lambda\inparen{\frac{1}{q}+\lambda}^{i_r-t/2} \eperiod $$

Note that we must have $i_r \ge t/2$. If $i_r < t/2$, the multiplier $(\frac{1}{q}+\lambda)^{i_r-t/2} > 1$, which would require $|\delta_{C_{i_r-1},f}| > q^k \Lambda$. This is impossible because the event $W_{i_1,\dots,i_r}$ guarantees monotonicity up to $i_r-1$, ensuring $|\delta_{C_{i_r-1},f}| \le |\delta_{C_0,f}| \le \norm{f}_\infty \le q^k \Lambda$.

Since $i_r \ge t/2$ and $1/q+\lambda > 1/q$, we can safely lower-bound the value:
$$ |\delta_{C_{i_r-1},f}| > q^k \Lambda \inparen{\frac{1}{q}}^{i_r-t/2} = q^{k-i_r+t/2} \Lambda \eperiod $$
Squaring this lower bound, the denominator is strictly greater than $\lambda^2 q^{2k - 2i_r + t} \Lambda^2$. Therefore,
    \begin{align*}
    \Pr_{u_{i_r}} [ W_{i_1,\dots,i_r} \mid W_{i_1,\dots,i_{r-1}} \cap T_{i_{r-1}} ] &\le \frac{\inparen{\frac{q-1}{q}}^2 \cdot q^{k-i_r+1} \Lambda^2}{\lambda^2 \cdot q^{2k-2i_r+t} \Lambda^2} \\
    &\le \frac{1}{\lambda^2} q^{-k + i_r + 1 - t} \\
    &\le \frac{q}{\lambda^2} q^{-t} \quad \text{ (since } i_r \le k \text{)} \eperiod
    \end{align*}

Returning to \cref{eq:ManyBadStepsBoundGeneral}, we substitute this uniform bound for all $\tau$ bad steps. Recalling our exact substitution $\lambda = \frac{t}{4qk}$:
    \begin{align*}
    \PROver{C_0,\dots, C_k}{S_\textrm{$F$-smooth}(f)\text{ and } S_{\textrm{monotone}}(f) \textrm{ and \underline{not} } S_{\textrm{smooth}}(f)} 
    &\le \sum_{i_1,\dots,i_\tau} \Pr [W_{i_1,\dots,i_\tau} \cap T_k] \\
    &\le \sum_{i_1,\dots,i_\tau} \prod_{r=1}^\tau \inparen{\frac{16 q^3 k^2}{t^2} q^{-t}} \\
    &\le \binom{k}{\tau} \inparen{\frac{16 q^3 k^2}{t^2} q^{-t}}^\tau \eperiod
    \end{align*}

    We bound the binomial coefficient using the standard inequality $\binom{k}{\tau} \le k^\tau$. Substituting this into our probability bound isolates the exponent:
    \begin{align*}
        \PROver{C_0,\dots, C_k}{S_\textrm{$F$-smooth}(f)\text{ and } S_{\textrm{monotone}}(f) \textrm{ and \underline{not} } S_{\textrm{smooth}}(f)} &\le k^\tau \inparen{\frac{16 q^3 k^2}{t^2} q^{-t}}^\tau \\
        &= \inparen{\frac{16 q^3 k^3}{t^2} q^{-t}}^\tau \\
        &= q^{-\tau \inparen{t - \log_q\inparen{\frac{16 q^3 k^3}{t^2}}}} \eperiod
    \end{align*}

    Now we bound $\tau$. Since $q \ge 2$, we have $\ln q \ge \ln 2 \approx 0.693$. Using the standard bound $\log_q(1+x) \le \frac{x}{\ln q}$, let $L = \log_q\inparen{1 + \frac{t}{4k}}$. We have $L \le \frac{t}{4k \ln q} \le \frac{t}{2.77 k}$. For valid fractional parameters where $t \le k$, we have $L \le 0.36 < 1$, ensuring $c = 1 - L > 0$. Substitute $c = 1-L$ into the required steps condition:
    \begin{align*}
        \tau \ge \inparen{1 - \frac{1}{1-L}} + t\inparen{\frac{1}{1-L} - \frac{1}{2}} &= k\inparen{\frac{-L}{1-L}} + t\inparen{\frac{1}{1-L} - \frac{1}{2}} \\
        &= \frac{1}{1-L} \inparen{-kL + t - \frac{t(1-L)}{2}} \\
        &\ge \frac{1}{1-L} \inparen{-\frac{t}{4\ln q} + \frac{t}{2} + \frac{tL}{2}} \\
        &\ge t \inparen{\frac{1}{2} - \frac{1}{4\ln 2}} \ge 0.13 t \eperiod
    \end{align*}
    Therefore, $\tau \ge \lfloor 0.13 t \rfloor \ge 0.1 t$ for sufficiently large $t$.

    Substituting this simplified lower bound into our probability inequality yields the final dependence:
    \begin{align*}
        \PROver{C_0,\dots, C_k}{S_\textrm{$F$-smooth}(f)\text{ and } S_{\textrm{monotone}}(f) \textrm{ and \underline{not} } S_{\textrm{smooth}}(f)} &\le q^{-0.13t \inparen{t - \log_q\inparen{\frac{16 q^3 k^3}{t^2}}}} \\
        &\le q^{-0.1t^2}= q^{- \Omega(t^2)} \eperiod 
    \end{align*}
    which completes the proof.
\end{proof}
We now show that smoothness of $F_f$ implies smoothness for $f$ with very high probability. In \cref{lem:ConcentratedSmooth} we showed this implication ``modulo'' the possibility of failure due to horrible steps. Below we deal with this obstacle: we show that if $f$ fails to be smooth then there exists some translation  $f_z$ that fails to be smooth but also avoids horrible steps. The claim then follows from \cref{lem:fSmoothOrIncreasing} via a union bound on all translations of $f$.
\begin{lemma} \label{lem:fSmooth}
In the setting of \cref{lem:fSmoothOrIncreasing}, 
$$\PROver{C_0,\dots, C_k}{S_\textrm{$F$-smooth}(f) \textrm{ and \underline{not} } S_{\textrm{smooth}}(f)} \le q^n\cdot q^{-0.1t^2}\eperiod$$
\end{lemma}
\begin{proof}
    We first prove the following implication:
    \begin{equation}\label{eq:ExistenceOfSmooth}
        \textrm{If }\textrm{\underline{not} }S_{\textrm{smooth}}(f)\quad\textrm{then}\quad\exists z\in \F_q^n~~S_{\textrm{monotone}}(f_z) \textrm{ and \underline{not} } S_{\textrm{smooth}}(f_z)\eperiod
    \end{equation}
    
    Assume $S_{\textrm{smooth}}(f)$ fails at a first time $0\le j\le k$, meaning $\inabs{\delta_{C_j,f}} > q^{k-j+t}\Lambda$, where $\Lambda = \max\inset{\norm{f}_\infty q^{-k}, \sqrt \Gamma}$. 
    
    We construct a sequence of vectors $z_j, \dots, z_0$ backward, starting with $z_j = 0$. For $0 \le i < j$, after choosing $z_{i+1}$, we choose $z_i = z_{i+1} + \gamma_i u_{i+1}$ with $\gamma_i \in \F_q$ maximizing
    \[
        \inabs{\delta_{C_i, f_{z_{i+1}+\gamma u_{i+1}}}} \eperiod
    \]
    Since $C_{i+1}=C_i+\spn\{u_{i+1}\}$, we can express the deviation over the larger code as an average over the $q$ cosets:
    \[
        \delta_{C_{i+1},f_{z_{i+1}}} = \frac1q\sum_{\gamma\in\F_q} \delta_{C_i,f_{z_{i+1}+\gamma u_{i+1}}} \eperiod
    \]
    Because the absolute value of an average is bounded by its maximum component, our choice of $\gamma_i$ implies:
    \[
        \inabs{\delta_{C_i,f_{z_i}}} \ge \inabs{\delta_{C_{i+1},f_{z_{i+1}}}} \eperiod
    \]
    
    Next, setting $z = z_0$, we replace the moving translations $f_{z_i}$ with the single fixed translate $f_z$. For every $0\le i\le j$, we have $z-z_i\in C_i$ by construction. Hence, translating by $z-z_i$ preserves the subspace $C_i$, and therefore:
    \[
        \delta_{C_i,f_z}=\delta_{C_i,f_{z_i}} \eperiod
    \]
    Together with the previous bound, this yields:
    \[
        \inabs{\delta_{C_i,f_z}} \le \inabs{\delta_{C_{i-1},f_z}}
    \]
    for all $1 \le i \le j$. Since $S_{\textrm{smooth}}(f_z)$ already fails at time $j$, the definition of $S_{\textrm{monotone}}(f_z)$ imposes no conditions after the first failure time, so the required monotonicity is fully satisfied.

    Furthermore, this fixed translation inherits the original failure at time $j$. Since $z_j=0$, it follows that $z-z_j = z \in C_j$. Thus:
    \[
        \inabs{\delta_{C_j,f_z}} = \inabs{\delta_{C_j,f}} > q^{k-j+t}\Lambda \implies \neg S_{\textrm{smooth}}(f_z) \ecomma
    \]
    where we use the fact that $\norm{f_z}_\infty = \norm{f}_\infty$, ensuring the threshold $\Lambda$ remains identical.
    
    This establishes \cref{eq:ExistenceOfSmooth}. The lemma then follows by applying a union bound:
    \begin{align*}
    &\PROver{C_0,\dots,C_k}{S_\textrm{$F$-smooth}(f)\textrm{ and \underline{not} } S_{\textrm{smooth}}(f)} \\
    &\le \PROver{C_0,\dots,C_k}{\exists z\in \F_q^n~~S_\textrm{$F$-smooth}(f_z)\textrm{ and }S_{\textrm{monotone}}(f_z)\textrm{ and \underline{not} } S_{\textrm{smooth}}(f_z)} \\
    &\le \sum_{z\in \F_q^n}\PROver{C_0,\dots,C_k}{S_\textrm{$F$-smooth}(f_z)\textrm{ and }S_{\textrm{monotone}}(f_z)\textrm{ and \underline{not} } S_{\textrm{smooth}}(f_z)} \\
    &\le q^{n}\cdot q^{-0.1t^2} \eperiod
    \end{align*}
    The first inequality uses \cref{eq:ExistenceOfSmooth} and the fact that $F_f = F_{f_z}$ implies $S_\textrm{$F$-smooth}(f) = S_\textrm{$F$-smooth}(f_z)$. The final bound follows directly from \cref{lem:fSmoothOrIncreasing}.
\end{proof}

\subsubsection{Proof of \cref{thm:singleFunction}}
We now prove \cref{thm:singleFunction}, which controls the deviation for all translations of a single function $f$. The core insight is that every translation $f_z$ shares the exact same averaged convolution $F_f$. By securing a global smoothness bound for this single shared $F_f$, we use \cref{lem:fSmooth} to bound the probability to $q^{-\Omega(n^2)}$, then using union bound over all $q^n$ translations.
\begin{proof}[Proof of \cref{thm:singleFunction}]
    Let $c=\frac \eta 3$. Let $\Gamma = q^{-c n}$. Let $S_{\textrm{$F$-smooth}}(f)$ be the event that the construction sequence $C_0,\dots,C_k$ is $(F_f, \Gamma)$-smooth. This requires $\inabs{\delta_{C_i, F_f}} \le q^{k-i}\Gamma$ for all $0 \le i \le k$. The lower bound $\delta_{C_i, F_f} \ge -q^{k-i}\Gamma$ holds deterministically since 
    $$\delta_{C_i,F_f} = \frac{F_f(C)}{|C|}-1 = \widehat{F_f}(C^\perp)-1 = \sum_{y\in C^\perp\setminus \{0\}} \widehat{F_f}(y) \ge 0\eperiod$$
    Here, the second transition is Parseval's identity, the third is since $F_f$ is normalized so $\widehat{F_f}(0)=1$, and the inequality is because \cref{lem:F_fFourier} implies $\widehat{F_f}(y) \ge 0$ for all $y$.
    
Applying \cref{lem:randomCodeConcentration} to $C_i$ with $\Delta =q^{k-i}\Gamma$ yields:
    \[
        \PR{\delta_{C_i, F_f} > q^{k-i}\Gamma} \le \frac{q-1}{q^i \inparen{q^{k-i}\Gamma - |F_f(0)-1|/q^i}^2} \norm{F_f}_2^2 + \frac{q^{i-n}}{q-1} \le \frac{4(q-1)q^{\beta n}}{q^{2k-i}\Gamma^2} + \frac{q^{i-n}}{q-1} \eperiod
    \]
    where the second inequality comes from $q^{k-i}\Gamma = q^{\beta n - i + 2\eta n / 3}$ tightly dominating the subtraction, $|F_f(0)-1|/q^i \le q^{\beta n - i}$, and substituting $\norm{F_f}_2^2 \le q^{\beta n}$. Then we use union bound for $S_{\textrm{$F$-smooth}}(f)$:
    \begin{align*}
        \PR{\neg S_{\textrm{$F$-smooth}}(f)} &\le \sum_{i=0}^k \inparen{ \frac{4(q-1)q^{\beta n}}{q^{2k-i}\Gamma^2} + \frac{q^{i-n}}{q-1} } \\
        &\le \frac{4 q^{\beta n + 1}}{q^k \Gamma^2} + \frac{q^{k-n+1}}{(q-1)^2} \\
        &= 4q \cdot q^{-\eta n / 3} + \frac{q^{k-n+1}}{(q-1)^2} \eperiod
    \end{align*}

We turn to bound $\PR{S_{\textrm{smooth}}(f_z)}$. Let $t = \theta n$ for $\theta = \eta/12$. By the hypothesis $\eta n \ge 240\log_q n$, we have $t = \eta n/12 \ge 20\log_q n$, and $t = \eta n/12 \le Rn = k$ since $\eta < R$. Hence the condition $20\log_q(n) \le t \le k$ required by \cref{lem:fSmoothOrIncreasing,lem:fSmooth} is satisfied. For each $z \in \F_q^n$, define $S_{\textrm{smooth}}(f_z)$ as the event that the sequence is $(f_z, \Gamma')$-smooth, where
    \[ 
        \Gamma' = q^t \max\inset{\norm{f}_\infty q^{-k}, \sqrt{\Gamma}} \eperiod 
    \]
    Because $F_{f_z} = F_f$, the event $S_{\textrm{$F$-smooth}}(f_z)$ is equivalent to $S_{\textrm{$F$-smooth}}(f)$. By \cref{lem:fSmooth},
    \begin{equation*}
        \PR{S_{\textrm{$F$-smooth}}(f) \cap \neg S_{\textrm{smooth}}(f_z)} \le q^{n - 0.1t^2} \eperiod
    \end{equation*}
    
    We can bound the probability that any translation fails to be smooth by conditioning on the global smoothness of $F_f$:
    \begin{align*}
        \PR{\exists z \in \F_q^n ~~ \neg S_{\textrm{smooth}}(f_z)} 
        &\le \PR{\inparen{\exists z \in \F_q^n ~~ \neg S_{\textrm{smooth}}(f_z)} \cap S_{\textrm{$F$-smooth}}(f)} + \PR{\neg S_{\textrm{$F$-smooth}}(f)} \\
        &\le \sum_{z \in \F_q^n} \PR{\neg S_{\textrm{smooth}}(f_z) \cap S_{\textrm{$F$-smooth}}(f)} + \PR{\neg S_{\textrm{$F$-smooth}}(f)} \\
        &\le q^n \cdot q^{n - 0.1t^2} + 4q \cdot q^{-\eta n / 3} + \frac{q^{k-n+1}}{(q-1)^2} \\
        &\le q^{2n - \frac{\eta^2 n^2}{1440}} + 4q \cdot q^{-\eta n / 3} + \frac{q^{k-n+1}}{(q-1)^2} \eperiod
    \end{align*}

    To conclude, calculate $\Gamma'$. Since $\norm{f}_\infty = q^{\beta n}$, yielding $\norm{f}_\infty q^{-k} = q^{\beta n - Rn} = q^{-\eta n}$. We then get $\sqrt{\Gamma} = q^{-\eta n / 6}$. Substituting these bounds:
    \begin{align*}
        \Gamma' &\le q^{\theta n} \max\inset{q^{-\eta n}, q^{-\eta n / 6}} \\
        &= q^{\frac{\eta n}{12}} \cdot q^{-\frac{\eta n}{6}} \\
        &= q^{-\frac{\eta n}{12}} \eperiod
    \end{align*}
    By setting $\eps = \eta / 12 > 0$, we achieve $\max_{z \in \F_q^n} \inabs{\delta_{C_k, f_z}} \le q^{-\eps n}$, which completes the proof.

\end{proof}

\subsubsection{Proof of \cref{thm:mainGeneral}}
We are now ready to prove \cref{thm:mainGeneral}. The proof analyzes a sequence of functions $f^{(0)}, \dots, f^{(d)}$ generated by repeated averaged convolution, starting with $f^{(0)} = f$. Because each convolution squares the concentration parameter, the top iterate $f^{(d)}$ is highly Fourier-concentrated, allowing its smoothness to be bounded by the code's dual-typicality. We then work backward, inductively applying \cref{lem:fSmooth} to transfer this smoothness step-by-step from $f^{(m)}$ down to the base function $f^{(0)}$.
\begin{proof}[Proof of \cref{thm:mainGeneral}] 

    Let $S_\textrm{typical}$ denote the event that $C_0,\dots,C_k$ is $\gamma$-dual-typical where $\gamma = q^{\frac {\eta n}{100}}$ and $k = Rn$. By \cref{lem:typical}, 
    $$\PR{S_{\textrm{typical}}} \ge 1-\frac{n(k+1)}{\gamma} = 1 - n(Rn+1)q^{-\frac{\eta n}{100}} \eperiod$$

    Choose $d = \max\inset{0, \left\lceil \log_2 \left( \frac{\log\inparen{(q-1)/(q^{\eta/100}-1)}}{\log(1/\alpha)} \right) \right\rceil}$, so that $1+(q-1)\alpha^{2^d} \le q^{\frac{\eta}{100}}$. 
    
    Fix $f\in \cF$. We claim that 
    \begin{equation}\label{eq:fProbablyGood}
        \PR{S_{\textrm{typical}}\textrm{ and }\inabs{\delta_{C,f}} > q^{-\eps n}} \le d \cdot q^{n - 0.1t^2}\eperiod
    \end{equation}
    Observe that \cref{eq:fProbablyGood} implies the theorem since
    \begin{align*}
    \PR{\exists f\in  \cF \mid \inabs{\delta_{C,f}} > q^{-\eps n}} &\le \PR{\textrm{not } S_{\textrm{typical}}} + \PR{S_{\textrm{typical}}\textrm{ and } \exists f\in \cF ~~ \inabs{\delta_{C,f}} > q^{-\eps n}} \\
    &\le \PR{\textrm{not } S_{\textrm{typical}}} + \sum_{f\in \cF} \PR{S_{\textrm{typical}}\textrm{ and } \inabs{\delta_{C,f}} > q^{-\eps n}} \\
    &\le n(Rn+1)q^{-\frac{\eta n}{100}} + |\cF|\cdot d \cdot q^{n - 0.1t^2}
    \end{align*}

    We turn to proving \cref{eq:fProbablyGood}. 
    Define the sequence $f^{(0)} := f$, and $f^{(m)} := F_{f^{(m-1)}}$ for $m \ge 1$. Applying \cref{lem:concentration_F_f} we get that $f^{(m)}$ is $\alpha^{2^m}$-Fourier-concentrated. 

    Assuming $S_{\textrm{typical}}$, \cref{lem:ConcentratedSmooth} implies $C_0,\dots,C_k$ is $(f^{(d)}, \Gamma_d)$-smooth, with:
    $$ \Gamma_d = q^{-k} \cdot \gamma \cdot \inparen{1 + (q-1)\alpha^{2^d}}^n \le q^{-k + \frac{\eta n}{100} + \frac{\eta n}{100}} = q^{-k + \frac{\eta n}{50}} \eperiod $$
    Define the sequence $\Gamma_d \dots \Gamma_0$ where for all $1 \le m\le d$:
    \begin{equation}\label{definition_G}
        \Gamma_{m-1}=q^t\max\left\{\|f^{(m-1)}\|_\infty q^{-k},\sqrt{\Gamma_m}
\right\}.
    \end{equation}

    Define the event $T_{\textrm{smooth}}(f,i)$ as $C_0 \dots C_k$ being $(f^{(i)},\Gamma_i)$-smooth. Let $T_{\textrm{smooth}}(f)$ be the event that $T_{\textrm{smooth}}(f,i)$ holds for all $0\le i\le d$. Applying \cref{lem:fSmooth} to $f^{(m-1)}$ with $\Gamma = \Gamma_{m}$ yields
    $$ \PR{T_{\textrm{smooth}}(f,m) \cap \neg T_{\textrm{smooth}}(f,m-1)} \le q^{n - 0.1t^2} \eperiod $$

    Therefore, 
    $$\PR{S_{\textrm{typical}}\textrm{ and }\inabs{\delta_{C,f}} > \Gamma_0} \le \PR{S_{\textrm{typical}}\textrm{ and \underline{not} }T_{\textrm{smooth}}(f,0)}\eperiod$$

    Suppose that the event on the right-hand side holds, and let $0\le m\le d$ be the maximum index for which $T_{\textrm{smooth}}(f,i)$ \emph{does not} hold. Recall that $m\ne d$ since $S_{\textrm{typical}}$ implies $T_{\textrm{smooth}}(f,d)$. Thus,
    \begin{align*}
        \PR{S_{\textrm{typical}}\textrm{ and \underline{not} }T_{\textrm{smooth}}(f,0)} 
        &\le \PR{\exists~ 0\le m\le d-1 ~~ T_{\textrm{smooth}}(f,m+1) \textrm{ and \underline{not} } T_{\textrm{smooth}}(f,m)} \\
        &\le \sum_{m=0}^{d-1} \PR{T_{\textrm{smooth}}(f,m+1) \textrm{ and \underline{not} } T_{\textrm{smooth}}(f,m)} \\
        &\le \sum_{m=1}^d q^{n - 0.1t^2} \\
        &= d \cdot q^{n - 0.1t^2} \eperiod
    \end{align*}

    Now, we bound the error term $\Gamma_0$. By definition, $C_k$ is $(f^{(0)}, \Gamma_0)$-smooth. Setting $i=k$, we obtain:
    \begin{align*}
        \inabs{\delta_{C_k, f}} \le q^{k-k} \cdot \Gamma_0 = \Gamma_0 \eperiod
    \end{align*}

Unfolding the recurrent definition \cref{definition_G} and setting $A_j=\norm{f^{(j)}}_\infty\cdot q^{-k}$:
    \begin{align*}
        \Gamma_0 &= q^t \max\inset{A_0, \Gamma_1^{1/2}} \\
            &= q^t \max\inset{A_0, \inparen{q^t \max\inset{A_1, \Gamma_2^{1/2}}}^{1/2}} \\
            &\,\,\,\vdots \\
            &\le q^{2t} \cdot \max \inset{ \Gamma_d^{2^{-d}}, \max_{0 \le j \le d-1} \inparen{\norm{f^{(j)}}_\infty q^{-k}}^{2^{-j}} } \eperiod \numberthis \label{eq:Gamma0Max}
    \end{align*}

    We bound the first term of the maximum on the right-hand side of \cref{eq:Gamma0Max}. Recall that $\Gamma_d \le q^{-Rn + \frac{\eta n}{50}}$. Thus,
    \begin{align*}
        q^{2t} \cdot \Gamma_d^{2^{-d}} &\le q^{2t} \cdot \inparen{q^{-Rn + \frac{\eta n}{50}}}^{2^{-d}} \\
        &= q^{-n \inparen{R 2^{-d} - \frac{\eta 2^{-d}}{50} - \frac{2t}{n}}} \eperiod\numberthis \label{eq:gamma0FirstTerm}
    \end{align*}
    
    We turn to bound the other terms in the right-hand side of \cref{eq:Gamma0Max}. We bound the $L_\infty$ norm of $f^{(j)}$. Recall that $f^{(j)} = F_{f^{(j-1)}} = f^{(j-1)} * \check{f}^{(j-1)}$, and so
    \begin{align*}
    \norm{f^{(j)}}_\infty &= \norm{f^{(j-1)} * \check{f}^{(j-1)}}_\infty \\
    &= \max_w q^{-n} \sum_{x \in \F_q^n} f^{(j-1)}(x)\check{f}^{(j-1)}(w-x) \\
    &= \max_w q^{-n} \sum_{x \in \F_q^n} f^{(j-1)}(x)f^{(j-1)}(x-w) \\
    &\le \max_w q^{-n} \norm{f^{(j-1)}}_\infty \sum_{x \in \F_q^n} f^{(j-1)}(x) \\
    &= q^{-n} \norm{f^{(j-1)}}_\infty \inparen{q^n \norm{f^{(j-1)}}_1} \\
    &= \norm{f^{(j-1)}}_\infty \eperiod
    \end{align*}
    Then we get a bound on all $\norm{f^{(j)}}_\infty$:
    $$ \norm{f^{(j)}}_\infty \le \norm{f^{(j-1)}}_\infty \le \dots \le \norm{f^{(0)}}_\infty = \norm{f}_\infty \eperiod $$
    
    Therefore, for all $j \ge 0$, we have $\norm{f^{(j)}}_\infty \le \norm{f}_\infty = q^{\beta n}$. Substituting this into our maximum yields:
    $$ \norm{f^{(j)}}_\infty q^{-k} \le q^{\beta n} q^{-Rn} = q^{(\beta - R)n} = q^{-\eta n} \eperiod $$
    Substituting this bound back into \cref{eq:Gamma0Max}:
    \begin{align*}
        q^{2t} \cdot \max_{0 \le j \le d-1} \inparen{\norm{f^{(j)}}_\infty q^{-k}}^{2^{-j}} 
        &\le q^{2t} \cdot \max_{0 \le j \le d-1} \inparen{q^{-\eta n}}^{2^{-j}} \\
        &= q^{2t} \cdot q^{-\eta n 2^{-(d-1)}} \\
        &= q^{-n \inparen{\eta 2^{-d+1} - \frac{2t}{n}}} \eperiod \numberthis \label{eq:gamma0SecondtTerm}
    \end{align*}
    
    Set $\eps_1 := R 2^{-d} - \frac{\eta 2^{-d}}{50}$ and $\eps_2 := \eta 2^{-d+1}$.
    Finally, letting $\eps = \min(\eps_1, \eps_2) - \frac{2t}{n}$, and noting by our assumption on $t$ that $\eps > 0$, we have by \cref{eq:Gamma0Max,eq:gamma0FirstTerm,eq:gamma0SecondtTerm}:
    \begin{align*}
        \Gamma_0 &\le \max\inset{q^{-n\inparen{\eps_1 - \frac{2t}{n}}}, q^{-n\inparen{\eps_2 - \frac{2t}{n}}}} \\
        &\le q^{-\eps n} \eperiod
    \end{align*}

    Therefore, assuming $S_{\textrm{typical}}$, the event $T_{\textrm{smooth}}(f,0)$ guarantees $\inabs{\delta_{C_k, f}} \le q^{-\eps n}$. Because $T_{\textrm{smooth}}(f,0)$ fails with probability at most $d \cdot q^{n - 0.1t^2}$, we have established \cref{eq:fProbablyGood}, which completes the proof.
\end{proof}

\subsubsection{Proof of \cref{cor:SimplerMain}}\label{sec:corProof}
Finally, we use \cref{thm:mainGeneral} to prove the simplified version given in  \cref{cor:SimplerMain}.

\begin{proof}[Proof of \cref{cor:SimplerMain}]
Let $d$ be the auxiliary integer from \cref{thm:mainGeneral}, and set $\rho_*:=2^{-d}$.
We first show that $\rho_*\ge\lambda$.

Write
\[
    B:=\log(1/\alpha)
    \quad\text{and}\quad
    L:=\log\left(\frac{q-1}{q^{\eta/100}-1}\right).
\]
We claim that
\[
    \rho_*\ge \min\left\{1,\frac{B}{2L}\right\}.
\]
Indeed, if $L\le B$, then $d=0$ and $\rho_*=1$.
Otherwise, $d=\lceil\log_2(L/B)\rceil$, so $2^d\le 2L/B$, and hence $\rho_*=2^{-d}\ge B/(2L)$.

Since
\[
    q^{\eta/100}-1=e^{\eta\log q/100}-1\ge \frac{\eta\log q}{100},
\]
we have
\[
    \frac{q-1}{q^{\eta/100}-1}
    \le \frac{100(q-1)}{\eta\log q}
    \le \frac{100q}{\eta}.
\]
Thus $L\le\log(100q/\eta)$, and therefore
\[
    \rho_*
    \ge
    \min\left\{1,\frac{\log(1/\alpha)}{2\log(100q/\eta)}\right\}
    =\lambda.
\]

Apply \cref{thm:mainGeneral} with $t=\theta n$.
Since $\eta\le 50\beta/51$, we have $\beta\ge 51\eta/50$, and hence
\[
    2^{-d}\left(R-\frac{\eta}{50}\right)
    =\rho_*\left(\beta+\frac{49}{50}\eta\right)
    \ge 2\eta\rho_*.
\]
Therefore,
\[
    2^{-d}\min\left\{R-\frac{\eta}{50},2\eta\right\}=2\eta\rho_*.
\]
Since $\theta<\eta\lambda\le\eta\rho_*$, the condition on $t$ in \cref{thm:mainGeneral} is satisfied.

The theorem gives $\eps=2\eta\rho_*-2\theta$, so for every $f\in\cF$,
\[
    \inabs{\delta_{C,f}}
    \le q^{-2(\eta\rho_* -\theta)n}
    \le q^{-2(\eta\lambda-\theta)n}.
\]

It remains to simplify the failure probability.  By \cref{thm:mainGeneral},
this probability is at most
$$n(Rn+1)q^{-\eta n/100}+|\cF|d q^{n-\theta^2n^2/10}\eperiod$$
Since \(R\le 1\) and \(\eta\le 50/51\), the first term is at most
\(q^{-\eta n/200}\) whenever \(n\ge c\eta^{-1}\log_q(1/\eta)\), after
increasing the universal constant \(c\).  For the second term, the definition
of \(d\), together with \(q^{\eta/100}-1\ge \eta\log q/100\) and
\(\log(1/\alpha)\ge 1-\alpha\), gives
$$d\le c'\log(q/(\eta(1-\alpha)))$$ for a universal constant \(c'>0\).
Thus the second lower bound on \(n\) in the statement implies
\(\log_q d\le \theta^2n^2/20\), again after increasing \(c\), and hence
\(d q^{n-\theta^2n^2/10}\le q^{n-\theta^2n^2/20}\).  Therefore the failure
probability is at most
\[
    q^{-\eta n/200}+|\cF|q^{n-\theta^2n^2/20},
\]
as claimed.
\end{proof}

\section{$\alpha$-Concentration of Combinatorial Rectangles}

In this section, we show $\alpha$-Fourier-concentration for normalized indicator functions of combinatorial rectangles, which are the subsets relevant to zero-error list-recovery, or list-recovery from erasures. In full generality, we fix subsets $S_1,\dots,S_n \subseteq \F_q$, and put $\ell_i := |S_i|$ for all $i \in [n]$. This argument will only apply to the special case that $q=p$, a prime integer. Some such restriction is necessary: for example, if $q = 2^t$, if each $S_i$ is an $\F_2$-subspace of $\F_q$, then it is not hard to see that any such $\alpha$-Fourier-concentration is not possible.\footnote{One could, for example, weaken the requirement to being that $\mathrm{char}(q) > \max_i \ell_i$; however, we do not pursue that here.} 

The main tool which we will use is the following bound on sums over subsets of $p$-th roots of unity, proved by Benhamouda et al~\cite[Lemma~3.11]{BDIR21}. Below, $\omega_p = e^{2\pi i/p}$ denotes a primitive $p$-th root of unity. 

\begin{lemma}[Bound on subset sums] \label{lem:bound-on-subset-sums}
    Let $S \subseteq \F_p$ be a subset of size $\ell$. Then 
    \[
        \left|\frac{1}{\ell}\sum_{x \in S}\omega_p^{x}\right| \leq \frac{\sin(\pi \ell/p)}{\ell\,\sin(\pi/p)} \eperiod
    \]
\end{lemma}

In fact, the proof establishes that the sum is maximized if the sum is over ``consecutive'' $p$-th roots of unity, i.e., it is of the form $\left|\sum_{j=0}^{\ell-1}\omega_p^{j+i}\right|$ for some $i \in \{0,1,\dots,p-1\}$. This is predictable, upon visualizing the $p$-th roots of unity as lying on the radius one disk in the complex plane: to maximize the magnitude of the sum, all the $p$-th roots of unity should ``as much as possible'' point in the same direction. Additionally we remark that the same result would hold for any $m$-th root of unity and any subset of $\Z_m$, for any positive integer $m$.

We can now easily prove the requisite $\alpha$-Fourier-concentration. 

\begin{lemma} [$\alpha$-Fourier-concentration of rectangles]\label{lem:zero-error-con}
    Let $p$ be a prime, $n \in \N$, and let $S_1 ,\dots,S_n \subseteq \F_p$ be subsets, of size $\ell_1,\dots,\ell_n$, respectively. Let $f = \frac{p^n}{\ell_1\cdots \ell_n}1_{S_1 \times \cdots \times S_n}$ denote the normalized indicator for the combinatorial rectangle $S_1 \times \cdots \times S_n$. Then, $f$ is $\alpha$-Fourier-concentrated for 
    \[
        \alpha :=\max_{i \in [n]}\left\{ \frac{\sin(\pi\ell_i/p)}{\ell_i\, \sin(\pi/p)} \right\} \eperiod
    \]
    If $2\leq \ell_i\leq p$ for all $i\in[n]$, then $\alpha\leq \cos(\pi/p)<1$.
\end{lemma}

\begin{proof}
    Let $\xi \in \F_p^n$ be given, and let $T = \supp(\xi) = \{i \in [n] \suchthat \xi_i \neq 0\}$. We need to prove
    \[
        |\widehat f(\xi)| = \left|\Eover{X \sim \F_p^n}{f(X)\, \omega_p^{-\langle X,\xi\rangle}}\right| \leq \alpha^{|T|} \ .
    \]
    Firstly, as $f$ is normalized we have
    \begin{align}
        \Eover{X \sim \F_p^n}{f(X)\, \omega_p^{-\langle X,\xi\rangle}} = \Eover{X \sim S_1 \times \cdots \times S_n}{ \omega_p^{-\langle X,\xi\rangle}} \eperiod \label{eq:Fourier-changing-expectation}
    \end{align}
    As $S_1 \times \cdots \times S_n$ is a product set, the coordinates $X_1,\dots,X_n$ of $X$ are independent, so 
    \begin{align}
        \eqref{eq:Fourier-changing-expectation} = \prod_{i=1}^n \Eover{X_i \sim S_i}{\omega_p^{X_i\cdot\xi_i}}\eperiod \label{eq:before-support}
    \end{align}
    Now, for each $i \in [n]$ with $\xi_i=0$ the expectation is clearly 1. Hence,
    \begin{align*}    
    \eqref{eq:before-support} = \prod_{i\in T} \Eover{X_i \sim S_i}{\omega_p^{X_i\cdot\xi_i}}= \prod_{i\in T} \left(\frac{1}{\ell_i}\sum_{x_i \in S_i}\omega_p^{-x_i\cdot \xi_i}\right) = \prod_{i\in T} \left(\frac{1}{\ell_i}\sum_{x_i \in -\xi_i \cdot S_i}\omega_p^{x_i}\right) \eperiod 
    \end{align*}
    Taking absolute values and applying \Cref{lem:bound-on-subset-sums} allows us to upper bound $|\widehat f(\xi)|$ as
    \begin{align*}
        |\widehat f(\xi)| = \prod_{i \in T}\left|\frac{1}{\ell_i}\sum_{x_i \in -\xi_i \cdot S_i}\omega_p^{x_i}\right| \leq \prod_{i \in T}\frac{\sin(\pi\ell_i/p)}{\ell_i\sin(\pi/p)} \leq \alpha^{|T|} \ ,
    \end{align*}
    where the last inequality recalls we defined $\alpha$ as $\max_{i \in [n]}\left\{ \frac{\sin(\pi\ell_i/p)}{\ell_i\, \sin(\pi/p)} \right\}$.

    To see the last claim about the $\cos(\pi/p)$ upper bound on $\alpha$ when $2\leq \ell_i\leq p$ for all $i$, it suffices to show that $\frac{\sin(\pi \ell/p)}{\ell \sin(\pi/p)}$ is decreasing in $\ell$, since $\frac{\sin(2\pi/p)}{2\sin(\pi/p)}=\cos(\pi/p)$.
    To see that $\frac{\sin(\pi \ell/p)}{\ell \sin(\pi/p)}$ is decreasing in $\ell$, note that this is equivalent to
    \begin{equation*}
        \frac{\sin(\pi (\ell+1)/p)}{\ell+1}\leq \frac{\sin(\pi \ell/p)}{\ell},
    \end{equation*}
    and that
    $\sin(x)/x$ is decreasing for $x\in(0,\pi]$.
\end{proof}

Finally, in practice what we need to lower bound is $\log(1/\alpha)$. For $\alpha = \cos\frac{\pi}{p}$, we record below the estimate $\log\frac{1}{\cos\pi/p} \geq \Omega(1/p^2)$. 

\begin{lemma} \label{lem:lower-bound-on-alpha}
    There is a constant $c>0$, such that
    \[
        \log\frac{1}{\cos(\pi/p)} = \log \sec(\pi/p) \geq \frac{\pi^2}{8p^2\ln 2} 
    \]
    for all $p > 2$.
\end{lemma}

\begin{proof}
    For $|z| < \frac{\pi}{2}$, the Taylor series for $\sec z$ about $z=0$ is given by
    \[
    \sec z = \sum_{j=0}^\infty
    \frac{(-1)^jE_{2j}}{(2j)!}z^{2j}.
    \]
    where the $E_{2j}$ are the Euler numbers~\cite[p.~75, eq.~4.3.68]{AS48}. Since all these coefficients are positive (the Euler numbers $E_{2j}$ are positive iff $j$ is even~\cite[pp.~804--805]{AS48}), we in particular have $\sec z \geq 1+\frac{z^2}{2}$. Hence, since $\frac{\pi}{p} < \frac{\pi}{2}$ by the assumption on $p$, 
    \[
        \log\left(\sec\left(\tfrac{\pi}{p}\right)\right) \geq \log\left(1+\frac{\pi^2}{4p^2}\right) \eperiod
    \]
    
    Then, since $\ln(1+y) \geq \frac{y}2$ for all $y\in [0,1]$, as $p>2$ also guarantees $\frac{\pi^2}{4p^2} \leq 1$, we find
    \[
        \log\left(1+\frac{\pi^2}{4p^2}\right) \geq \frac{1}{\ln 2}\ln\left(1+\frac{\pi^2}{4p^2}\right) \geq \frac{1}{\ln 2} \cdot \frac12 \cdot \frac{\pi^2}{4p^2} \ ,
    \]    
    so we get the claimed lower bound.
\end{proof}

\section{Improved Resilience of Linear Ramp Secret Sharing Schemes against Local Leakage Functions} \label{sec:leakage-resilience}

In this section we apply our results to the study of leakage-resilient secret sharing schemes.
Roughly speaking, an $(n,t)$-threshold secret sharing scheme encodes a secret $s$ into a set of $n$ shares $S_1,\dots,S_n$, each assigned to one of $n$ parties, with the property that every subset of at least $t$ parties can pool their shares and perfectly reconstruct $s$, while any subset of fewer than $t$ parties learns no information about the secret.
It is not hard to see that in any $(n,t)$-threshold secret sharing scheme the length (in bits) of each share must be at least the length of the secret.

To obtain small share size compared to the number of parties $n$, we must weaken the requirements above by allowing a gap between the \emph{reconstruction threshold $\trec$} (with the property that any subset of parties of size at least $\trec$ can reconstruct the secret) and the \emph{privacy threshold $\tpriv$} (with the property that any subset of parties of size at most $\tpriv$ learns no information about the secret)~\cite{CCX13}.
This corresponds to the notion of \emph{ramp} secret sharing.
The notion of $(n,t)$-threshold secret sharing corresponds to the special case where $\trec=t$ and $\tpriv=\trec-1=t-1$.
We present a formal definition below.

\begin{definition}[Ramp secret sharing scheme]
    An \emph{$(n,\trec,\tpriv)$-ramp secret sharing scheme} with secret space $\cS$ is a pair of algorithms $(\Share,\Rec)$ with $\Share$ a randomized algorithm that receives as input a secret $s\in \cS$ and outputs a share $S_i$ for the $i$-th party, $i\in[n]$, and $\Rec$ a deterministic algorithm that receives a subset $T\subseteq[n]$ of parties and the shares $(S_i)_{i\in T}$ corresponding to parties in $T$ and outputs a guess for the secret satisfying the following properties:
    \begin{itemize}
        \item \textbf{Privacy:} For any subset $T\subseteq[n]$ such that $|T|\leq \tpriv$ and any two secrets $s,s'\in\cS$, the random variables $\Share_T(s)$ and $\Share_T(s')$ (i.e., the shares of $s$ and $s'$ corresponding to parties in $T$, respectively) are identically distributed.

        \item \textbf{Reconstruction:} For any subset $T\subseteq[n]$ such that $|T|\geq \trec$ and any secret $s\in\cS$ we have
        \begin{equation*}
            \Pr[\Rec(\Share(s)_T,T)=s]=1,
        \end{equation*}
        with the probability taken over the random coins of $\Share$.
    \end{itemize}

    When $\tpriv=\trec-1$, we say that $(\Share,\Rec)$ is an \emph{$(n,\trec)$-threshold secret sharing scheme}.
\end{definition}

\emph{Linear} secret sharing schemes are especially important objects in cryptography (for example, in secure multiparty computation).
Informally, these are secret sharing schemes that operate over some finite field $\F_q$, and such that local linear combinations of shares of two secrets yield a correct sharing of the linear combination of the secrets.
More precisely, we have the following definition.
\begin{definition}[Linear ramp secret sharing scheme]
    We say that a ramp secret sharing scheme $(\Share,\Rec)$ is \emph{linear} if the following holds:
    \begin{itemize}
        \item The $\Share$ procedure maps a secret $s\in\F_q$ to shares $S_i\in\F_q^m$ for some integer $m$;

        \item For every $s\in\F_q$, the tuple $\Share(s)$ is uniformly distributed over an affine subspace of $(\F_q^m)^n$;

        \item If $S=(S_1,\dots,S_n)=\Share(s)$ and $S'=(S'_1,\dots,S'_n)=\Share(s')$ for any two secrets $s,s'\in\F_q$, then $\alpha S+\beta S'$ is identically distributed to $\Share(\alpha s+\beta s')$, for any $\alpha,\beta\in\F_q$.
    \end{itemize}
\end{definition}

We remark that ramp secret sharing schemes are the best we can hope for in regimes where $n$ is much larger than $q$, which is our main focus here.
For example, we must always have $\trec-\tpriv\geq \frac{n-\tpriv+1}{q}$~\cite{CCX13}, regardless of the size of the secret and of linearity.
In particular, if $q$ is held constant and $\tpriv\leq \trec\leq cn$ for some constant $c<1$, then necessarily $\trec-\tpriv =\Omega(n)$.

\subsection{Leakage-Resilience of Linear Secret Sharing Schemes}

Motivated by side-channel attacks in cryptography, Benhamouda, Degwekar, Ishai, and Rabin \cite{BDIR21} studied a different mode of attack on secret sharing schemes.
They considered an attacker that learns some side information about \emph{every} share (as opposed to an attacker that learns full information about a small subset of shares).
The goal is then to ensure that the secret remains hidden from such an adversary, in which case we call the secret sharing scheme \emph{leakage-resilient}.
Studying this requires specifying the leakage functions that the adversary can apply to the shares.
We first provide a general definition of leakage-resilience against an abstract family of leakages, and then discuss the important special case on which we focus (bounded local leakage).

\begin{definition}[Statistical distance]
    The \emph{statistical distance} between two discrete random variables $X$ and $Y$ supported on a finite set $\cS$, denoted by $\SD(X,Y)$, is defined as
    \begin{equation*}
        \SD(X,Y)=\max_{A\subseteq \cS}|\Pr[X\in A]-\Pr[Y\in A]| = \frac{1}{2}\sum_{s\in \cS}|\Pr[X=s]-\Pr[Y=s]|.
    \end{equation*}
    We may write $X\approx_\eps Y$, and say that $X$ and $Y$ are \emph{$\eps$-close}, when $\SD(X,Y)\leq \eps$.
\end{definition}

\begin{definition}[Leakage-resilient ramp secret sharing scheme]
    An \emph{$(n,\trec,\tpriv,\cG,\eps)$-leakage-resilient ramp secret sharing scheme} is an $(n,\trec,\tpriv)$-ramp secret sharing scheme with the following additional property:
    \begin{itemize}
        \item \textbf{Leakage-resilience against $\cG$:} For any leakage function $g\in\cG$ and
        any two secrets $s,s'\in\cS$ with corresponding shares $S=(S_1,\dots,S_n)=\Share(s)$ and $S'=(S'_1,\dots,S'_n)=\Share(s')$, we have
        \begin{equation*}
            g(S)\approx_{\eps} g(S').
        \end{equation*}
    \end{itemize}
    When the remaining parameters are not relevant, we may also say that $(\Share,\Rec)$ is \emph{$(\cG,\eps)$-leakage-resilient}.
\end{definition}

We will be interested in linear ramp secret sharing schemes with $m=1$, i.e., whose secrets and shares are elements of $\F_q$.
The most well-studied leakage model for secret sharing is \emph{bounded local leakage}, which corresponds to independently applying a leakage function with bounded output length to each share.
For simplicity, and because this already yields a highly non-trivial problem, we focus on the special case where leakage functions have $1$-bit output. 
More precisely, this corresponds to the family $\cG$ containing all functions $g=(g_1,\dots,g_n)$ of the form
\begin{equation*}
    g(S) = (g_1(S_1),\dots,g_n(S_n)),
\end{equation*}
for possibly distinct bounded-output functions $g_i:\F_q\to\bits$. 
Secret sharing schemes that are $(\cG,\eps)$-leakage-resilient for this family $\cG$ are commonly called \emph{$\eps$-locally leakage-resilient}, and we use this name too. If we wish to emphasize the other parameters, we will write ``$(n,\trec,\tpriv,\eps)$-locally leakage-resilient''.

\subsection{Prior Work on Leakage-Resilience of Linear Secret Sharing Schemes and Our Result}

The study of the leakage-resilience of linear threshold secret sharing schemes has attracted significant attention ever since its introduction in~\cite{BDIR21}.
Establishing the local leakage-resilience of a secret sharing scheme becomes easier when the reconstruction threshold grows.
Given this, one basic question in this area is to pinpoint the smallest reconstruction threshold (as a function of the number of parties $n$) that allows for resilience against bounded local leakage.

\paragraph{General bounded local leakage.}
The state-of-the-art result for locally leakage-resilient linear threshold secret sharing schemes is due to Sun and Wootters~\cite{SW26b} and Nguyen~\cite{Ngu26}, who showed that Shamir secret sharing is locally leakage-resilient for any threshold $t\geq (1/2-\delta)n$, with $\delta>0$ a small constant.
The previous best result was due to Maji, Paskin-Cherniavsky, Suad, and Wang~\cite{MPSW21}, who showed the existence of $\eps$-locally leakage-resilient linear $(n,t)$-threshold secret sharing schemes with threshold $t=(\frac{1}{2}+\eps)n$ and exponentially small $\eps$, for any $\eps>0$, over an exponentially large field of prime order $q=2^{\Theta(n)}$.\footnote{It is easy to see that local leakage-resilience cannot be achieved over fields of characteristic $2$~\cite{BDIR21}.}
This is achieved by considering secret sharing schemes induced by random linear codes, as we discuss in more detail below.
It has been conjectured that there exist locally leakage-resilient linear threshold secret sharing schemes with threshold $t=cn$ for any constant $c>0$, and that this holds even for Shamir's secret sharing~\cite{BDIR21}.

Motivated by the prospect of constructing leakage-resilient secret sharing schemes with small share size 
(for which a gap between the privacy and reconstruction thresholds is unavoidable, as discussed above), Tjuawinata and Xing~\cite{TX22} studied the leakage-resilience of linear ramp secret sharing schemes over constant-sized fields.
They use algebraic-geometric codes to obtain explicit constructions of such schemes.
Looking ahead, in the regime of our results where the field size $q$ is bigger than some sufficiently large constant, their main results~\cite[Theorems 1 and 2 and Corollaries VI.2 and VI.4]{TX22} only give non-trivial leakage-resilience error $\eps$ for reconstruction thresholds $\trec> n/2$.

Summarizing, the current state of affairs is that the existence of leakage-resilient linear threshold and ramp secret sharing schemes with reconstruction threshold
$\trec<0.49n$ (and a non-trivial gap between privacy and reconstruction thresholds in the case of ramp secret sharing) remains open.
Our result, which we state below in simplified form for constant field size, makes progress on this question whenever the field size does not grow too fast, in which case we must necessarily focus on ramp secret sharing.

\begin{theorem}[Simplified version of \Cref{thm:rampss-full}]\label{thm:rampss-short}
    Fix an arbitrary constant $R\in(0,1)$.
    Given $n\in\N$, let $q$ be a prime satisfying
    \begin{equation*}
        2^{800/R} \leq q \leq n^{0.19}.
    \end{equation*}
    Then, for all sufficiently large $n$ (depending only on $R$) there is an $(n,\trec,\tpriv,\eps)$-locally leakage-resilient ramp secret sharing scheme over $\F_q$ with
    \begin{equation*}
       Rn - \frac{1.01n}{\log q} \leq \tpriv < \trec \leq Rn + \frac{1.01n}{\log q}
    \end{equation*}
    and 
    \begin{equation*}
        \eps \leq 4q^{- \frac{R^2 n^{3/5}}{3200}}.
    \end{equation*}
\end{theorem}

Note that, in particular, the gap between reconstruction and privacy thresholds can be made $O(n/\log n)$ when $q\geq n^\gamma$ for some constant $\gamma\in(0,0.19)$.
The next sections are devoted to proving a more general version of \cref{thm:rampss-short}.
We begin by setting up the standard underlying framework for building linear ramp secret sharing schemes out of linear codes, which matches that used by \cite{MPSW21}.

\subsection{Linear Ramp Secret Sharing from Linear Codes}\label{sec:LCs-to-LSS}

Towards proving \cref{thm:rampss-short}, we will study a classical construction of linear secret sharing schemes from error-correcting codes due to Massey~\cite{massey1995some}, which we proceed to describe.
Let $C^+\subseteq\F_q^{n+1}$ be a linear code with corresponding generator matrix $G^+\in\F_q^{(n+1)\times(k+1)}$ in systematic form, which we assume to be of the form
\[
    G^+ = \begin{bmatrix}
        1 & 0 \\
        v & G
    \end{bmatrix} \ ,
\]
where $v \in \F_q^n$ and $G \in \F_q^{n \times k}$. 
Then, the \emph{Massey secret sharing scheme} associated with $C^+$ has secret space $\F_q$ and works as follows:
\begin{itemize}
    \item To share a secret $s\in\F_q$, choose $x\in\F_q^{k+1}$ by setting $x_0=s$ and sampling $x_1,\dots,x_k\sim\F_q$.
    Then, compute $c^+=(c^+_0,c^+_1,\dots,c^+_n)=G^+ x$, and set the $i$-th share to be $c^+_i$ for $i\in[n]$.
    Note that $c^+_0=s$.

    \item Given a subset of parties $T\subseteq[n]$ and the corresponding shares $(S_i)_{i\in T}$, find the unique codeword $c^+\in C^+$ such that $c^+_i=S_i$ for all $i\in T$ and output $s=c^+_0$.
\end{itemize}

It is not hard to see that the minimum distance and dual distance of $C^+$ respectively control the reconstruction and privacy thresholds of the resulting ramp secret sharing scheme, as formalized in the next standard lemma.
For completeness we provide a short proof in \cref{app:sec2}.

\begin{lemma}\label{lem:massey-prop}
    Let $C^+\subseteq\F_q^{n+1}$ be a linear code with distance $d$ and dual distance $d^\perp$.
    Then, the Massey secret sharing scheme associated with $C^+$ is an $(n,\trec=n-d+2,\tpriv=d^\perp-2)$-ramp secret sharing scheme.
\end{lemma}

\begin{remark}
    If $C^+$ is an MDS code of dimension $k$, then \cref{lem:massey-prop} yields a linear $(n,k)$-threshold secret sharing scheme.
    First,
    since $d=(n+1)-k+1=n-k+2$, we get that $\trec=n-d+2=k$.
    Second, since the dual of $C^+$ is also MDS with dimension $k^\perp=n+1-k$, we get that $d^\perp=(n+1)-k^\perp+1=k+1$, and so $\tpriv=d^\perp-2=k-1=\trec-1$.

    If the field size $q$ is sufficiently large compared to the block length $n$, then a random linear code will be MDS with high probability.
    This fact was leveraged by~\cite{MPSW21} to obtain locally leakage-resilient linear threshold secret sharing schemes from random linear codes over exponentially large fields.
\end{remark}

To study the local leakage-resilience of a Massey secret sharing scheme based on a linear code $C^+$ the following alternative perspective is useful.
Let $G\in\F_q^{n\times k}$ be the matrix obtained by removing the first column and first row of $G^+$.
Then, $G$ is the generator matrix of a code $C\subseteq\F_q^n$, which we call the \emph{reduced code} of $C^+$.
Denote the first column of $G^+$ by $v^+=(1,v_1,\dots,v_n)^\top$, and let $v=(v_1,\dots,v_n)^\top$.
Sharing a secret $s\in\F_q$ using the Massey secret sharing scheme associated with $C^+$ is equivalent to the following: sample $x\sim\F_q^k$, compute $c=(c_1,\dots,c_n)^\top=s\cdot v + Gx$, and set the $i$-th share to be $c_i$.
In other words, the shares of $s$ are generated by sampling a vector uniformly at random from the affine subspace $s\cdot v+C$.
This alternative perspective helps us see that the leakage-resilience of this secret sharing scheme is controlled by the discrepancy of the reduced code $C$, as formalized in the next lemma which already appears, at least implicitly, in prior work~\cite{BDIR21}. 
For completeness we provide a proof in \cref{app:sec2}.

\begin{lemma}\label{lem:LR-discrepancy-reduced}
    Let $(\Share,\Rec)$ be the Massey ramp secret sharing scheme associated to a linear code $C^+\subseteq \F_q^{n+1}$ with reduced code $C\subseteq\F_q^n$. Let $\cG$ be a collection of leakage functions closed under input-translations (i.e., if $g \in \cG$, then so is the function $x \mapsto g(x+z)$ for all $z \in \F_q^n$).\footnote{Note that this is satisfied by the family of local leakage functions.}
    Suppose that 
    \begin{equation*}
        \SD(g(C),g(U_n)) \leq \eps
    \end{equation*}
    for all tuples of leakage functions $g \in \cG$.
    Then, $(\Share,\Rec)$ is $(2\eps,\cG)$-leakage-resilient.
\end{lemma}

\subsection{Leakage Resilience of Massey Ramp Secret Sharing Schemes from Random Linear Codes}

We now combine the standard framework described in \cref{sec:LCs-to-LSS} with our new results to prove \cref{thm:rampss-short}. In fact, we will prove the stronger statement that a random linear code yields a linear ramp secret sharing scheme with the desired properties with high probability, and we will obtain a more general relationship between the various parameters.

We will make use of the concept of an \emph{induced} rectangle, obtained as follows. Given a local leakage function $g=(g_1,\dots,g_n):\F_q^n \to \{0,1\}^n$, we say a rectangle $S_1\times\cdots \times S_n$ is \emph{induced by $g$} if it is obtained as the preimage of some leakage bits. 
In other words, for some $b_1,\dots,b_n \in \{0,1\}$, we have $S_i = g^{-1}(b_i)$ for all $i \in [n]$. Note that any local leakage function partitions $\F_q^n$ into $2^n$ induced rectangles.

\subsubsection{Leakage-Resilience of Low-Rate Random Linear Codes}

We begin by proving our key lemma, which applies \cref{thm:mainGeneral} to the case of ``large'' rectangles. Formally, for a parameter $\gamma>0$ which is a bit smaller than the code's rate (and may be chosen arbitrarily small, at the cost of increasing $q$), we will first consider rectangles 
$S = S_1 \times \cdots \times S_n$ that are $\gamma$-large, defined as follows.

\begin{definition}[$\gamma$-large and $\gamma$-small combinatorial rectangles]
    Let $S = S_1 \times \cdots \times S_n \subseteq \F_q^n$ be a combinatorial rectangle. Given $\gamma>0$, we say that it is \emph{$\gamma$-large} if $|S| \geq q^{(1-\gamma)n}$. If it is \emph{not} $\gamma$-large, we call it \emph{$\gamma$-small}.  
\end{definition}

Before we present the formal argument for leakage-resilience, we give a brief overview.
By applying \Cref{cor:SimplerMain}, we can show that a random codeword lands in a large rectangle with approximately the correct probability (this is controlled by $|\delta_{f,C}|$, where $f$ is the normalized indicator function for the rectangle). To derive our main result proving leakage-resilience against all local leakage functions, we will observe that the uniform distribution already puts so little probability mass on $\gamma$-small rectangles that their contribution to the statistical distance is negligible. We will then couple this result with \cref{lem:LR-discrepancy-reduced} to show the existence of Massey ramp secret sharing schemes with reconstruction and privacy thresholds well below the $1/2$-barrier appearing in prior works.

One nuisance is that a rectangle induced by the local leakage function $g=(g_1,\dots,g_n)$ could have some sides of length 1. 
This causes an issue as such rectangles are not $\alpha$-concentrated for $\alpha<1$, preventing an immediate blackbox application of \Cref{thm:mainGeneral}. To account for this, we first note that $\gamma$-large rectangles cannot have more than $\gamma n$ sides of length $1$. 
Assuming that the dual distance of $C$ is greater than $\gamma n$ (which is very likely for a random linear code of rate slightly more than $\gamma$), we then find that the code shortened to the coordinates with side-lengths $>1$ is still distributed as a random linear code, and therefore we may apply \Cref{thm:mainGeneral} to this shortened code. Furthermore the ``$\delta$-values'' in the shortened code exactly equal the $\delta$-values from the original code, so if the $\delta$-values in the shortened code are controlled, the same is true for the original code. 

We now proceed to the formal proof.
The following corollary of \Cref{thm:mainGeneral}, which conveniently does not place any restriction on the relationship between $\beta$ and $\eta$, will be useful. We defer the proof to \Cref{app:sec2}.

\begin{corollary} \label{cor:LR-simple-large-rectangles}
    Consider a family $\cF$ of normalized indicator functions of combinatorial rectangles whose sides all have length greater than $1$.
    Set $\beta = \max_{f\in \cF} \frac 1n\cdot \log_q \norm f_\infty$.
    Fix a prime field $\F_q$ of size $3\leq q\leq \frac{n^{\lambda}}{\log n}$ for some constant $\lambda\in(0,1/5)$.
    Let $C\subseteq \F_q^n$ be a random linear code of rate $R = \beta + \eta\le 1$ for some constant $\eta > 0$.
    Then, for sufficiently large $n$, with probability at least
    \begin{equation*}
        1 - q^{-\frac{\eta n}{200}} - |\cF|q^{-9n^{1+\lambda}}
    \end{equation*}
    over the choice of $C$, we have
    \begin{equation*}
        |\delta_{C,f}|\le q^{-\frac{\eta^2 \cdot n^{1-2\lambda}}{400}}
    \end{equation*}
    for every $f\in \cF$.
\end{corollary}

We now show how to obtain a bound for \emph{all} $\gamma$-large combinatorial rectangles, even those that might have sides of length $1$ (and therefore technically fail to be $\alpha$-concentrated with $\alpha<1$).

\begin{lemma}\label{lem:LR-large-rectangles}
    Fix any $0<\eta,\gamma<1$ with $\gamma+\eta<1$, and fix $\lambda\in(0,1/5)$. Given $n \in \N$, let $q$ be a prime satisfying 
    \[
        \max\left(2^{\frac{1}{(1-\gamma)^{1+\lambda}}},2^{400/\eta}\right)\leq q\leq \frac{(n(1-\gamma))^{\lambda}}{\log (n(1-\gamma))}.
    \]
    Then, the following holds for all sufficiently large $n$. 
    Let $C \subseteq \F_q^n$ be a random linear code of rate $R = \gamma+\eta$. 
    Then, with probability at least 
    \begin{align}
        1 - 2^{n}q^{-\frac{\eta n}{200}} - q^{-8(1-\gamma)^{1+\lambda}n^{1+\lambda}} - q^{-\frac{n\eta}2} \geq 1-2q^{-\frac{\eta n}{400}} \label{eq:failure-prob-LR-large-rectangles}
    \end{align}
    over the choice of $C$,
    for all $\gamma$-large combinatorial rectangles $S = S_1 \times \cdots \times S_n$, 
    \[
        |\delta_{C,S}| \leq q^{-\frac{\eta^2n^{1-2\lambda}}{400}}
        \eperiod
    \]
\end{lemma}

\begin{proof}
    Let $S = S_1 \times \cdots \times S_n$ be a given $\gamma$-large combinatorial rectangle, i.e., it satisfies $|S| \geq q^{(1-\gamma)n}$. Let $T \subseteq [n]$ denote the subset of coordinates $i$ for which $|S_i| > 1$. 
    Observe that $|T| \geq (1-\gamma)n$.
    Then, letting $t=|\overline{T}|$, we get $t\leq \gamma n$.

    Let us now consider the code $C$ shortened to the coordinates in $T$. That is, define
    \[
        C_T:=\{c_T : c \in C,c_{\overline T}=0\} \subseteq \F_q^{|T|} \ ,
    \]
    where the notation $c_T=(c_i)_{i \in T}$ denotes the subvector obtained by just keeping the coordinates in $T$ (and similarly for $c_{\overline T}$).

    First, we recall the structure of its generator matrix. Let $G \in \F_q^{n \times k}$ be the generator matrix sampled for $C$ (where $k = Rn =(\gamma+\eta)n$), where we recall that for our probability model that means $G$ is uniformly random. Let $G_T$ (respectively, $G_{\overline T}$) denote the matrix consisting of just the rows in $T$ (respectively, $G_{\overline T}$). 
    Note that $G_T$ and $G_{\overline T}$ are uniform and independent.

    To obtain a generator matrix $C_T$, one can do the following. Firstly, let $K_{\overline T}:=\ker(G_{\overline T}) = \{x \in \F_q^k:G_{\overline T}x=0\}$, let $a=\dim(K_{\overline T})$, and let $A \in \F_q^{a \times k}$ be a matrix whose columns form a basis for $K_{\overline T}$ (in other words, it maps $\F_q^a$ injectively into $K_{\overline T}$). Then, we find that $G_TA \in \F_q^{a \times n-t}$ is a generator matrix for $C_T$. Furthermore, conditioned on the even that $G_{\overline T}$ has rank $|\overline T| = t \leq \gamma n$, we have $a = k-t$ and the matrix $G_T A$ is distributed uniformly at random over $\F_q^{(k-t) \times (n-t)}$. That is, $C_T$ is a random linear code according to our probability model. 

    Let now $\cF_T$ consist of all normalized indicator functions $f_S=\frac{q^n}{|S|}1_{S}$ of $\gamma$-large combinatorial rectangles $S=S_1\times \cdots\times S_n$ with $\{i \in [n]:|S_i|>1\}=T$. Observe $\sum_{T:|T| \geq (1-\gamma)n}|\cF_T| \leq 2^{qn} = q^{nq/\log_2 q}$.
    Fix some $f \in \cF_t$, and let $S$ denote the combinatorial rectangle it corresponds to. Let $a_i$ for $i \in \overline T$ be such that $S_i=\{a_i\}$.

    Let us first consider a (fixed) linear code $D \subseteq \F_q^n$ with the property that $\overline T$ does not support a nonzero dual codeword, i.e., for all  $x \in D^\perp\setminus \{0\}$ it holds that $\supp(x) := \{i \in [n]:x_i\neq 0\} \not\subseteq \overline T$. 
    Then,
    \[
        \PROver{c\sim D}{\forall i \in \overline{T}, c_i=a_i} = \frac{|\{c \in D:\forall i \in T, c_i=a_i\}|}{|D|} = q^{-|\overline T|} = q^{-t} = \PROver{x \sim \F_q^n}{\forall i \in \overline T, x_i=a_i}\eperiod
    \]
    Now, fix any codeword $c^* \in D$ with $c^*_i=a_i$ for all $i \in \overline T$, and let $S^* = S_1^* \times \cdots \times S_n^*$ be defined as $S_i^*:=S_i-c_i^* = \{a_i-c_i^*:a_i \in S_i\}$. Then, since $S^*_i=\{0\}$ for all $i\in\overline{T}$, observe that 
    \[
        \PROver{c\sim D}{\forall i \in T, c_i \in S_i\mid\forall i \in \overline T, c_i=a_i}=\PROver{c \sim D_T}{c \in S^*_T} \ ,
    \]
    where $S^*_T = \prod_{i \in T}S_i^*$. Thus, if $f_T^* = \frac{q^{n-t}}{|S_T^*|}1_{S_T^*}$ denotes the normalized indicator function for $S_T^*$ (viewed as a subset of $\F_q^{|T|} = \F_q^{n-t}$), we find
    \begin{align*}
        \delta_{D,f} &= \frac{\PROver{c \sim D}{c \in S}}{\PROver{x \sim \F_q^n}{x \in S}} - 1 = \frac{\PROver{c \sim D}{c \in S}}{\PROver{x \sim \F_q^n}{x \in S}} - 1 = \frac{\PROver{c \sim D_T}{c \in S_T^*}}{\PROver{x \in \F_q^{n-t}}{x \in S_T^*}} - 1 = \delta_{D_T,f_T^*} \ .
    \end{align*}
    
    Consequently, if $\cE_T$ denotes the event that our random linear code $C$ does not support a dual codeword on $\overline T$, we have 
    \[
        \PROver{C}{\delta_{C,f} = \delta_{C_T,f_T^*} \;\bigg|\; \cE_{\overline T}} = 1 \ .
    \]
    Let now $\cF_T^*$ consist of all functions of the form $f_T^*$ obtained from an $f \in \cF_T$ as above. We look to apply \cref{cor:SimplerMain} to $C_T$ with this $\cF_T^*$. Let $\tau := t/n$. Conditioned on the event $\cE_{\overline T}$, it follows that $C_T$ is a random linear code of rate 
    \[
        R_{T}:=\frac{k-t}{n-t} = \frac{R-\tau}{1-\tau} = \frac{\gamma+\eta-\tau}{1-\tau} = \frac{\gamma-\tau}{1-\tau} + \frac{\eta}{1-\tau}\eperiod
    \]
    Additionally, given any $f_T^* \in \cF_T^*$, supposing it is induced by a combinatorial rectangle $S_T^* = \prod_{i \in T}S_i^*$, since we have $|S_T^*| \geq q^{(1-\gamma)n} = q^{\frac{1-\gamma}{1-\tau}(n-t)}$, it follows that 
    \[
        \|f_T^*\|_\infty = \frac{|\F_q^{n-t}|}{|S_T^*|} \leq \frac{q^{n-t}}{q^{\frac{1-\gamma}{1-\tau} (n-t)}} 
    \]
    with equality for some $f_T^* \in \cF_T^*$, and so 
    \[
        \beta_T = \max_{f_T^* \in \cF_T^*} \frac 1{n-t}\log_q\|f_T^*\|_\infty = \frac{1}{n-t}\log_q \frac{q^{n-t}}{q^{\left(\frac{1-\gamma}{1-\tau}\right)(n-t)}}=1-\frac{1-\gamma}{1-\tau} = \frac{\gamma-\tau}{1-\tau} \eperiod
    \]
    Thus, we see that $C_T$ has (design) rate $\beta_T+\frac{\eta}{1-\tau}$. Put $\eta_T:= \frac{\eta}{1-\tau}$. Then, applying \Cref{cor:LR-simple-large-rectangles}, if $C_T^* \subseteq \F_q^{n-t}$ is a random linear code of rate $R_T$, we get that

    \begin{equation}\label{eq:UB-prob-fstar}
        \PROver{C_T^*}{\exists f_T^* \in \cF_T^* \text{ s.t. } |\delta_{C_T^*,f_T^*}|\geq q^{-\frac{\eta_T^2(n-t)^{1-2\lambda}}{400}}} \leq q^{-\frac{\eta_T (n-t)}{200}} + |\cF_T|q^{-9(n-t)^{1+\lambda}} \ .
    \end{equation}
    We use this to obtain a bound on the probability that some $\delta_{C,f}$ is large for some $f \in \cF$. 
    First, using that 
    \[
        \eta_T^2(n-t)^{1-2\lambda} = \frac{\eta^2}{(1-\tau)^2}(1-\tau)^{1-2\lambda}n^{1-2\lambda} \geq \eta^2n^{1-2\lambda} \ ,
    \]
    we obtain the uniform bound
    \begin{align}\label{eq:uniform-UB-Delta}
        q^{-\frac{\eta_T^2(n-t)^{1-2\lambda}}{400}} \leq q^{-\frac{\eta^2n^{1-2\lambda}}{400}} =: \Delta \eperiod
    \end{align}
    Now, recalling that conditioned on the event $\cE_{\overline T}$ (which we recall is that the coordinates in $\overline T$ do not support a nonzero dual codeword) we have $\delta_{C,f} = \delta_{C_T,f_T}$, we get that 
    \begin{align*}
        \PROver{C}{\exists f \in \cF \text{ s.t. }|\delta_{C,f}| \geq \Delta} &\leq \sum_{\substack{T \subseteq [n] \\ |T| \geq (1-\gamma)n}} \PROver{C}{\exists f \in \cF_T \text{ s.t. }|\delta_{C,f}| \geq \Delta} \\
        &\leq \sum_{\substack{T \subseteq [n] \\ |T| \geq (1-\gamma)n}} \PROver{C}{\exists f \in \cF_T \text{ s.t. }|\delta_{C,f}| \geq \Delta\mid\cE_{\overline T}} + \sum_{\substack{T \subseteq [n] \\ |T| \geq (1-\gamma)n}}\PROver{C}{\neg \cE_{\overline T}}.
    \end{align*}
    We proceed to bound both of the above sums. 
    To bound the leftmost sum, observe that $q^{-\frac{\eta_T(n-t)}{200}} = q^{-\frac{\eta n}{200}}$, as $\eta_T = \frac{\eta}{1-\tau}$ and $n-t = (1-\tau)n$. Therefore, 
    we have that 
    \begin{align}
        \sum_{\substack{T \subseteq [n] \\ |T| \geq (1-\gamma)n}} \PROver{C}{\exists f \in \cF_T \text{ s.t. }|\delta_{C,f}| \geq \Delta\mid\cE_{\overline T}} &= \sum_{\substack{T \subseteq [n] \\ |T| \geq (1-\gamma)n}}\PROver{C_T^*}{\exists f_T^* \in \cF_T^* \text{ s.t. }|\delta_{C_T^*,f_T^*}| \geq \Delta}\nonumber\\
        &\leq \sum_{\substack{T \subseteq [n] \\ |T| \geq (1-\gamma)n}}q^{-\frac{\eta n}{200}} + |\cF_T|q^{-9(1-\gamma)^{1+\lambda}n^{1+\lambda}}\nonumber\\
        &\leq 2^{n}q^{-\frac{\eta n}{200}} + |\cF|q^{-9(1-\gamma)^{1+\lambda}n^{1+\lambda}} \nonumber\\
        &\leq 2^n q^{-\frac{\eta n}{200}} + q^{-8(1-\gamma)^{1+\lambda}n^{1+\lambda}} \ .
        \label{eq:UB-left-sum}
    \end{align}
    The first inequality follows from \cref{eq:UB-prob-fstar,eq:uniform-UB-Delta}.
    The second inequality uses the fact that there are at most $2^n$ subsets $T\subseteq[n]$ of size $|T|\geq (1-\gamma)n$.
    The last line uses
    \begin{equation}\label{eq:UB-size-family}
        |\cF| \leq 2^{qn} \leq 2^{n^{1+\lambda}} \leq q^{(1-\gamma)^{1+\lambda}n^{1+\lambda}}, 
    \end{equation}
    and the final inequality in \cref{eq:UB-size-family} uses the assumption $q \geq 2^{\frac{1}{(1-\gamma)^{1+\lambda}}}$. 

    As for the rightmost sum, we note that $x \in C^\perp$ if and only if $G^\top x=0$, where $G \in \F_q^{n \times Rn}$ is the uniformly random matrix generating $C$. Hence, we have
    \begin{align}
        \sum_{\substack{T \subseteq [n] \\ |T| \geq (1-\gamma)n}}\PROver{C}{\neg \cE_T} &= \sum_{\substack{T \subseteq [n] \\ |T| \geq (1-\gamma)n}}\PROver{C}{\exists x \in C^\perp\setminus\{0\} \text{ s.t. }\supp(x) \subseteq \overline T} \nonumber\\
        &\leq \sum_{\substack{T \subseteq [n] \\ |T| \geq (1-\gamma)n}}\sum_{\substack{x \in \F_q^n\setminus\{0\}\\ \supp(x) \subseteq \overline T}} \PROver{G}{G^\top x= 0}\nonumber\\
        &= \sum_{\substack{U \subseteq [n]: |U| \leq \gamma n}} (q^{|U|}-1)\,q^{-(\gamma+\eta)n}\nonumber\\
        &\leq q^{-(\gamma+\eta)n} \sum_{u=1}^{\floor{\gamma n}} \binom{n}{u}q^u \nonumber \\
        &\leq q^{-(\gamma+\eta)n} \cdot 2^n\cdot q^{\gamma n} \nonumber \\
        &=  q^{-n(\eta - 1/\log q)}\nonumber \\
        &\leq q^{-n\eta/2} \ . \label{eq:UB-right-sum}
    \end{align}
    The last inequality follows from our assumption $q\geq 2^{400/\eta}$, which in particular guarantees that $\eta - 1/\log q \geq \eta/2$.
    Thus, putting \cref{eq:UB-left-sum,eq:UB-right-sum} together, we conclude that
    \[
        |\delta_{f,C}| \leq \Delta = q^{-\frac{\eta^2n^{1-2\lambda}}{400}}
    \]
    with probability at least
    \begin{equation}\label{eq:prob-LB-restate}
        1 - 2^{n}q^{-\frac{\eta n}{200}} - q^{-8(1-\gamma)^{1+\lambda}n^{1+\lambda}} - q^{-\frac{n\eta}2} \eperiod
    \end{equation}
    Lastly, we observe that the terms $q^{-8(1-\gamma)^{1+\lambda}n^{1+\lambda}}$ and $q^{-\frac{n\eta}2}$ are both $o\left(2^{n}q^{-\frac{\eta n}{200}}\right)$. 
    Therefore, for sufficiently large $n$ the probability lower bound in \cref{eq:prob-LB-restate} is at least
    \[
        1 - 2^{n+1}q^{-\frac{\eta n}{200}} \eperiod
    \]
    Additionally, since $q \geq 2^{400/\eta}$ by hypothesis, we get that $2^{n+1}q^{-\frac{\eta n}{200}}\leq 2 q^{-\frac{\eta n}{400}}$, and so the probability lower bound in \cref{eq:prob-LB-restate} is at least
    \begin{equation*}
        1 - 2 q^{-\frac{\eta n}{400}}.
    \end{equation*}
    This justifies the second lower bound on the failure probability given in \cref{eq:failure-prob-LR-large-rectangles} in the lemma statement.
\end{proof}

\subsubsection{The Main Result for Massey Ramp Secret Sharing}

The following theorem is an easy consequence of \cref{lem:massey-prop,lem:LR-discrepancy-reduced,lem:LR-large-rectangles}. It shows that the Massey ramp secret sharing scheme induced by a random linear code of any positive rate will be locally leakage-resilient with high probability.

We first state two well-known facts that will be useful in our analysis.
The first one is easily implied by the proof of \cite[Section~4.2.2, Theorem~4.2]{GRS}.

\begin{lemma} \label{lem:prob-bound-to-GV}
    Suppose that $C \subseteq \F_q^n$ is sampled from a distribution over linear codes of rate $R$ such that for all $x \in \F_q^n\setminus\{0\}$ we have $\Pr[x \in C] \leq q^{k-n}$.
    Let $\theta>0$ be arbitrary (possibly depending on $n$ and $q$).
    Then, $C$ has minimum distance at least $h_q^{-1}(1-R-\theta)n$ with probability at least $1-q^{-\theta n}$.
\end{lemma}

Second, we recall the following simple bound on the $q$-ary entropy function, for which a proof can be found in~\cite[proof of Proposition 3.3.4]{GRS}. We will use this bound to control the \emph{inverse} of $h_q$, where we recall that the $q$-ary entropy function $h_q$ has a well-defined continuous inverse on the interval $[0,1-1/q]$.

\begin{lemma}\label{lem:bound-hq}
    For any positive integer $q$ and real number $\rho\in(0,1-1/q]$ we have
    \begin{equation*}
        h_q(\rho)\leq \rho+\frac{1}{\log q} \eperiod
    \end{equation*}
    Therefore, $h_q^{-1}(\rho+\frac{1}{\log q})\geq \rho$ whenever $0<\rho+\frac{1}{\log q}\leq 1-\frac{1}{q}$.
\end{lemma}

\begin{theorem} \label{thm:rampss-full}
    Fix arbitrary $\gamma,\eta>0$ satisfying $\gamma+\eta<1$ and $\lambda\in(0,1/5)$. 
    Given $n \in \N$, let $q$ be a prime satisfying 
    \[
        \max\left(2^{\frac{1}{(1-\gamma)^{1+\lambda}}},2^{2/\gamma},2^{400/\eta}\right)\leq q \leq \frac{(n(1-\gamma))^\lambda}{\log(n(1-\gamma))} \eperiod
    \]
    Furthermore, let $\theta \in \left(\frac{\log_qn}n,\min(\gamma,1-(\gamma+\eta))-\frac{1}{n+1}\right)$. Then, there exists $n_0 := n_0(\gamma,\eta,\lambda)$ such that, for all $n \geq n_0$, the following holds. 
    Let $C \subseteq \F_q^n$ be a random linear code of dimension $Rn$ with $R = \gamma+\eta$, let $v \in \F_q^n$ be sampled uniformly (and independently of $G$), and let $C^+ \subseteq \F_q^{n+1}$ be the code generated by 
    \[
        G^+ = \begin{bmatrix}
        1 & 0 \\ v & G
        \end{bmatrix} \in \F_q^{(n+1)\times(k+1)} \eperiod
    \]
    Then, the Massey ramp secret sharing scheme associated with $C^+$ is an $(n,\trec,\tpriv,\eps)$-locally leakage-resilient ramp secret sharing scheme with
    \begin{align*}
         &\trec \leq (1-h_q^{-1}(1-R-\theta))n+3\leq \left(R+\theta+\frac{1}{\log q}\right)n+2\ ,\\
         &\tpriv \geq h_q^{-1}(R-\theta)n-2 \geq \left(R-\theta-\frac{1}{\log q}\right)n - 2\ ,
    \end{align*}
    and 
    \begin{equation*}
        \eps \leq 2^{n+1}q^{-\gamma n} + 2q^{- \frac{\eta^2 n^{1-2\lambda}}{800}} \leq 2q^{-\frac{\gamma n}2} + 2q^{- \frac{\eta^2 n^{1-2\lambda}}{800}} 
    \end{equation*}
    with probability at least
    \begin{equation*}
        1- q^{-(1-R-\theta)n} -2q^{-\theta(n+1)} -2q^{-\frac{\eta n}{400}}
    \end{equation*}
    over the sampling of $C^+$. 
\end{theorem}

\begin{proof}
    Note that $C^+\subseteq\F_q^{n+1}$ is a linear code of design rate $R'$ satisfying $R\leq R'\leq R+\frac{1}{n+1}$. We first note that $C^+$ has rate $R'$ with high probability. Indeed, this is the probability that $G^+$ has rank $k+1$, which is the same as the probability that $G$ has rank $k$, which is
    \begin{equation}\label{eq:prob-full-rank}
        \prod_{i=0}^{k-1}\left(1-\frac{1}{q^{n-i}}\right) \geq 1-kq^{k-n} \geq 1-q^{-(1-R)n +\log_q n}\geq 1-q^{-(1-R-\theta)n}
    \end{equation}
    since $\theta> \frac{\log_q n}{n}$ by hypothesis. 
    
    We next observe that if $x \in \F_q^{n+1}\setminus \{0\}$, then
    \begin{align}
        \PROver{v,G}{x \in C^+} = \PROver{v,G}{\exists m \in \F_q^{k+1}\setminus\{0\}\text{ s.t. }G^+m=x} \leq \sum_{m \in \F_q^{k+1}\setminus\{0\}}\PROver{v,G}{G^+m=x} \eperiod \label{eq:prob-sum-over-message}
    \end{align}
    Now, if one writes $m = (m_0,\tilde m)$ with $m_0\in\F_q$ and $\tilde m \in \F_q^k$ and $x = (x_0,\tilde x)$ with $x_0\in\F_q$ and $\tilde x \in \F_q^n$, we get that $G^+m = x$ if and only if $m_0 v + G\tilde m=x$, which holds if and only if $m_0=x_0$ and $G\tilde m = \tilde x-x_0v$. 
    If $\tilde m=0$, then we must have $m_0=x_0\neq 0$ (as otherwise $m=0$), and so $G\tilde m=\tilde x-x_0v$ if and only if $v=\tilde x$, which occurs with probability $q^{-n}$.
    Otherwise, if $\tilde m\neq 0$, we get that $G\tilde m$ is uniformly random over $\F_q^n$, so again $\Pr[G\tilde m = \tilde x-x_0v]=q^{-n}$. Consequently,
    the right-hand side of \cref{eq:prob-sum-over-message} is 
    \[
        \sum_{\substack{m \in \F_q^{k+1}\setminus\{0\}\\m_0=x_0}} q^{-n} \leq q^k \cdot q^{-n} = q^{(k+1)-(n+1)} \eperiod
    \]
    We can then apply \Cref{lem:prob-bound-to-GV} to conclude that $C^+$ has distance $d\geq h_q^{-1}(1-R'-\theta)(n+1) \ge h_q^{-1}(1-R-\tfrac1{n+1}-\theta)(n+1)$ with probability at least 
    \begin{equation}\label{eq:prob-good-dist}
        1-q^{-\theta (n+1)}.
    \end{equation}
    
    As for the dual of $C^+$, we have $x \in (C^+)^\perp$ if and only if $x^\top G^+=0$. Writing again $x = (x_0,\tilde x)$ with $x_0 \in \F_q$ and $\tilde x \in \F_q^n$, this is equivalent to $\langle v,\tilde x\rangle=-x_0$ and $\tilde x^\top G=0$. If $\tilde x = 0$, then $x_0\neq 0$, so the probability that $\langle v,\tilde x \rangle=-x_0$ is $0$ (we have $\Pr[\langle v,\tilde x\rangle=0]=1$). Else, $(\langle v,\tilde x\rangle,\tilde x^TG)$ is distributed uniformly over $\F_q^{k+1}$ (over the randomness of $v$ and $G$), so it equals $(-x_0,0)$ with probability $q^{-(k+1)}$. Hence, 
    \[
        \PROver{v,G}{x \in (C^+)^\perp} \leq q^{-(k+1)} = \frac{q^{(n+1)-(k+1)}}{q^{n+1}} \ ,
    \]
    so we may again apply \Cref{lem:prob-bound-to-GV} to conclude that $(C^+)^\perp$ has distance $d^\perp\geq h_q^{-1}(R'-\theta)(n+1) \geq h_q^{-1}(R-\theta)(n+1)$ with probability at least
    \begin{equation}\label{eq:prob-dual-dist}
        1-q^{-\theta (n+1)}.
    \end{equation}
    
    Assuming these three favourable events hold, we conclude via \cref{lem:massey-prop} that the associated Massey ramp secret sharing scheme will have reconstruction threshold
    \begin{equation*}
        \trec = n-d+2 \leq \left(1-h_q^{-1}(1-R-\tfrac1{n+1}-\theta)\right)n + 2
        \leq \left(R+\theta+\tfrac{1}{\log q}\right)n+3
    \end{equation*}
    and privacy threshold
    \begin{equation*}
        \tpriv = d^\perp - 2 \geq h_q^{-1}(R-\theta)n - 2\geq \left(R-\theta-\tfrac{1}{\log q}\right)n - 2.
    \end{equation*}
    The rightmost bounds on $\trec$ and $\tpriv$ above follow from \cref{lem:bound-hq}, noting that $\frac{1}{q}\leq R\leq 1-\frac{1}{q}$ (with the upper bound holding by the lower bound on $q$ from the theorem statement, which implies that $(1-\frac{1}{q})\log q>2\log(\frac{2}{1-2\gamma})> R\log q$) and that $\theta<\min(R,1-R)-\frac{1}{n+1}$ by hypothesis. 

    Now, we consider the leakage-resilience property. 
    Let $g=(g_1,\dots,g_n)$ be a local leakage function, and let $\cS_{g}$ denote the set of its induced rectangles (i.e., the set of all rectangles $S=S_1\times\cdots\times S_n$ such that for each $i\in[n]$ we have $S_i=g_i^{-1}(b_i)$ for some $b_i\in\bits$). 
    By \Cref{lem:LR-discrepancy-reduced}, it suffices to bound
    \begin{align*}
        \SD(g(C),g(U_n)) &= \max_{\cA\subseteq \cS_g} \left(\PROver{x \sim \F_q^n}{\exists S\in\cA: x \in S} - \PROver{x \sim C}{\exists S\in\cA: x \in S}\right)\\
        &=\max_{\cA\subseteq \cS_g} \left(\sum_{S\in\cA}\left(\PROver{x \sim \F_q^n}{x \in S} - \PROver{x \sim C}{x \in S}\right)\right).
    \end{align*}
    By definition of statistical distance, we are then guaranteed the existence of a subset $\cA_g \subseteq \cS_g$ such that 
    \begin{equation}\label{eq:SD-induced}
        \SD(g(C),g(U_n)) = \sum_{S \in \cA_g}\left(\PROver{x \sim \F_q^n}{x \in S} - \PROver{x \sim C}{x \in S}\right) \eperiod
    \end{equation}
    Partition $\cA_g = \cA_{g,s} \sqcup \cA_{g,l}$, where $\A_{g,s}$ consists of all $\gamma$-small rectangles in $\cA_g$, and $\cA_{g,l}$ consists of the large rectangles.
    We can then rewrite \cref{eq:SD-induced} as
    \[
       \SD(g(C),g(U_n)) = \sum_{S \in \cA_{g,s}}\left(\PROver{x \sim \F_q^n}{x \in S} - \PROver{x\sim C}{x \in S}\right) + \sum_{S \in \cA_{g,l}}\left(\PROver{x \sim \F_q^n}{x \in S} - \PROver{x \sim C}{x \in S}\right) \eperiod
    \]
    The sum over $S\in\cA_{g,s}$ can be bounded as 
    \begin{align}
        \sum_{S \in \cA_{g,s}}\left(\PROver{x \sim \F_q^n}{x \in S} - \PROver{x\sim C}{x \in S}\right) \leq \sum_{S \in \cA_{g,s}}\PROver{x \sim \F_q^n}{x \in S} = \sum_{S \in \cA_{g,s}}\frac{|S|
        }{q^n} \leq |\cA_{g,s}|\, q^{-\gamma n} \leq 2^n \, q^{-\gamma n}\ , \label{eq:bound-on-gamma-small}
    \end{align}
    where the second inequality follows from the fact that $S$ being $\gamma$-small means that $|S|<q^{(1-\gamma)n}$.

    As for the sum over $S\in\cA_{g,l}$, let $f_{S}$ denote the normalized indicator function for a combinatorial rectangle $S$. 
    Recalling that $\delta_{C,f_{S}} = \frac{\PROver{x \sim C}{x \in S}}{\PROver{x \sim \F_q^n}{x \in S}}-1$, we have 
    \begin{equation*}
        \PROver{x \sim \F_q^n}{x \in S} - \PROver{x \sim C}{x \in S} = \PROver{x \sim \F_q^n}{x \in S}(-\delta_{C,f_{S}}) \leq \PROver{x \sim \F_q^n}{x \in S}\cdot|\delta_{C,f_{S}}|\eperiod
    \end{equation*}
    Therefore,
    \begin{equation}\label{eq:bound-on-gamma-large}
        \sum_{S \in \cA_{g,l}}\left(\PROver{x \sim \F_q^n}{x \in S} - \PROver{x \sim C}{x \in S} \right) \leq \sum_{S \in \cA_{g,l}}\PROver{x \sim \F_q^n}{x \in S}\cdot |\delta_{C,f_{S}}| \leq \max_{S \in \cA_{g,l}}|\delta_{C,f_{S}}| \eperiod
    \end{equation}
    By \Cref{lem:LR-large-rectangles} (using our hypotheses on $\gamma$, $\eta$, $\lambda$, and $q$), for sufficiently large $n$ it holds that
    \begin{equation}\label{eq:bound-on-max-delta}
        \max_{S \in \cA_{g,l}}|\delta_{C,f_{S}}| \leq q^{-\frac{\eta^2n^{1-2\lambda}}{400}}
    \end{equation}
    with probability at least
    \begin{equation}\label{eq:prob-delta}
        1-2q^{-\eta n/400}
    \end{equation}
    over the sampling of $C$.

    Combining \cref{eq:bound-on-gamma-small,eq:bound-on-gamma-large,eq:bound-on-max-delta} and taking a union bound over the probability lower bounds for the four events in \cref{eq:prob-good-dist,eq:prob-full-rank,eq:prob-dual-dist,eq:prob-delta}, we obtain $\eps$-leakage resilience for
    \[
        \eps \leq 2^{n+1}q^{-\gamma n} + 2q^{-\frac{\eta^2n^{1-2\lambda}}{800}}
    \]
    with probability at least
    \begin{equation*}
        1 - q^{-(1-R-\theta)n}-2q^{-\theta(n+1)}-2q^{-\eta n/400},
    \end{equation*}
    as claimed. 
    Lastly, using our hypothesis that $q \geq 2^{2/\gamma}$ we get that $2^n \leq q^{\frac{\gamma n}2}$, and so $2^{n+1}q^{-\gamma n} \leq 2q^{-\frac{\gamma n}2}$. That is, we can further bound
    \[
        \eps \leq 2q^{-\frac{\gamma n}2} + 2q^{- \frac{\eta^2 n^{1-2\lambda}}{800}}\ ,
    \]
    as desired.
\end{proof}

In our study of ``low-rate'' leakage-resilient ramp secret sharing, we are mainly interested in how large we can take $q$ and how small we can make the gap between the privacy and reconstruction thresholds.
Since \cref{thm:rampss-full} has several parameters that make it hard to parse answers to these questions, we now present an example instantiation of this theorem that showcases these answers more transparently and which includes \cref{thm:rampss-short} as a special case.
Roughly speaking, the instantiation below shows that we can take $q$ as large as $n^{1/5-\nu}$ for an arbitrary constant $\nu>0$ while simultaneously achieving reconstruction threshold $\approx Rn$ for an arbitrary constant $R>0$ and threshold gap $\approx 2n/\log q$.
Note that the threshold gap can be made $O(n/\log n)$ when $q$ is polynomial in $n$.

\begin{corollary}\label{cor:rampss-simpler}
    Fix arbitrary constants $R\in(0,1)$ and $\nu\in(0,1/5)$.
    Given $n\in\N$, let $q$ be a prime satisfying
    \begin{equation*}
        2^{800/R} \leq q \leq n^{1/5-\nu}.
    \end{equation*}
    Then, for all sufficiently large $n$ (depending only on $R$ and $\nu$) there is an $(n,\trec,\tpriv,\eps)$-locally leakage-resilient ramp secret sharing scheme over $\F_q$ with
    \begin{equation*}
       Rn - \frac{(1+\nu)n}{\log q} \leq \tpriv < \trec \leq Rn + \frac{(1+\nu)n}{\log q}
    \end{equation*}
    and 
    \begin{equation*}
        \eps \leq 4q^{- \frac{R^2 n^{3/5}}{3200}}.
    \end{equation*}
\end{corollary}
\begin{proof}
    Fix $R\in(0,1)$ and $\nu\in(0,1/5)$.
    We instantiate \cref{thm:rampss-full} with $\gamma=\eta=R/2$, $\lambda=\frac{1}{5}-\frac{\nu}{2}$, and $\theta = \frac{2\log_q n}{n}$.
    Note that for sufficiently large $n$ we get that 
    \begin{equation*}
        \frac{\log_q n}{n} <\theta = \frac{2\log_q n}{n} < \min\left(\gamma,1-(\gamma+\eta)-\frac{1}{n+1}\right) = \min\left(\frac{R}{2},1-R-\frac{1}{n+1}\right),
    \end{equation*}
    as required for the application of \cref{thm:rampss-full}.
    Regarding the constraints on $q$ for the application of this theorem, note that under our choices of $\gamma$, $\eta$, and $\lambda$ above together with our hypotheses on $q$ from the corollary statement we get that
    \begin{equation*}
        \max\left(2^{\frac{1}{(1-\gamma)^{1+\lambda}}},2^{2/\gamma},2^{400/\eta}\right) = 2^{400/\eta}=2^{R/800}\leq q \leq n^{1/5-\nu} <  \frac{(n(1-\gamma))^\lambda}{\log(n(1-\gamma))}
    \end{equation*}
    for all sufficiently large $n$, since $\lambda=\frac{1}{5}-\frac{\nu}{2}>\frac{1}{5}-\nu$.

    Then, \cref{thm:rampss-full} guarantees that the Massey ramp secret sharing scheme associated with the random linear code $C^+$ will be an $(n,\trec,\tpriv,\eps)$-locally leakage-resilient ramp secret sharing scheme for
    \begin{equation*}
        \left(R-\theta-\frac{1}{\log q}\right)n-2\leq \tpriv<\trec \leq \left(R+\theta+\frac{1}{\log q}\right)n+2
    \end{equation*}
    and
    \begin{equation*}
        \eps \leq 2 q^{-\frac{\gamma n}{2}}+2q^{-\frac{\eta^2 n^{1-2\lambda}}{800}} \leq 2 q^{-\frac{Rn}{4}} + 2q^{-\frac{R^2 n^{3/5}}{3200}} \leq 4q^{-\frac{R^2 n^{3/5}}{3200}},
    \end{equation*}
    recalling that $\gamma=\eta=R/2$ and $\lambda<1/5$, with probability at least
    \begin{equation*}
        1-q^{-(1-R-\theta)n} - 2q^{-\theta(n+1)}-2q^{-\frac{\eta n}{400}} \geq 1-q^{\frac{(1-R)n}{2}}-\frac{1}{n^2}-2q^{-\frac{Rn}{800}}>0
    \end{equation*}
    over the sampling of $C^+$, where the inequalities hold for all sufficiently large $n$, recalling that $\theta=\frac{2\log_q n}{n}$.

    It remains to simplify the bounds on $\tpriv$ and $\trec$ above for sufficiently large $n$.
    It suffices to note that
    \begin{equation*}
        \left(\theta+\frac{1}{\log q}\right)n+2 = 2\log_q n+\frac{n}{\log q} +2 \leq \frac{n}{\log q}\left(1 + \frac{4\log n}{n}\right) \leq \frac{(1+\nu)n}{\log q},
    \end{equation*}
    where the first inequality uses the fact that $q\leq n$ and the last inequality holds for all sufficiently large $n$, since $\nu$ is a fixed constant. 
\end{proof}

The simplified \cref{thm:rampss-short} follows directly from \cref{cor:rampss-simpler} by setting $\nu=0.01$.

\addtocontents{toc}{\protect\setcounter{tocdepth}{0}}
\section*{Acknowledgements}
\addtocontents{toc}{\protect\setcounter{tocdepth}{2}}

Dean Doron is supported in part by 
Israel Science Foundation grant \#857/25 and by NSF-BSF grant \#2022644.
Tal Leonov is supported by Israel Science Foundation grant \#3450/24. Jonathan Mosheiff is supported by Israel Science Foundation grant \#3450/24 and an Alon Fellowship. 
Henrique Navas is funded by FCT - Fundação para a Ciência e a Tecnologia, I.P., under grant 2025.05868.BD and by national funds through FCT, I.P., and, when eligible, co-funded by EU funds under project/support UID/50008/2025 – Instituto de Telecomunicações, with DOI \href{https://doi.org/10.54499/UID/50008/2025}{10.54499/UID/50008/2025}. 
Nicolas Resch is supported in part by an NWO (Dutch Research Council) Veni grant VI.Veni.222.347. 
João Ribeiro 
is funded by the European Union (LESYNCH, 101218842).
Views and opinions expressed are however those of the authors only and do not necessarily reflect those of the European Union or the European Research Council Executive Agency. Neither the European Union nor the granting authority can be held responsible for them.
João Ribeiro is also funded by national funds through FCT – Fundação para a Ciência e a Tecnologia, I.P., and, when eligible, co-funded by EU funds under project/support UID/50008/2025 – Instituto de Telecomunicações, with DOI \href{https://doi.org/10.54499/UID/50008/2025}{10.54499/UID/50008/2025}.

\bibliographystyle{alpha}
\bibliography{refs}

\addtocontents{toc}{\protect\setcounter{tocdepth}{1}}

\appendix

\section{Deferred Proofs}

\subsection{Deferred Proofs from \Cref{sec:main-results}}\label{app:sec1}

\begin{proof}[Proof of \Cref{thm:list-rec}]
    Observe that the second event in the statement implies the first one. Indeed, let $S = S_1\times\dots\times S_n$ be a sequence of sets $S_i \subseteq \F_q$, each of size at most $\ell$. Then there are sets $S'_1,\dots, S'_n$, such that $S_i \subseteq S'_i$ and $\inabs{S'_i} = \ell$ for all $i$. By the second event,
    $$\inabs{S\cap C} \le \inabs{S'\cap C} \le \left(1+q^{-c\gamma n/(q^2\log(q/\gamma))}\right)\frac{|C|\ell^n}{q^n}\le \left(1+q^{-c\gamma n/(q^2\log(q/\gamma))}\right)q^{\gamma n}\ecomma$$
    implying the first event. It thus suffices to prove that the second event holds with probability at least $$1-2q^{-c\gamma n/(q^2\log(q/\gamma))}\eperiod$$

    Let $\cF$ be the family of normalized indicators
    $f_S=(q^n/|S|)\mathbf{1}_S$ over rectangles defined by $S_1,\dots, S_n$ each of size exactly $\ell$. Then
    \[
        |\cF|=\binom{q}{\ell}^n\le 2^{qn}\le q^{qn},
    \]
    and $\|f_S\|_\infty=(q/\ell)^n$, so the capacity parameter in
    \cref{cor:SimplerMain} is $\beta=1-\log_q\ell$.

    By \cref{lem:zero-error-con}, every $f_S\in\cF$ is
    $\alpha$-Fourier-concentrated with $\alpha=\cos(\pi/q)$. Therefore the
    parameter $\lambda$ from \cref{cor:SimplerMain} satisfies
    \[
        \lambda
        =
        \min\left\{
            1,\frac{\log(1/\cos(\pi/q))}{2\log(100q/\gamma)}
        \right\}
        \ge
        \frac{c_1}{q^2\log(q/\gamma)}
    \]
    for a universal constant $c_1>0$, using \Cref{lem:lower-bound-on-alpha} to lower bound $\log(1/\cos(\pi/q)) \geq \Omega(1/q^2)$.

    Choose the theorem constant $c>0$ so that
    $c\le \min\{c_1/4,1/200\}$, and set
    \[
        \theta:=\frac{c\gamma}{q^2\log(q/\gamma)}.
    \]
    Then $\theta<\gamma\lambda$, and
    \[
        2(\gamma\lambda-\theta)
        \ge
        \frac{c\gamma}{q^2\log(q/\gamma)}.
    \]
    Also, since $\gamma\le \frac12(1-\log_q\ell)=\beta/2$, the hypothesis
    $\gamma\le 50\beta/51$ of \cref{cor:SimplerMain} is satisfied.

    Taking $c'$ large enough, the assumed
    lower bound
    \[
        n\ge c'\frac{q^5\log^2(q/\gamma)}{\gamma^2}
    \]
    implies the lower bound $n\ge n_0(q,\gamma,\alpha,\theta)$ required in
    \cref{cor:SimplerMain}. Hence \cref{cor:SimplerMain}, applied with
    $\eta=\gamma$, gives
    \[
        |\delta_{C,f_S}|
        \le
        q^{-c\gamma n/(q^2\log(q/\gamma))}
    \]
    for every $f_S\in\cF$, except with failure probability at most
    \[
        q^{-\gamma n/200}+|\cF|q^{n-\theta^2n^2/20}.
    \]

    The first term is at most
    $q^{-c\gamma n/(q^2\log(q/\gamma))}$ by the choice of $c$. For the second,
    using $|\cF|\le q^{qn}$ and increasing $c'$ if necessary, the lower bound
    on $n$ gives $\theta^2n/20\ge q+1+\theta$. Therefore
    \begin{equation}\label{eq:to-erasures1}
        |\cF|q^{n-\theta^2n^2/20}
        \le
        q^{qn+n-\theta^2n^2/20}
        \le
        q^{-\theta n}
        =
        q^{-c\gamma n/(q^2\log(q/\gamma))}.
    \end{equation}
    Thus the desired deviation event holds with probability at least
    \[
        1-2q^{-c\gamma n/(q^2\log(q/\gamma))}.
    \]

    On this event, every rectangle $S$ with side length exactly $\ell$
    satisfies
    \[
        |C\cap S|
        =
        \left(1\pm q^{-c\gamma n/(q^2\log(q/\gamma))}\right)
        \frac{|C|\ell^n}{q^n}\eperiod
    \]
\end{proof}

\begin{proof}[Proof of \Cref{thm:list-rec-erasures}]
The proof is nearly identical to the proof of \cref{thm:list-rec}, so we only highlight the few differences. Letting $S$ be any sequence of sets that corresponds to list recovery with $\zeta$-fraction of erasures, we have that $|S| = q^{\zeta n}\ell^{(1-\zeta)n}$ (we can assume that exactly $\zeta$-fraction of the lists satisfy $S_i=\F_q$ and the rest satisfy $|S_i| = \ell$). Keeping with the same notation as above, we have that $\norm{f_S}_{\infty} = (q/\ell)^{(1-\zeta)n}$, so indeed $\beta = (1-\zeta)(1-\log_{q}\ell)$. Our choice of $\gamma$ still satisfies $R < 1$ and $\gamma \le \beta/2$, and note that the family $\mathcal{F}$ is still $\alpha$-concentrated with the same $\alpha$.

However, $|\mathcal{F}|$ has changed, so we need to verify that the bound on the failure probability still holds. A straightforward bound gives
\[
    |\mathcal{F}| = \binom{n}{\zeta n}\binom{q}{\ell}^{(1-\zeta)n} \le \sum_{i=0}^{n}\binom{n}{i}\binom{q}{\ell}^{n-i} = \left( \binom{q}{\ell}+1 \right)^n \le 2^{qn} \le q^{qn},
\]
so \Cref{eq:to-erasures1} still holds. The rest of the proof is unchanged.
\end{proof}

\subsection{Deferred Proofs from \Cref{sec:leakage-resilience}}\label{app:sec2}

\begin{proof}[Proof of \cref{lem:massey-prop}]
    To see the claim about reconstruction, note that since $C^+$ has minimum distance $d$, it corrects any $d-1$ erasures.
    Therefore, $(n+1)-(d-1)=n-d+2$ shares are sufficient to recover $c^+=Gx$, and so we recover the secret $s=c^+_0$.
    To see the claim about privacy, note that since $C^+$ has dual distance $d^\perp$, then the uniform distribution over $C^+$ is $(d^\perp-1)$-wise independent.
    This means that if $c^+\sim C^+$, then for any indices $1\leq i_1<i_2<\cdots <i_{d^\perp-2}\leq n$ we have that $(c^+_0,c^+_{i_1},\dots,c^+_{i_{d^\perp-2}})$ is uniformly distributed over $\F_q^{d^\perp-1}$.
    Therefore, the shares associated with parties $i_1,\dots,i_{d^\perp-2}$ are independent of the secret $c_0^+$.
\end{proof}

\begin{proof}[Proof of \cref{lem:LR-discrepancy-reduced}]
    Fix an arbitrary tuple of leakage functions $g=(g_1,\dots,g_n)$.
    Fix also two secrets $s,s'\in\F_q$, with corresponding tuples of shares $S=(S_1,\dots,S_n)$ and $S'=(S'_1,\dots,S'_n)$, respectively. 
    Recall that the shares of $s$ (resp.\ $s'$) correspond to a uniformly random sample from $s\cdot v+C$ (resp.\ $s'\cdot v+C$).
    Let $X$ and $X'$ be independent and uniformly distributed over $C$.
    Define also $g_s(z) = g(s\cdot v+z)$ for any $s\in\F_q$. By the assumption that $\cG$ is closed under input-translations, we have $g_s \in \cG$. 
    Then, 
    \begin{align*}
        \SD(g(S),g(S')) &\leq  \SD(g(s\cdot v +X),g(U_n)) + \SD(g(s'\cdot v +X'),g(U_n))\\
        & = \SD(g_s(X),g(U_n)) +  \SD(g_{s'}(X'),g(U_n))\\
        &= \SD(g_s(X),g_s(U_n)) +  \SD(g_{s'}(X'),g_{s'}(U_n))\\
        &= \SD(g_s(C),g_s(U_n)) +  \SD(g_{s'}(C),g_{s'}(U_n))\\
        &\leq 2\eps.
    \end{align*}
    The first inequality follows by the triangle inequality, the second equality uses the fact that $g(U_n)$ and $g_s(U_n)$ are identically distributed for any $s\in\F_q$, and the last inequality uses the hypothesis from the lemma statement applied to $g_s$ and $g_{s'}$.
\end{proof}

\begin{proof}[Proof of \Cref{cor:LR-simple-large-rectangles}]
    We apply \cref{thm:mainGeneral} and refer to notation in the context of that theorem's statement.
    
    First, by \cref{lem:zero-error-con}, $\cF$ is a family of normalized $\alpha$-concentrated functions for $\alpha\leq \cos(\pi/q)$.
    Since $\log(1/\alpha)\geq \log(1/\cos(\pi/q))\geq \frac{\pi^2}{8q^2\ln 2}$ by \Cref{lem:lower-bound-on-alpha} for all $q\geq 3$, we get that
    \begin{align}
        d &= \max\inset{0, \left\lceil \log_2 \left( \frac{\log_2\inparen{(q-1)/(q^{\eta/100}-1)}}{\log_2(1/\alpha)} \right) \right\rceil}\nonumber\\
        &\leq \max\inset{0, \left\lceil \log_2 \left(\frac{8q^2\ln 2}{\pi^2}\cdot \log_2\left(\frac{q-1}{q^{\eta/100}-1}\right) \right) \right\rceil}\nonumber\\
        &\leq \left\lceil \log_2 \left(\frac{8q^2\ln 2}{\pi^2}\cdot \log_2\left(\frac{q-1}{2^{\eta/100}-1}\right) \right) \right\rceil\nonumber\\
        &\leq 1+\log_2 \left(\frac{8q^2}{\pi^2}\cdot \log_2\left(\frac{q-1}{2^{\eta/100}-1}\right) \right).
        \label{eq:UB-d}
    \end{align}
    Furthermore, observe that $\min(R-\frac{\eta}{50},2\eta) = \min(\beta + \frac{49\eta}{50},2\eta) \geq \frac{49\eta}{50}$.
    Choose $t=10n^{\frac{1+\lambda}{2}}$. 
    Then, using the observations above, for all sufficiently large $n$ we have that
    \begin{align*}
        \eps &= 2^{-d} \min\left(R-\frac{\eta}{50},2\eta\right) - \frac{2t}{n} \\
        &\geq 2^{-d}\cdot \frac{49\eta}{50} -\frac{2t}{n}\\
        &\geq \frac{\pi^2 \eta}{32q^2 \log\left(\frac{q-1}{2^{\eta/100}-1}\right)}
        - 20 n^{-\frac{1-\lambda}{2}}\\
        &\geq \frac{\eta \cdot n^{-2\lambda}}{\log\left(\frac{1}{2^{\eta/100}-1}\right)}- 20 n^{-\frac{1-\lambda}{2}}\\
        &\geq (2^{\eta/100}-1) \eta\cdot n^{-2\lambda}-20 n^{-\frac{1-\lambda}{2}}\\
        &\geq \frac{\eta^2\cdot  n^{-2\lambda}}{100\log_2 e} -20 n^{-\frac{1-\lambda}{2}}.
    \end{align*}
    The third inequality uses the fact that $q \leq \frac{n^\lambda}{\log n}$, and holds for large enough $n$. The last inequality uses the fact that $2^x-1\geq \frac{x}{\log_2 e}$ for all $x>0$.
    Since $\lambda<1/5$ we get that $n^{-\frac{1-\lambda}{2}}=o(n^{-2\lambda})$, and so for all sufficiently large $n$ we have
    \begin{equation*}
        \eps \geq \frac{\eta^2 \cdot n^{-2\lambda}}{200\log_2 e} \geq \frac{\eta^2 \cdot n^{-2\lambda}}{400}.
    \end{equation*}

    Concluding, from \cref{thm:mainGeneral} we get that for sufficiently large $n$ it holds that
    \begin{equation*}
        |\delta_{C,f}|\le q^{-\frac{\eta^2 \cdot n^{1-2\lambda}}{400}}
    \end{equation*}
    for every $f\in \cF$ with probability at least
    \begin{align*}
        1 - n(Rn+1)q^{-\frac{\eta n}{100}} - |\cF| \cdot d \cdot q^{n - 10 n^{1+\lambda}}
    \end{align*}
    over the choice of $C$.
    It remains to argue the desired lower bound on the probability. 
    Using \cref{eq:UB-d}, for large enough $n$ we get that 
    \begin{equation*}
        d\leq 1+\log_2 \left(\frac{8q^2}{\pi^2}\cdot \log_2\left(\frac{q-1}{2^{\eta/100}-1}\right) \right) 
        \leq 1+\log_2 \left(\frac{8q^2}{\pi^2}\cdot \left(\log_2 q + \frac{100}{\eta}\right)\right) \leq \log_2 n + \log_2\left(\frac{100}{\eta}\right),
    \end{equation*}
    from the fact that $\log_2\left(\frac{1}{2^x-1}\right)\leq \frac{1}{x}$ for all $x>0$ and that $q\leq \frac{n^\lambda}{\log n}$ for $\lambda<1/5$.
    Combining this observation with the fact that $\eta$ is constant, we conclude that the probability is at least
    \begin{align*}
        1 - n(Rn+1)q^{-\frac{\eta n}{100}} - |\cF| \cdot d \cdot q^{n - 10 n^{1+\lambda}} \geq 1 - q^{-\frac{\eta n}{200}} -  |\cF|\cdot q^{-9n^{1+\lambda}}
    \end{align*}
    for all sufficiently large $n$. 
\end{proof}

\end{document}